\documentclass{article}

\usepackage{amsmath,amsthm,amssymb,bm,color,graphicx,mathrsfs,mathtools}
\usepackage{amscd,verbatim}
\usepackage{bbding,wasysym,pifont}
\usepackage[breaklinks=true,colorlinks=true,linkcolor=blue,urlcolor=blue,citecolor=blue]{hyperref}
\usepackage{graphicx}
\usepackage{subfigure}
\usepackage{comment}
\usepackage{lipsum}
\usepackage[arrow,matrix,curve]{xy}
\usepackage{authblk}
\usepackage{tikz-cd}
\usepackage[numbers]{natbib}

\usepackage{tikz-cd}
\usetikzlibrary{decorations.markings}
\tikzset{->-/.style={decoration={
  markings,
  mark=at position #1 with {\arrow{>}}},postaction={decorate}}}
\tikzset{->-/.default=0.5}

\newcommand\wt{\widetilde}
\newcommand\wh{\widehat}

\newcommand{\CC}{{\mathbb C}}

\newcommand{\RR}{{\mathbb R}}
\newcommand{\TT}{{\mathbb T}}
\newcommand{\ZZ}{{\mathbb Z}}

\newcommand{\EE}{{\mathbb E}}

\newcommand{\C}{\mathbb{C}}
\newcommand{\R}{\mathbb{R}}
\newcommand{\Z}{\mathbb{Z}}

\newcommand{\cK}{{\mathcal K}}

\newcommand{\cB}{{\mathcal B}}

\newcommand{\sG}{{\mathscr G}}
\newcommand{\fS}{{\mathsf S}}

\newcommand{\vect}[1]{\mathbf{#1}}

\newcommand{\im}{\mathrm{i}}
\newcommand{\cE}{\mathcal{E}}

\newcommand{\vI}{\raisebox{\depth}{\scalebox{1}[-1]{$\Finv$}}}
\newcommand{\vG}{\Finv}
\newcommand{\vC}{\color{brown}{\raisebox{\depth}{\scalebox{1}[-1]{$\Finv$}}}}
\newcommand{\vCG}{\color{brown}{\Finv}}

\newcommand{\pg}{\mathsf{pg}}
\newcommand{\ztwo}{\{\pm 1\}}

\newtheorem{theorem}{Theorem}[section]

\newtheorem{lemma}[theorem]{Lemma}
\newtheorem{proposition}[theorem]{Proposition}
\newtheorem{corollary}[theorem]{Corollary}

\newtheorem{thm}{Theorem}[section]
\newtheorem{prop}[thm]{Proposition}
\newtheorem{lem}[thm]{Lemma}
\newtheorem{cor}[thm]{Corollary}
\theoremstyle{definition}

\theoremstyle{remark}

\theoremstyle{remark}
\newtheorem{definition}[theorem]{Definition}
\newtheorem{remark}[theorem]{Remark}

\newtheorem{example}[theorem]{Example}

\begin{document}

\title{Crystallographic bulk-edge correspondence: glide reflections and twisted mod 2 indices}

\author[1]{Kiyonori Gomi}
\affil[1]{Department of Mathematical Sciences, Shinshu University, Matsumoto, Nagano 390-8621, Japan}

\author[2]{Guo Chuan Thiang}
\affil[2]{School of Mathematical Sciences, University of Adelaide, SA 5000, Australia}


\maketitle

\begin{abstract}
A 2-torsion topological phase exists for Hamiltonians symmetric under the wallpaper group with glide reflection symmetry, corresponding to the unorientable cycle of the Klein bottle fundamental domain. We prove a mod 2 twisted Toeplitz index theorem, which implies a bulk-edge correspondence between this bulk phase and the exotic topological zero modes that it acquires along a boundary glide axis. 
\end{abstract}


\tableofcontents

\section{Introduction}
In the field of topological phases of matter, the \emph{bulk-boundary correspondence} is of paramount importance, allowing topological invariants of spectrally-gapped physical systems in the bulk to be detected at a boundary as certain types of exotic gapless modes. Heuristically, the interface between two bulk gapped phases having different topological invariants must violate the gapped condition in order for ``continuous interpolation'' between the invariants to occur. Furthermore, the gapless modes at the interface reflect the difference in the bulk invariants on either side, thereby enjoying ``topological protection''. The possibility of combining, almost paradoxically, an insulating (gapped) bulk and a conducting (gapless) boundary in a single ``topological material'' opens the path for numerous applications.

The experimentally discovered topological phases exhibit such a correspondence --- indeed it is often difficult to probe the bulk invariants directly, and it is through detection of topological boundary modes that the topological non-triviality of a bulk phase can be inferred. Recently discovered topological boundary modes include the Dirac cones of 3D topological insulators \cite{Hsieh}, and Fermi arcs of Weyl semimetals \cite{WTVS, Xu}. Going further back, the role of edge states in relation to the quantised bulk Hall conductance for the Integer Quantum Hall Effect (IQHE) was already studied theoretically in \cite{Halperin, Hatsugai}. 

Nowadays, topological phases subject to \emph{crystallographic} symmetry constraints are also under intense study \cite{SSG2,Bradlyn,Kruthoff,Wieder,GT}. Most existing works address the classification problem --- identifying and computing the invariants that classify bulk topological phases with given crystallographic symmetry. Following the ideas of \cite{Kitaev}, the relevant topological invariants in the absence of crystallographic point group symmetries can be viewed as certain $K$-theory classes for (quasi)momentum space. Consideration of crystallographic symmetries in this approach led to the introduction of twisted equivariant $K$-theory into the subject \cite{FM, Thiang, SSG2}, and also motivated the study of new generalised notions of $K$-theory twists \cite{FM, Gomi2}, with \emph{graded} twists of particular interest in this paper.

Given the expected universality of the bulk-boundary correspondence, it is desirable to exhibit it as a mathematical theorem. This has been achieved for the IQHE \cite{KRS}, as well as other non-interacting electron systems in the so-called standard ``complex symmetry classes'' \cite{PSB} (no antiunitary or crystallographic point group symmetries), both in the clean limit and in certain disordered regimes. The general idea is that a certain topological boundary (or index) homomorphism maps $K$-theoretic bulk invariants to boundary invariants, and that via an index theorem, this map is also an analytic (Fredholm) index constructed out of a concrete physical model (with Hilbert spaces, Hamiltonians, boundary conditions etc.). As an illustrative example, we explain in detail how the bulk-edge correspondence for ``chiral AIII class'' topological invariants, realised in the Su--Schrieffer--Heeger (SSH) model, can be understood as a classical index theorem for Toeplitz operators. Much less is known mathematically when general crystallographic symmetries are present. 

We will show how this story can be run in reverse --- instead of using a known index theorem to formulate and prove a correspondence of bulk and boundary invariants, we motivate and apply the \emph{crystallographic bulk-boundary correspondence} heuristic to formulate new mathematical index theorems in twisted $K$-theory, which we proceed to prove. The index theorem then justifies the bulk-boundary correspondence. For example, the classical Gohberg--Krein index theorem can be anticipated simply by analysing the bulk-boundary correspondence in the SSH model pictorially (see Fig.\ \ref{fig:SSH}). We focus on the wallpaper group $\pg$, which is physically simple --- the $\ZZ^2$ lattice translations are enhanced by glide reflections --- but rich enough that a twisted $K$-theory group is needed to classify $\pg$-symmetric topological phases. When chiral/sublattice symmetry is also present, there is in fact a 2-torsion bulk phase which we call the \emph{Klein bottle phase}, and it should be detectable through boundary modes along a 1D glide axis. We verify this hypothesis by writing down a concrete 2D tight-binding Hamiltonian which realises the Klein bottle phase, and which indeed has an exotic zero mode on a glide axis boundary. 

An important physical insight for crystallographic bulk-boundary correspondences is that the boundary can have symmetry group a \emph{frieze group} (or more generally a \emph{subperiodic group} \cite{Kopsky}) rather than an ordinary (lower-dimensional) crystallographic space group. In our $\sf{pg}$ example with boundary (an edge) taken to be a glide axis, the symmetry group for the edge is the frieze group $\sf{p11g}$, which is abstractly $\ZZ$ but whose generator is geometrically a glide reflection effecting an exchange of the label ``above/below the axis''. The distinction of $\sf{p11g}$ from the group generated by an ordinary translation is achieved by the data of a nontrivial $\ZZ_2$-grading on $\sf{p11g}\cong\ZZ$ recording the ``above/below'' exchange. The importance of such graded frieze/subperiodic groups in formulating crystallographic bulk-boundary correspondences is further explored in \cite{GT}. Our 2D $\sf{pg}$ example is a basic toy model in which a correspondence of crystallographic $\mathbb{Z}/2$ bulk/edge phases can be formulated rigorously. In 3D, some proposals for material realisation of $\mathbb{Z}/2$ topological crystalline insulators protected by nonsymmorphic symmetry have been made, e.g.\ \cite{FF} pp.\ 3, and \cite{SSG1}. 

We mention in passing earlier work \cite{RH,Hatsugai2}, related to the SSH model, where the role of chiral symmetry in producing edge states in graphene and 2D $d$-wave superconductors is explained. There, the Berry phase angle is constrained to be multiples of $\pi$ so that the exponentiated phase (${\rm U}(1)$ holonomy) is not arbitrary but $\ZZ_2=\{\pm1\}$ valued. Those edge states are not the same as the intrinsically $\ZZ/2$ edge modes in our chiral $\pg$-symmetric model where the glide reflections (together with chiral symmetry) play a fundamental role. A discussion of crystallographic bulk-boundary correspondence appears in \cite{Kubota2}, but our setup is very different; for instance, the geometric choice of boundary picks out an appropriate ``boundary symmetry group'' which the index should respect.

Mathematically, our bulk and edge topological invariants are defined in twisted $K$-theory, with respect to the bulk $\pg$ symmetries and the graded frieze group symmetry of the edge glide axis. There is a natural topological index map (an equivariant Gysin push-forward map) which ``integrates out'' the quasimomenta transverse to the glide axis, and facilitates the computation of the bulk and edge $K$-theory groups (Proposition \ref{prop:twistedK1}). We prove a mod 2 index theorem (Theorem \ref{thm:twistedindextheorem}) which states that the analytic index map taking the ($K$-theory class) of a Hamiltonian to its edge zero modes, is equal to the topological index map. This demonstrates that the Klein bottle phase is indeed topologically non-trivial and is associated to a new type of topologically protected edge state, characterised by a subtle mod 2 index.

\subsection*{Outline}
We begin with a careful analysis of the bulk-boundary correspondence in the SSH model in Section \ref{sec:SSHsection}, and explain how it is a consequence of a classical index theorem. In Section \ref{sec:pgtight}, we construct a tight-binding model with $\pg$ and chiral symmetry, illustrate some ``dimerised'' Hamiltonians, and argue in Section \ref{sec:1Dedge} that one of them (the ``Klein bottle phase'') exhibits 2-torsion topological zero modes when truncated along a glide axis. We start the detailed mathematical analysis in Section \ref{sec:twistedindices}, introducing twisted (Bloch) vector bundles and their classification by twisted $K$-theory, and deducing that for $\pg$ symmetry, the bulk invariant can have 2-torsion. In Section \ref{sec:edgeinvariants}, we introduce graded symmetry groups (relevant for glide axes) which induce graded $K$-theory twists for the edge topological invariant, and clarify connections between several models for $K$-theory with such twists. In Section \ref{sec:genToeplitz}, we construct the bulk-edge homomorphisms topologically and analytically, and prove a mod 2 index theorem that shows their equality. We also outline a crystallographic T-duality which is suggested by our results. Detailed computations of the bulk and edge twisted $K$-theory groups are provided in an Appendix.
\medskip

{\bf Notation}: $\ZZ_2$ denotes the 2-element group $\ztwo$ written multiplicatively, while $\ZZ/2=\{0,1\}$ is the additive version.

\section{Warm up: 1D Su--Schrieffer--Heeger model}\label{sec:SSHsection}
\subsection{1D chiral symmetric Hamiltonians}
Consider Hamiltonians $H$ with chiral/sublattice symmetry $\fS$ (i.e.\ $H\fS=-\fS H$) as well as a 1D lattice $\ZZ$ of translation symmetries. The Brillouin zone is a circle $S^1$, with coordinate $k\in [0,2\pi]/_{\sim 0\equiv 2\pi }$. A basic model begins with a tight-binding Hilbert space $l^2(\ZZ) \otimes \CC^2$, representing a degree of freedom at each site of an infinite chain, which is split by a sublattice operator $\fS=\fS^\dagger, \fS^2=1$ which commutes with translations by $\ZZ$. This means that the chain comprises two translation invariant sublattices labelled by $A$ and $B$ (Fig.\ \ref{fig:SSH}), corresponding respectively to the $+1$ and $-1$ eigenspaces of $\fS$. With a choice of unit cell, each containing one A and one B site, and a basis for its two degrees of freedom so that $\CC[\binom{1}{0}],\CC[\binom{0}{1}]$ corresponds to the A and B sites respectively, we can write the Hilbert space as $l^2(\ZZ)\oplus l^2(\ZZ)$ which Fourier transforms to $L^2(S^1)\oplus L^2(S^1)$. Then translation by $n\in\ZZ$ becomes pointwise multiplication, for each $k\in S^1$, by $e^{\im nk}$ on each direct sum factor, while $\fS$ acts as
$\fS(k)=\begin{pmatrix}1 & 0 \\ 0 & -1\end{pmatrix}$. 

By definition of chiral/sublattice symmetry, $H$ anticommutes with $\fS$, so 
\begin{equation}
H(k)=\begin{pmatrix}0 & U(k) \\ U(k)^* & 0 \end{pmatrix}, \quad U(k)\in\CC.\label{SSHmatrix}
\end{equation} 
We assume that $k\mapsto H(k)$ is continuous, so that $k\mapsto U(k)$ is continuous\footnote{This is related to decay of the hopping terms as the hopping range goes to infinity.}. If $H$ is furthermore gapped, $(H(k))^2>0$ so we need $U(k)U(k)^*\neq 0$, so that $U(k)\in\CC^*\cong \text{GL}(1)$. We can homotope $H$ to $\text{sgn}(H)$, which corresponds to replacing $U(k)$ by $U(k)/|U(k))|\in\text{U}(1)$. There is an obvious topological invariant, which is the winding number of $U:S^1\rightarrow \text{U}(1)$. A typical $H$ with such a symmetry is the Hamiltonian of Su--Schrieffer--Heeger used to model the polymer polyacetylene.

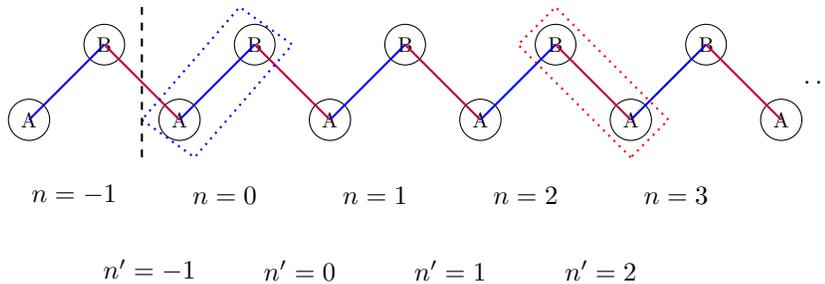
\begin{figure}
\begin{tikzpicture}[every node/.style={scale=1}]
\foreach \p in {(0,0),(2,0),(4,0),(6,0),(8,0),(10,0)}{\draw \p node[circle,draw,scale=0.8] {A};}
\foreach \p in {(1,1),(3,1),(5,1),(7,1),(9,1)}{\draw \p node[circle,draw,scale=0.8] {B};}

\draw[thick,blue] (0,0) -- (1,1);
\draw[thick,blue] (2,0) -- (3,1);
\draw[thick,blue] (4,0) -- (5,1);
\draw[thick,blue] (6,0) -- (7,1);
\draw[thick,blue] (8,0) -- (9,1);

\draw[thick,purple] (1,1) -- (2,0);
\draw[thick,purple] (3,1) -- (4,0);
\draw[thick,purple] (5,1) -- (6,0);
\draw[thick,purple] (7,1) -- (8,0);
\draw[thick,purple] (9,1) -- (10,0);

\node at (10.5,0.5) {$\ldots$};
\draw[dotted,blue,thick] (1.5,0) -- (2.2,-0.5) -- (3.5,1) -- (2.8,1.5) -- (1.5,0);

\draw[dotted,red,thick] (8,-0.5) -- (6.5,1) -- (7,1.5) -- (8.5,0) -- (8,-0.5);

\node at (0.6,-1) {$n=-1$};
\node at (2.6,-1) {$n=0$};
\node at (4.6,-1) {$n=1$};
\node at (6.6,-1) {$n=2$};
\node at (8.6,-1) {$n=3$};

\node at (1.6,-2) {$n'=-1$};
\node at (3.6,-2) {$n'=0$};
\node at (5.6,-2) {$n'=1$};
\node at (7.6,-2) {$n'=2$};

\draw[dashed,thick] (1.5,-0.5)--(1.5,1.5);

\end{tikzpicture}
\caption{An infinite chain with two sublattices. Blue and red dotted rectangles indicate two different choices of unit cells, with respective labels $n$ and $n'$. Blue/red connections represent two dimerised limits $H_{\rm blue}, H_{\rm red}$ of chiral symmetric Hamiltonians. Only $H_{\rm red}$ acquires a dangling zero A-mode when truncated to $n\geq 0$.}\label{fig:SSH}
\end{figure}

\subsection{Toeplitz index theorem and dangling boundary modes}
Let us try to understand the position space meaning of $U$ and its winding number $\text{Wind}(U)$. The ``hopping term'' $U_{\rm blue}$ taking A to B rightwards within a unit cell is represented, after Fourier transform, by $\begin{pmatrix} 0 & 0 \\ 1 & 0\end{pmatrix}$, whereas the term $U_{\rm red}$ taking B to A rightwards \emph{changes} unit cell and is represented by $\begin{pmatrix} 0 & e^{\im k} \\ 0 & 0\end{pmatrix}$. The general Hamiltonian is a self-adjoint combination of powers of $U_{\rm blue}, U_{\rm red}$ which is also required to be gapped and compatible with $\fS$, so that after Fourier transform, it has the form in Eq.\ \eqref{SSHmatrix}. 

Consider the ``fully-dimerised'' Hamiltonian $H_{\rm blue}=U_{\rm blue}+U^\dagger_{\rm blue}=\begin{pmatrix} 0 & 1 \\ 1 & 0\end{pmatrix}$ which has winding number $0$. Another fully-dimerised Hamiltonian is $H_{\rm red}=U_{\rm red}+U^\dagger_{\rm red}=\begin{pmatrix} 0 & e^{\im k} \\ e^{-\im k} & 0\end{pmatrix}$ has winding number $1$. More generally, if a $\fS$-compatible gapped Hamiltonian has Wind($U$)$=w$, this means that its ``dimerised limit'' is given by a $B\rightarrow A$ hopping term which crosses $w$ unit cells to the right.

It is clear that $H_{\rm blue}$ and $H_{\rm red}$ are unitarily equivalent by translating the $A$ sublattice by one unit, or equivalently, choosing a different unit cell convention such that $T_{\rm red}$ is an intra-cell hopping term rather than an inter-cell one\footnote{In Fourier space, this is effected by the large gauge transformation \cite{Thiang2}
$\begin{pmatrix} e^{\im k} & 0 \\ 0 & 1\end{pmatrix}$ corresponding to shifting the origin of the $A$ lattice by one unit.}, see the difference in site labels $n, n'$ in Fig.\ \ref{fig:SSH}. Thus the absolute winding number has some inherent ambiguity, although the \emph{change} in winding number does not --- this is already familiar from the notion of \emph{polarisation}. The difference between the winding numbers for $H_{\rm blue}$ and $H_{\rm red}$ (computed using either unit cell convention) means that they cannot be deformed into one another without violating the gapped condition.

Nevertheless, when a boundary of the chain is specified, the unit cells are uniquely defined under the convention that only complete unit cells are allowed on the half-line --- this corresponds to a ``gauge-fixing''. With respect to a boundary, the bulk topological invariant of $H$ (the winding number) is well-defined, and corresponds to a boundary invariant (number of dangling zero modes) through a Toeplitz index theorem, as we will now explain. In particular, there is a well-defined notion of a ``zero/trivial phase'' which will have no such boundary zero modes.

The Hilbert subspace of $L^2(S^1)$ spanned by the $A$ (resp.\ $B$) degrees of freedom in the right-half-line comprises the functions with non-negative Fourier coefficients --- this is also called the \emph{Hardy space} $\mathcal{H}^2_A$ (resp.\ $\mathcal{H}^2_B$). After truncating $L^2(S^1)\oplus L^2(S^1)$ to $\mathcal{H}^2_A\oplus \mathcal{H}^2_B$, the hopping operator $U_{\rm blue}$ remains unitary, whereas $U_{\rm red}$ becomes only an isometry $T_{U_{\rm red}}$ since $T_{U_{\rm red}}T_{U_{\rm red}}^\dagger=1-p_{n_A=0}$ where $p_{n_A=0}$ is the projection onto the A-site at $n=0$. In fact, $T_{U_{\rm red}}$ is a \emph{Toeplitz} operator with \emph{symbol} the invertible function $U_{\rm red}(k)=e^{\im k}$, which is Fredholm with index equal to $-\text{Wind}(U_{\rm red})=-1$. 

Recall that a bounded operator is \emph{Fredholm} if its kernel and cokernel are finite-dimensional. Its \emph{Fredholm index} is the difference of these dimensions, so for example, the index of $T_{U_{\rm red}}$ is the number of zero modes of $T_{U_{\rm red}}$ minus the number of zero modes of $T_{U_{\rm red}}^\dagger$. The product of two Fredholm operators is again Fredholm, and the index of a product is the sum of the indices. The \emph{Toeplitz operator} $T_U$ with continuous symbol $U\in C(S^1)$ is the compression
 of the multiplication operator $M_U$ on $L^2(S^1)$ to the Hardy space,
\begin{equation*}
\begin{tikzcd}
\mathcal{H}^2\ar[r,"\iota"]\ar[rrr,out=-30,in=210,swap,"T_U"] & L^2(S^1)\ar[r,"M_U"] & L^2(S^1) \ar[r,"p"] & \mathcal{H}^2.
\end{tikzcd}
\end{equation*}
Here, $\iota$ is the inclusion and $p$ is the orthogonal projection onto $\mathcal{H}^2$ so that $p\circ\iota$ is the identity map. The operator $T_U$ is Fredholm iff its symbol $U$ is invertible everywhere \cite{Coburn}. The Gohberg--Krein index theorem identifies the (analytic) Fredholm index of an invertible Toeplitz operator with the (topological) winding number index of its symbol, up to a sign \cite{GK}. 

The SSH model provides a physical interpretation of this index theorem as a bulk-edge correspondence identifying the bulk index/winding number with the number of unpaired edge modes\footnote{More precisely, the number of unpaired $A$ modes minus the number of unpaired $B$ modes.}. Intuitively, the bulk-with boundary Hamiltonian $\breve{H}_{\rm red}$ ($\text{Wind}=1$) acquires an unpaired zero-energy $A$ mode because the $B\rightarrow A$ hopping term from unit cell $n=-1$ to $n=0$ is ``cut off'' by the boundary and replaced by the zero operator on the $A$ site at $n=0$, On the other hand, the hopping terms on the right half-line remain unchanged for $\breve{H}_{\rm blue}$ ($\text{Wind}=0$) so that all sites stay paired up.

\subsection{Cancellation of winding numbers and zero modes}
Higher winding numbers may be obtained by ``dimerising'' across more unit cells, which causes more bonds to be intercepted by the boundary, leaving behind more unpaired zero modes. This may be unnatural from the energetics point of view, so an alternative is to consider direct sums of the basic SSH model, so that there are $2N$ sites per unit cell. Then a $\fS$-symmetric Hamiltonian has the form 
\begin{equation*}
H(k)=\begin{pmatrix}0 & U(k) \\ U(k)^\dagger & 0 \end{pmatrix}, \quad U(k)\in {\rm GL}(N,\CC)
\end{equation*} 
which still has a topological invariant given by the ordinary winding number of det($U$). Allowing $N$ to be finite but arbitrarily large, we can form direct sums $H(k)\oplus H'(k)$ whose off-diagonal term is $U(k)\oplus U'(k)\in {\rm GL}(N+N',\CC)$. 

Let us analyse what happens to the boundary zero modes under direct sum. Fig.\ \ref{fig:SSH2} represents a dimerised Hamiltonian $H_{\rm green}$ with winding number $-1$. We see that the truncated $\breve{H}_{\rm green}$ acquires a dangling $B$ zero mode at the boundary. If we consider $H_{\rm red}\oplus H_{\rm green}$, their combined winding number is zero. In terms of the boundary edge modes, this says that the $A$ and $B$ boundary zero modes at position $n=0$ can be paired up and gapped out by turning on a boundary term which is compatible (anticommutes) with $\fS |_{n=0}$,
$$H_{\rm bdry}=\begin{pmatrix} 0 &a \\ \overline a & 0\end{pmatrix},$$
which has spectrum $\pm |a|$, cf.\ Fig.\ \ref{fig:SSH3}. Thus the two zero modes in combination may \emph{not} be topologically protected, in accordance with the cancelling winding numbers.

\begin{figure}
\begin{tikzpicture}[every node/.style={scale=1}]
\foreach \p in {(0,0),(2,0),(4,0),(6,0),(8,0)}{\draw \p node[circle,draw,scale=0.8] {A};}
\foreach \p in {(1,1),(3,1),(5,1),(7,1),(9,1)}{\draw \p node[circle,draw,scale=0.8] {B};}

\draw[thick,green] (-2,0) -- (1,1);
\draw[thick,green] (0,0) -- (3,1);
\draw[thick,green] (2,0) -- (5,1);
\draw[thick,green] (4,0) -- (7,1);
\draw[thick,green] (6,0) -- (9,1);
\draw[thick,green] (8,0) -- (11,1);

\node at (10.5,0.5) {$\ldots$};
\draw[dotted,thick] (1.5,0) -- (2.2,-0.5) -- (3.5,1) -- (2.8,1.5) -- (1.5,0);


\node at (0.6,-1) {$n=-1$};
\node at (2.6,-1) {$n=0$};
\node at (4.6,-1) {$n=1$};
\node at (6.6,-1) {$n=2$};
\node at (8.6,-1) {$n=3$};

\draw[dashed,thick] (1.5,-0.5)--(1.5,1.5);

\end{tikzpicture}
\caption{A dimerised Hamiltonian $H_{\rm green}$ with winding number $-1$. It acquires a dangling zero $B$-mode when truncated to $n\geq 0$.}\label{fig:SSH2}
\end{figure}
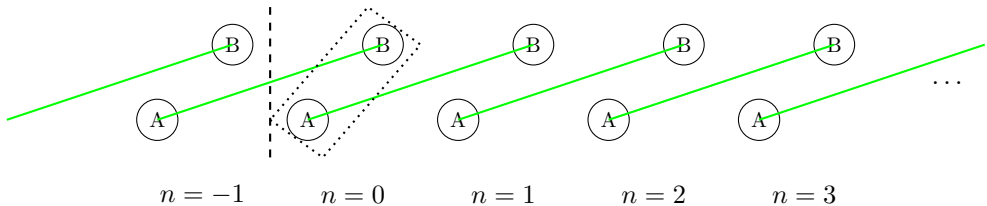

\begin{figure}
\begin{tikzpicture}[every node/.style={scale=1}]
\foreach \p in {(2,0),(4,0),(6,0),(8,0)}{\draw \p node[circle,draw,scale=0.8] {A};}
\foreach \p in {(3,1),(5,1),(7,1),(9,1)}{\draw \p node[circle,draw,scale=0.8] {B};}
\foreach \p in {(2,2),(4,2),(6,2),(8,2)}{\draw \p node[circle,draw,scale=0.8] {A};}
\foreach \p in {(3,3),(5,3),(7,3),(9,3)}{\draw \p node[circle,draw,scale=0.8] {B};}

\draw[thick,red] (4,2) -- (3,3);
\draw[thick,red] (6,2) -- (5,3);
\draw[thick,red] (8,2) -- (7,3);
\draw[thick,red] (10,2) -- (9,3);

\draw[thick,green] (2,0) -- (5,1);
\draw[thick,green] (4,0) -- (7,1);
\draw[thick,green] (6,0) -- (9,1);
\draw[thick,green] (8,0) -- (11,1);

\draw[thick,black] (2,2) -- (3,1);

\node at (10.5,0.5) {$\ldots$};
\draw[dotted,thick] (1.5,-0.5) -- (3.5,-0.5) -- (3.5,3.5) -- (1.5,3.5) -- (1.5,-0.5);


\node at (0.6,-1) {$n=-1$};
\node at (2.6,-1) {$n=0$};
\node at (4.6,-1) {$n=1$};
\node at (6.6,-1) {$n=2$};
\node at (8.6,-1) {$n=3$};

\draw[dashed,thick] (1.5,-0.7)--(1.5,3.7);

\end{tikzpicture}
\caption{The zero A-mode of $\breve{H}_{\rm red}$ and zero B-mode of $\breve{H}_{\rm green}$ at the boundary can be paired up (black line) and gapped out.}\label{fig:SSH3}
\end{figure}
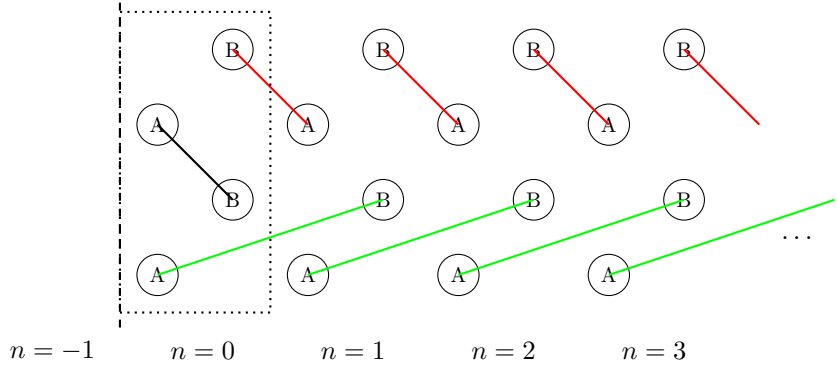

To handle direct sums, we need to be able to (1) consider Hamiltonians with $N$ finite but arbitrarily large, and (2) extract a topological invariant for $U\oplus U'$ which is additive with respect to $\oplus$, in accordance with the addition of boundary zero modes. This is precisely what $K$-theory allows us to do. In more detail, allowing for direct sums and noting that invertible matrices can be retracted to unitary ones, the $\fS$-symmetric Hamiltonians are classified by the homotopy classes of maps from $S^1$ into the infinite unitary group $\mathscr{U}$. Since $\mathscr{U}$ is a classifying space for the odd-degree complex $K$-theory functor \cite{Karoubi}, this classification group is just $K^{-1}(S^1)$.

For each $N$, the winding number of a map $S^1\rightarrow {\rm U}(N)$ can be defined as the ordinary winding number of its determinant, and this remains an invariant of its homotopy class. It is also a homomorphism with respect to matrix multiplication, in that Wind($UV$)$=$Wind($U$)+Wind($V$). A remarkable property of unitary matrices in $\mathscr{U}$ is that $U\oplus V=(U\oplus 1)(1\oplus V)$ becomes \emph{homotopic} to $UV\oplus 1$ in $\mathscr{U}$ (we may assume that $U$ and $V$ have the same size by appending 1 along diagonals, which does not change their winding numbers). Thus $\oplus$ and matrix multiplication are on the same footing with regards to topological invariants of maps $S^1\rightarrow\mathscr{U}$. In particular, Wind($\cdot$) is well-defined on $K$-theory classes $[U]$, giving a homomorphism $K^{-1}(S^1)\rightarrow \ZZ\cong K^0({\rm pt})$, which is in fact an isomorphism due to Bott periodicity applied to $K^{-1}(S^1)\cong \wt{K}^0(S^2)\cong K^{-2}({\rm pt})\cong K^0({\rm pt})$. 

On the analytic side, we can extend the discussion of Toeplitz operators to that on $(\mathcal{H}^2)^{\oplus N}$ for any finite $N$. A continuous $N\times N$ matrix-valued function $U$ on $S^1$ defines a multiplication operator by $U$ on $(L^2(S^1))^{\oplus N}$, whose truncation to $(\mathcal{H}^2)^{\oplus N}$ is the Toeplitz operator $T_U$ with symbol $U$. As before, $T_U$ is Fredholm iff $U$ is invertible, and the index of $T_U$ is equal to $-$Wind($U$). The index is of course unchanged if $T_U$ is modified into $T_U\oplus 1_{N'}$ acting on $N'$ extra copies of $\mathcal{H}^2$. Thus we can consider the map $[U]\mapsto$ Index($T_U$) as a homomorphism $K^{-1}(S^1)\rightarrow\ZZ$.
\begin{theorem}[Toeplitz index theorem, e.g.\ \cite{Benameur}]\label{thm:Toeplitz}
Let $U\in C(S^1,{\rm GL}(N,\CC))$ represent a class $[U]$ in $K^{-1}(S^1)$. The analytic index map $[U]\mapsto {\rm Index}(T_U)$ is equal to the topological index $-{\rm Wind}(U)$; both are isomorphisms into $\ZZ\cong K^0({\rm pt})$.
\end{theorem}

\begin{remark}\label{rem:cancellingzeroes}
Since we have seen that the matrix symbol $U_{\rm red}\oplus U_{\rm green}$ is homotopic to $U_{\rm blue}\oplus U_{\rm blue}$, the operator $\breve{H}_{\rm red}\oplus \breve{H}_{\rm green}$ is homotopic to $\breve{H}_{\rm blue}\oplus \breve{H}_{\rm blue}$ which has no dangling zero modes at all. This shows that boundary zero modes may be created or destroyed (``gapped out'') in A-B pairs, as $\breve{H}$ undergoes a homotopy through Toeplitz operators. 
\end{remark}

\begin{remark}
By appending to $T_U$ the identity operator on the orthogonal complement of $(\mathcal{H}^2)^{\oplus N}$ in $(L^2(S^1))^{\oplus N}$, one obtains an elliptic psuedo-differential operator of order 0 with the same analytic index as $T_U$. Then the Toeplitz index theorem may be viewed in the sense of Atiyah--Singer \cite{AS1}, and this applies to more general notions of Toeplitz operators \cite{dM}. In particular, if $U$ is smooth, the index may be computed by the local formula
$${\rm Index}(T_U)=-\frac{1}{2\pi\im}\int_{S^1}{\rm tr}(U^{-1}dU). $$
\end{remark}

\begin{remark}
Bott periodicity $K^{-2}(X)\cong K^0(X)$ may in fact be proved using the notion of the index bundle for a family of (Toeplitz) Fredholm operators \cite{A1, A2}, which we will utilise in Section \ref{sec:genToeplitz}. This approach is especially useful for noncommutative generalisations, see \cite{Cuntz, Phillips}.
\end{remark}

\subsection{Families of Toeplitz operators}
We can also consider a continuous family of Toeplitz operators $x\mapsto T_U(x)$ parameterised by a topological space $X$. For instance, take $\ZZ^2=\ZZ_x\oplus\ZZ_y$ where $\ZZ_x, \ZZ_y$ denote the subgroups of lattice translations in the $x$ and $y$ directions respectively. We can Fourier transform only with respect to the $\ZZ_x$ factor, to get a family $\cB_x\ni k_x\mapsto M_{U(k_x,\cdot)}$ of invertible operators on $(l^2(\ZZ_y))^{\oplus N}\cong (L^2(\cB_y))^{\oplus N}$ corresponding at each $k_x$ to multiplication on $(L^2(\cB_y))^{\oplus N}$ by the function $U(k_x,\cdot):k_y\mapsto U(k_x, k_y)$. Truncating $L^2(\cB_y)$ to Hardy space, we obtain a family $k_x\mapsto T_U(k_x)$ of Toeplitz operators, whose symbol at $k_x$ is the function $U(k_x,\cdot)$. If $U(k_x,\cdot)$ is invertible for all $k_x$, then $k_x\mapsto T_U(k_x)$ is a family of Fredholm operators, which has a (virtual) index \emph{bundle} well-defined as an element of $K^0(\cB_x)$. Because $\cB_x=S^1$, the class of the index bundle is determined by its (virtual) rank.

More interesting index bundles can appear for higher dimensional $X$. For a general compact connected $X$, a continuous function $U$ on $X\times S^1$ gives a family $x\mapsto T_U(x)$ of Toeplitz operators with virtual index bundle an element of $K^0(X)$. We may also associate a topological index as follows. The function $U$ determines the clutching data for a vector bundle over $X\times S^2$ and a $K$-theory class in $K^0(X\times S^2)$. Composing with the inverse of the Bott isomorphism gives an element in $K^0(X)$, which is in fact represented by the analytic index bundle \cite{A2}. 

By considering a wallpaper group which extends $\ZZ^2$ translations by glide reflection symmetries, we will construct a \emph{twisted} family of Toeplitz operators over $\cB_x$, whose index lives in a twisted $K$-theory of $\cB_x$.

\section{2D tight-binding model with glide reflection symmetry}\label{sec:pgtight}
\subsection{Space groups, unit cells, and orbifolds} A \emph{space group} (or \emph{crystallographic group}) $\sG$ in $d$-dimensions is a discrete cocompact subgroup of the group $\mathrm{Euc}(d)\cong\RR^d\rtimes{\mathrm O}(d)$ of isometries of Euclidean space $\EE^d$. Euclidean space $\EE^d$ is an affine version of the normal subgroup $\RR^d$ of translations, and ${\mathrm O}(d)$ is the orthogonal group fixing an origin. 

A space group $\sG$ has a maximal abelian subgroup $N=\mathscr{G}\cap\RR^d\cong\ZZ^d$ generated by $d$ independent translations, called the (translation) \emph{lattice}, which is also normal with quotient $F=\sG/\ZZ^d\subset {\rm O}(d)$ the finite \emph{point group}. Thus, there is a diagram
\begin{equation}
\begin{CD}
0@>>>\RR^d@>>>\RR^d\rtimes{\rm O}(d)@>>>{\rm O}(d)@>>>1
 \\
@.@AAA @ AAA @ AAA \\
0@>>>N=\ZZ^d@>>> \mathscr{G} @>>> F @ >>> 1\end{CD}\label{spacegroupsequence}
\end{equation}
where the vertical maps are inclusions. In Eq.\ \eqref{spacegroupsequence}, we have written $0=\{0\}$ for the trivial subgroup of $N\subset\RR^d$ written additively, and also $1=\{1\}$ for the trivial quotient group written multiplicatively. The top row expresses the Euclidean group as a semidirect product, while the bottom row expresses the space group $\mathscr{G}$ as an extension of $F$ by $N$.

A lift $\breve{g}\in\sG$ of an element $g\in F$ may be written as a Euclidean transformation $[t_g;O_g]\in \RR^d\rtimes{\mathrm O}(d)$ where $t_g$ is the translational part and $O_g$ is an orthogonal transformation fixing some origin. If $F$ can be lifted homomorphically back into $\sG$ as a subgroup, then $\mathscr{G}$ is isomorphic to a semidirect product $N\rtimes F$ and we say that $\sG$ is \emph{symmorphic}; the lifts $\breve{g}$ can then be chosen to have no translational part. 

An example of a \emph{nonsymmorphic} group is the torsion-free wallpaper group $\pg$,
$$ 0\longrightarrow \ZZ^2\longrightarrow \pg\longrightarrow \ZZ_2\longrightarrow 1,$$
in which any lift of the non-trivial point group element is a \emph{glide reflection} of infinite order. We will discuss the $\pg$ group in detail in the next Subsection.

By convention, a (primitive) \emph{unit cell} refers to a fundamental domain with respect to only the discrete translational subgroup $\ZZ^d$ of a space group. This is not unique but is always topologically a torus $\TT^d=\EE^d/\ZZ^d$, and can be thought of as the orbit space under lattice translations. It is important to distinguish $\TT^d$ from the \emph{Brillouin torus} $\cB$.

Since $\ZZ^d$ is normal in $\sG$, there is an induced action of $F$ on $\TT^d$, so that the unit cell $\TT^d$ comes equipped with a finite group action by $F$. The quotient $\TT^d/F \cong \EE^d/\sG$ is a fundamental domain for the \emph{full} space group $\sG$, and can be used as a basic ``tile'' for Euclidean space. 

The group $\pg$ provides a particularly nice example where $F=\ZZ_2$ acts \emph{freely} (i.e.\ without fixed points) so that the quotient $\TT^2/\ZZ_2$ is a \emph{manifold} --- in fact it is just the Klein bottle $\cK$, see Fig.\ \ref{fig:pgunitcell}. In general, $F$ acts with fixed points, so that $\TT^d/F$ is an \emph{orbifold} --- the orbifold approach to space groups can be found in \cite{CFHT}.

\subsection{Tight-binding model with $\pg$ symmetry}
The group $\pg$ is one of the two \emph{torsion-free} wallpaper groups --- the other being $\ZZ^2=\ZZ\oplus\ZZ$. Abstractly, it is isomorphic to a semidirect product 
$\ZZ\rtimes\ZZ$ in which the second copy of $\ZZ$ acts on the first by reflection. In the short exact sequence
\begin{equation}
0\longrightarrow N=\ZZ\oplus\ZZ\overset{i}{\hookrightarrow} \pg=\ZZ\rtimes\ZZ\overset{q}{\twoheadrightarrow}F=\ZZ_2\longrightarrow 1,\label{pgextension}
\end{equation}
the injection is\footnote{We use the semicolon for elements $(n_y;m_x)\in\pg$ to avoid confusion with elements $(n_y,n_x)\in N$. We also use $n_y$ instead of $n_x$ for the first coordinate so that the diagrams that follow are more convenient to draw.} $i:(n_y,n_x)\mapsto(n_y;2n_x)$ and the surjection is $q:(n_y;m_x)\mapsto (-1)^{m_x}$. The multiplication law in $\pg$ is 
$$(n_y;m_x)(n_y';m_x')=(n_y+(-1)^{m_x}n_y';m_x+m_x').$$
Since $\pg$ has no elements of order 2, it is not a split extension of $F$ by $N$, and is thus nonsymmorphic. The elements $(n_y;1)\in\pg$ are lifts of $-1\in F$, and they all have the property that they square to $(n_y;1)^2=(0;2)=i(0,1)$, i.e.\ a unit translation in the $n_x$ variable.

Concretely, we can identify $\pg$ as a subgroup of the Euclidean group ${\rm E}(2)$, where the latter can be written as $(\RR_y\oplus\RR_x)\rtimes {\rm O}(2)$ upon picking a base point and an orthonormal basis of translations. We use square brackets $[(t_y,t_x);O]$ to denote a Euclidean transformation. Then $-1\in F$ is identified with the reflection matrix $R_y=\begin{pmatrix} -1 & 0 \\ 0 & 1\end{pmatrix}$, and the general element $(n_y;m_x)\in\pg$ is identified with the Euclidean transformation $[(n_y,\frac{m_x}{2});R_y^{m_x}]$. Thus, for instance, $(0;1)$ is a \emph{glide reflection} about the horizontal $x$-axis, i.e.\ reflection of $y$-coordinate followed by half-translation of the $x$-coordinate. More generally, $(n_y;1)$ is a glide reflection about some horizontal \emph{glide axis} which is left invariant under $(n_y;1)$. We can then write Eq.\ \eqref{pgextension} more concretely as 
\begin{equation*}
0\longrightarrow \ZZ_y\oplus\ZZ_x\overset{i}{\hookrightarrow} \pg=\ZZ_y\rtimes\ZZ_g\overset{q}{\twoheadrightarrow}F=\ZZ_2\longrightarrow 1,
\end{equation*}
where $\ZZ_g$ is the subgroup of $\pg$ generated by a glide reflection about the $x$-axis. The quotient $\EE^2/\pg$ is a Klein bottle, see Fig.\ \ref{fig:pgunitcell}, which is a flat closed manifold. In reverse, the fundamental group of the Klein bottle is isomorphic to $\pg$.

{\bf Convention for labelling unit cells, glide axes, and atomic sites.}
Formally, the Cayley graph for $\pg$ embeds in $\EE^2$ as a set of ``atomic sites'', e.g.\ indicated by Fig.\ \ref{fig:pgunitcell} or the black symbols in Fig.\ \ref{fig:pglattice}. Pick one atomic site (Wyckoff position) just below (or possibly on) the glide axis for $(0;1)$ to be labelled by the identity element $(0;0)\in\pg$, then the other atomic sites can be labelled uniquely by $(n_y;m_x)\in\pg$ according to the transformation needed to reach there from $(0;0)$. We choose unit cells such that the unit cell labelled by $(n_y,n_x)\in N$ contains the two atomic sites $(n_y;2n_x), (n_y;2n_x+1)\in\pg$, and lies between the two glide axes labelled by $l=2n_y$ and $l=2(n_y+1)$, as in Fig.\ \ref{fig:pgunitcell}. Then the glide axis $l=2n_y+1=2(n_y+\frac{1}{2})$ passes through the middle of the unit cells at $(n_y,n_x), n_x\in\ZZ_x$.

\begin{figure}
\subfigure{
\begin{tikzpicture}[scale=0.9,every node/.style={scale=1.3}]
\draw[help lines,dotted,thick,step=2] (-2,0) grid (4,6);

\foreach \p in {(-1.8,0.5),(0.2,0.5),(2.2,0.5)}{\draw \p node[scale=0.6] {$\vI$};}
\foreach \p in {(-1.8,2.5),(0.2,2.5),(2.2,2.5)}{\draw \p node[scale=0.6] {$\vI$};}
\foreach \p in {(-1.8,4.5),(0.2,4.5),(2.2,4.5)}{\draw \p node[scale=0.6] {$\vI$};}

\foreach \p in {(-0.8,1.5),(1.2,1.5),(3.2,1.5)}{\draw \p node[scale=0.6] {$\vG$};}
\foreach \p in {(-0.8,3.5),(1.2,3.5),(3.2,3.5)}{\draw \p node[scale=0.6] {$\vG$};}
\foreach \p in {(-0.8,5.5),(1.2,5.5),(3.2,5.5)}{\draw \p node[scale=0.6] {$\vG$};}


\draw[dashed] (-3,0) -- (4,0);
\draw[dashed] (-2,1) -- (4,1);
\draw[dashed] (-3,2) -- (4,2);
\draw[dashed] (-2,3) -- (4,3);
\draw[dashed] (-3,4) -- (4,4);
\draw[dashed] (-2,5) -- (4,5);
\draw[dashed] (-3,6) -- (4,6);

\draw (0,2) -- (0,4) -- (2,4) -- (2,2) -- (0,2);
\draw[red] (1,1) -- (1,3) -- (3,3) -- (3,1) -- (1,1);

\node[scale=0.7] at (-1,-1) {$n_x=-1$};
\node[scale=0.7]  at (1,-1) {$n_x=0$};
\node[scale=0.7]  at (3,-1) {$n_x=1$};

\node[scale=0.7] at (-3,1) {$n_y=-1$};
\node[scale=0.7]  at (-3,3) {$n_y=0$};
\node[scale=0.7]  at (-3,5) {$n_y=1$};

\node[scale=0.7] at (5,0) {$l=-2$};
\node[scale=0.7] at (5,1) {$l=-1$};
\node[scale=0.7] at (5,2) {$l=0$};
\node[scale=0.7] at (5,3) {$l=1$};
\node[scale=0.7] at (5,4) {$l=2$};
\node[scale=0.7] at (5,5) {$l=3$};
\node[scale=0.7] at (5,6) {$l=4$};

\node[scale=0.5] at (0.27,2.27) {$(0;0)$};
\node[scale=0.5] at (1.27,3.26) {$(0;1)$};
\node[scale=0.5] at (2.27,2.27) {$(0;2)$};
\node[scale=0.5] at (3.27,3.26) {$(0;3)$};
\node[scale=0.5] at (0.27,0.27) {$(-1;0)$};
\node[scale=0.5] at (1.27,1.26) {$(-1;1)$};

\node[scale=0.7] at (0,3) {$\times$};


\end{tikzpicture}
}
\hspace{1em}
\subfigure{
\begin{tikzpicture}[scale=1.2,every node/.style={scale=1.5}]

\foreach \p in {(0.2,2.5)}{\draw \p node[scale=0.6] {$\vI$};}

\foreach \p in {(1.2,3.5)}{\draw \p node[scale=0.6] {$\vG$};}
\draw[->,>=latex,ultra thick,black]   (0,2) --(0.6,2);
\draw[->,>=latex,ultra thick,black]   (0,4) --(0.6,4);
\draw[->>,>=latex,ultra thick,black]   (0,2) --(0,3.3);
\draw[->>,>=latex,ultra thick,black]   (1,4) --(1,2.7);

\draw[->,>=latex,thick,black]   (1,2) --(1.6,2);
\draw[->,>=latex,thick,black]   (1,4) --(1.6,4);
\draw[->>,>=latex,thick,black]   (2,2) --(2,3.2);

\draw[dashed] (0,3) -- (2,3);

\draw[ultra thick] (0,2) -- (0,4) -- (1,4) -- (1,2) -- (0,2);
\draw[thin] (1,4) -- (2,4) -- (2,2) -- (1,2);

\node[white] at (0,0) {$.$};

\end{tikzpicture}
}
\caption{(L) An embedding of $\pg$ in Euclidean space with origin labelled $\times$. Dotted lines divide Euclidean space into unit cells labelled by $N=\ZZ^2$, each containing two atomic sites. Black solid lines enclose the unit cell labelled by $(n_y,n_x)=(0,0)$. For each $n_y\in\ZZ_y$, there are two glide axes (dashed horizontal lines), labelled by $l=2n_y, 2n_y+1$. Odd $l$ glide axes are intra-cell with respect to the black unit cell convention, while even $l$ ones are ``shared'' between unit cells with different $n_y$. The red square encloses another possible choice of unit cell. (R) Identifying pairs of glide-equivalent points in a unit cell gives a Klein bottle as the quotient space (fundamental domain of $\pg$).} \label{fig:pgunitcell}
\end{figure}
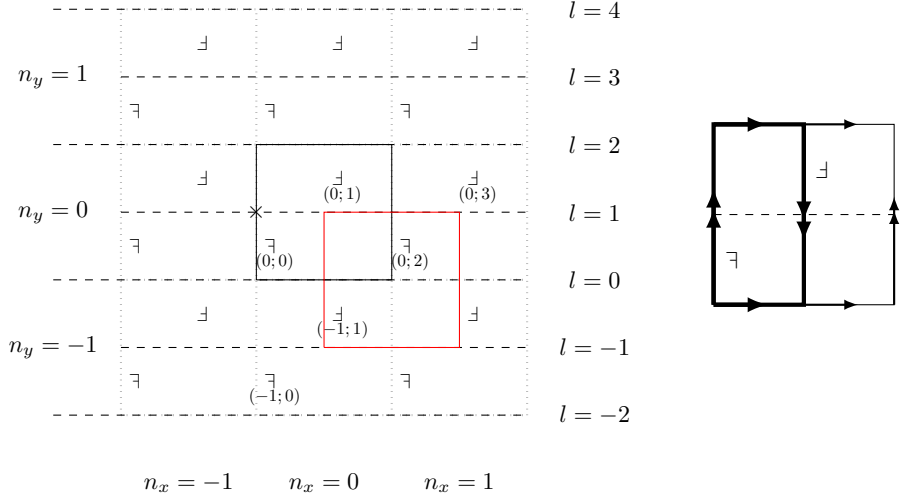

The glide axis bounding the bottom of the $n_y=0$ unit cells, i.e.\ $l=0$, will be taken to be the 1D boundary (henceforth called an \emph{edge}). \emph{Note that with this labelling of glides axes, the $x$-axis is at $l=1$, and the glide reflection $L^{(l)}$ with glide axis $l$ is effected by $(l-1;1)\in\pg$.}

Finally, choose a basis for the intra-cell $\CC\oplus\CC$ degrees of freedom, such that $\begin{pmatrix}1\\0\end{pmatrix}$ corresponds to the sites with even $m_x$ (depicted with the symbol $\vI$), and $\begin{pmatrix}0\\1\end{pmatrix}$ corresponds to the sites with odd $m_x$ (depicted with the reflected symbol $\vG$).

With the above labelling and basis conventions, the tight-binding Hilbert space is $l^2(\pg)\cong\l^2(\ZZ^2)\oplus\l^2(\ZZ^2)\cong L^2(\cB)\oplus L^2(\cB)$ where the latter isomorphism is the Fourier transform. Lattice translations $T_{\vect{n}}$ by $\vect{n}=(n_y,n_x)\in N$ act on $f\in L^2(\cB)\oplus L^2(\cB)$ as pointwise multiplication. For example, the generating horizontal and vertical translation operators $T_x\equiv T_{(0,1)}, T_y\equiv T_{(1,0)}$ are
\begin{equation*}
(T_x\cdot f)(\vect{k})=M_{u_x}f(\vect{k})\equiv\begin{pmatrix} u_x & 0 \\ 0 & u_x \end{pmatrix}f(\vect{k}), \qquad\qquad\qquad\qquad\;\;
\end{equation*}
\begin{equation*}
(T_y\cdot f)(\vect{k})=M_{u_y}f(\vect{k})\equiv\begin{pmatrix} u_y & 0 \\ 0 & u_y \end{pmatrix}f(\vect{k}), \quad \vect{k}=(k_x,k_y)\in\cB,
\end{equation*}
where $u_x:\vect k\rightarrow e^{\im k_x}$ and $u_y:\vect k\rightarrow e^{\im k_y}$ are the functions generating the Fourier algebra on $\cB$. 

Glide reflection $L^{(1)}$ about the $x$-axis $l=1$ takes
$$(n_y;m_x)\mapsto (0;1)(n_y;m_x)=(-n_y;1+m_x)\in\pg,$$
so
\begin{equation*}
(L^{(1)}\cdot f)(\vect{k})=\begin{pmatrix} 0 & u_x(\vect{k}) \\ 1 & 0\end{pmatrix}f(\vect{k}')\equiv M_{V}f(\vect{k}')\equiv M_VIf(\vect{k}),
\end{equation*}
where $V(\vect{k})\coloneqq\begin{pmatrix} 0 & u_x(\vect{k}) \\ 1 & 0\end{pmatrix}=\begin{pmatrix} 0 & e^{\im k_x}\\ 1 & 0\end{pmatrix}$, and $If(\vect{k})\coloneqq f(\vect{k}')\equiv f(k_x,-k_y)$. Here, we have written $\vect{k}'$ for $(k_x,-k_y)$. We can think of $M_V$ as ``horizontal half-translation'' and $I$ as vertical reflection about $l=1$. 

We also need glide reflection $L^{(0)}$ about the 1D edge $l=0$, which takes
$$(n_y;m_x)\mapsto (-1;1)(n_y;m_x)=(-1;0)(0;1)(n_y;m_x)=(-1-n_y;1+m_x)\in\pg,$$
so
\begin{equation*}
(L^{(0)}\cdot f)(\vect{k})=T_y^{-1}M_VIf(\vect{k})=M_{\overline{u_y}V}If(\vect{k})=\begin{pmatrix}0 & e^{\im (k_x-k_y)}\\ e^{-\im k_y} & 0\end{pmatrix}f(\vect{k}'),
\end{equation*}
It is easy to verify that $(L^{(0)})^2=T_x=(L^{(1)})^2$. 

What we have described is the Fourier transformed version of the regular representation of $\pg$ on $l^2(\pg)\cong L^2(\cB)\oplus L^2(\cB)$ (upon making certain choices as explained above). The Hilbert space $L^2(\cB)\oplus L^2(\cB)$, together with operators $T_y$ and $L^{(0)}$ generating a $\pg$ action, may thus be thought of as the section space of the ``regular bundle'' $\cE_{\rm reg}\cong\cB\times\CC^2$. In general, there may be internal degrees of freedom $\CC^n$ at each atomic site, and we would then have $n$ copies of $\cE_{\rm reg}$.

\subsubsection{Tight-binding model with $\pg$ and chiral symmetry}\label{sec:tightbindingmodel}
We wish to study tight-binding Hamiltonians $H$ with $\pg$ symmetry \emph{and} chiral (also called sublattice) symmetry represented by a unitary operator $\fS$. The latter symmetry means that $H\fS=-\fS H$ and $\fS^2=1$, and we also assume that $\fS$ commutes with $\pg$. Thus the $\pm 1$ eigenspaces of $\fS$ (Hilbert subspaces for the two sublattices) must each host the $\pg$ symmetry, so that there are \emph{four} degrees of freedom per unit cell. This is illustrated in Fig.\ \ref{fig:pglattice} where the brown $\vC, \vCG$ are chiral partners of the black $\vI, \vG$ respectively.

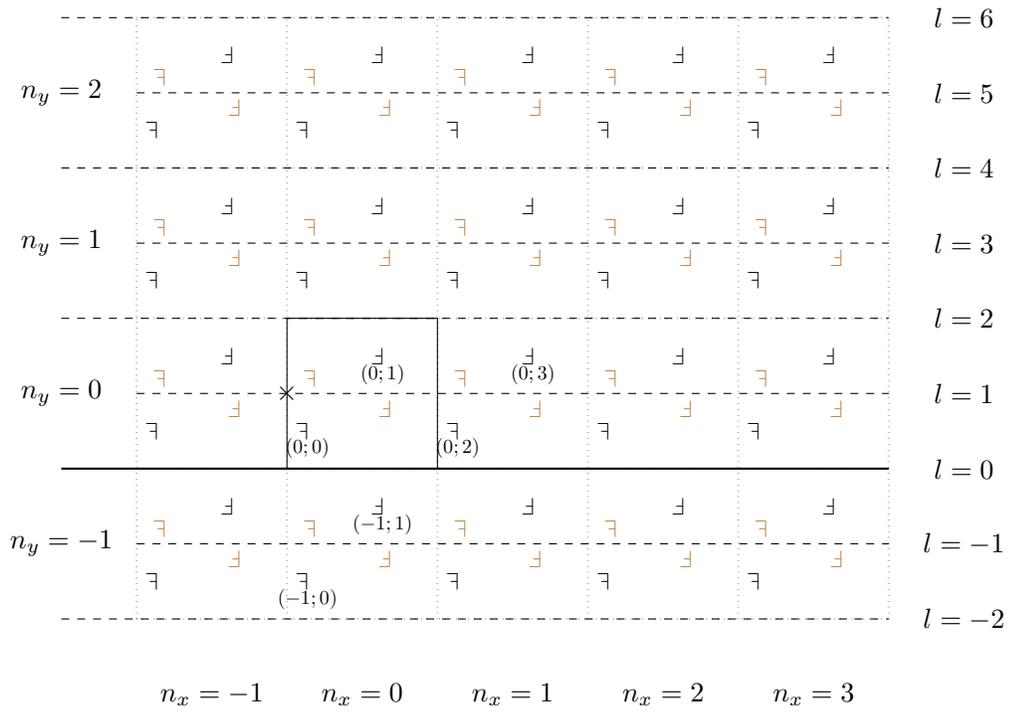
\begin{figure}
\begin{tikzpicture}[every node/.style={scale=1.5}]
\draw[help lines,dotted,thick,step=2] (-2,0) grid (8,8);

\foreach \p in {(-1.8,0.5),(0.2,0.5),(2.2,0.5),(4.2,0.5),(6.2,0.5)}{\draw \p node[scale=0.6] {$\vI$};}
\foreach \p in {(-1.8,2.5),(0.2,2.5),(2.2,2.5),(4.2,2.5),(6.2,2.5)}{\draw \p node[scale=0.6] {$\vI$};}
\foreach \p in {(-1.8,4.5),(0.2,4.5),(2.2,4.5),(4.2,4.5),(6.2,4.5)}{\draw \p node[scale=0.6] {$\vI$};}
\foreach \p in {(-1.8,6.5),(0.2,6.5),(2.2,6.5),(4.2,6.5),(6.2,6.5)}{\draw \p node[scale=0.6] {$\vI$};}

\foreach \p in {(-0.8,1.5),(1.2,1.5),(3.2,1.5),(5.2,1.5),(7.2,1.5)}{\draw \p node[scale=0.6] {$\vG$};}
\foreach \p in {(-0.8,3.5),(1.2,3.5),(3.2,3.5),(5.2,3.5),(7.2,3.5)}{\draw \p node[scale=0.6] {$\vG$};}
\foreach \p in {(-0.8,5.5),(1.2,5.5),(3.2,5.5),(5.2,5.5),(7.2,5.5)}{\draw \p node[scale=0.6] {$\vG$};}
\foreach \p in {(-0.8,7.5),(1.2,7.5),(3.2,7.5),(5.2,7.5),(7.2,7.5)}{\draw \p node[scale=0.6] {$\vG$};}

\foreach \p in {(-1.7,1.2),(0.3,1.2),(2.3,1.2),(4.3,1.2),(6.3,1.2)}{\draw \p node[scale=0.6] {$\vC$};}
\foreach \p in {(-1.7,3.2),(0.3,3.2),(2.3,3.2),(4.3,3.2),(6.3,3.2)}{\draw \p node[scale=0.6] {$\vC$};}
\foreach \p in {(-1.7,5.2),(0.3,5.2),(2.3,5.2),(4.3,5.2),(6.3,5.2)}{\draw \p node[scale=0.6] {$\vC$};}
\foreach \p in {(-1.7,7.2),(0.3,7.2),(2.3,7.2),(4.3,7.2),(6.3,7.2)}{\draw \p node[scale=0.6] {$\vC$};}

\foreach \p in {(-0.7,0.8),(1.3,0.8),(3.3,0.8),(5.3,0.8),(7.3,0.8)}{\draw \p node[scale=0.6] {$\vCG$};}
\foreach \p in {(-0.7,2.8),(1.3,2.8),(3.3,2.8),(5.3,2.8),(7.3,2.8)}{\draw \p node[scale=0.6] {$\vCG$};}
\foreach \p in {(-0.7,4.8),(1.3,4.8),(3.3,4.8),(5.3,4.8),(7.3,4.8)}{\draw \p node[scale=0.6] {$\vCG$};}
\foreach \p in {(-0.7,6.8),(1.3,6.8),(3.3,6.8),(5.3,6.8),(7.3,6.8)}{\draw \p node[scale=0.6] {$\vCG$};}


\draw[dashed] (-3,0) -- (8,0);
\draw[dashed] (-2,1) -- (8,1);
\draw[thick] (-3,2) -- (8,2);
\draw[dashed] (-2,3) -- (8,3);
\draw[dashed] (-3,4) -- (8,4);
\draw[dashed] (-2,5) -- (8,5);
\draw[dashed] (-3,6) -- (8,6);
\draw[dashed] (-2,7) -- (8,7);
\draw[dashed] (-3,8) -- (8,8);

\draw (0,2) -- (0,4) -- (2,4) -- (2,2) -- (0,2);




\node[scale=0.7] at (0,3) {$\times$};

\node[scale=0.7] at (-1,-1) {$n_x=-1$};
\node[scale=0.7]  at (1,-1) {$n_x=0$};
\node[scale=0.7]  at (3,-1) {$n_x=1$};
\node[scale=0.7]  at (5,-1) {$n_x=2$};
\node[scale=0.7]  at (7,-1) {$n_x=3$};

\node[scale=0.7] at (-3,1) {$n_y=-1$};
\node[scale=0.7]  at (-3,3) {$n_y=0$};
\node[scale=0.7]  at (-3,5) {$n_y=1$};
\node[scale=0.7]  at (-3,7) {$n_y=2$};

\node[scale=0.7] at (9,0) {$l=-2$};
\node[scale=0.7] at (9,1) {$l=-1$};
\node[scale=0.7] at (9,2) {$l=0$};
\node[scale=0.7] at (9,3) {$l=1$};
\node[scale=0.7] at (9,4) {$l=2$};
\node[scale=0.7] at (9,5) {$l=3$};
\node[scale=0.7] at (9,6) {$l=4$};
\node[scale=0.7] at (9,7) {$l=5$};
\node[scale=0.7] at (9,8) {$l=6$};

\node[scale=0.5] at (0.27,2.27) {$(0;0)$};
\node[scale=0.5] at (1.27,3.26) {$(0;1)$};
\node[scale=0.5] at (2.27,2.27) {$(0;2)$};
\node[scale=0.5] at (3.27,3.26) {$(0;3)$};
\node[scale=0.5] at (0.27,0.27) {$(-1;0)$};
\node[scale=0.5] at (1.27,1.26) {$(-1;1)$};


\end{tikzpicture}
\caption{Black and brown sublattices each corresponds to an embedding of $\pg$ in Euclidean space (origin $\times$). Each unit cell (e.g.\ the black square) contains two black and two brown atomic sites. The 1D edge will be taken to be the glide axis at $l=0$ (solid line).}\label{fig:pglattice}
\end{figure}

The Hilbert space is $\mathscr{H}_{\rm bulk}=l^2(\pg)\oplus l^2(\pg)\cong l^2(\ZZ^2)\oplus l^2(\ZZ^2)\oplus l^2(\ZZ^2)\oplus l^2(\ZZ^2)$, and our convention for the four degrees of freedom per unit cell is such that
\begin{equation*}
{\raisebox{\depth}{\scalebox{1}[-1]{$\Finv$}}}=\begin{pmatrix}1\\0\\0\\0\end{pmatrix},\quad 
\Finv=\begin{pmatrix}0\\1\\0\\0\end{pmatrix},\quad 
{\color{brown}{\raisebox{\depth}{\scalebox{1}[-1]{$\Finv$}}}}=\begin{pmatrix}0\\0\\1\\0\end{pmatrix},\quad 
{\color{brown}{\Finv}}=\begin{pmatrix}0\\0\\0\\1\end{pmatrix}.
\end{equation*}
With respect to the above, there is a sublattice exchange operator
$$\mathsf{X}=\begin{pmatrix} 0 & 0 & 1 & 0 \\ 0 & 0 & 0 & 1 \\ 1 & 0 & 0 & 0 \\ 0 & 1 & 0 & 0\end{pmatrix}$$
which implements $\vI\leftrightarrow\vC$ and $\vG\leftrightarrow\vCG$ within each unit cell. 

On the Fourier transformed space $f\in L^2(\cB)\otimes\CC^4$, the symmetry operators for $\pg$ and $\fS$ are generated by
\begin{equation*}
(\fS\cdot f)(\vect{k})=\begin{pmatrix}1 & 0 & 0& 0 \\ 0 & 1 & 0 & 0\\ 0 & 0 & -1 & 0 \\ 0 & 0 & 0 & -1\end{pmatrix}f(\vect{k}),\qquad (T_y\cdot f)(\vect{k})=e^{\im k_y}f(\vect{k}),
\end{equation*}
\begin{equation*}
(L^{(1)}\cdot f)(\vect{k})=\begin{pmatrix} L^{(1)} & 0 \\ 0 & L^{(1)}\end{pmatrix}f(\vect{k})=\begin{pmatrix} M_VI & 0 \\ 0 & M_VI\end{pmatrix}f(\vect{k}).
\end{equation*}

A Hamiltonian $H$ which anticommutes with $\fS$ and commutes with $\pg$ must, after Fourier transform with respect to $N\subset\pg$, be of the form
$$H(\vect{k})=\begin{pmatrix} 0 & U(\vect{k}) \\ U^\dagger(\vect{k}) & 0\end{pmatrix} $$
where $U$ is a $2\times 2$ matrix-valued function which satisfies 
\begin{equation}
U(\vect{k})V(\vect{k})=V(\vect{k})U(\vect{k}'),\quad
U^\dagger(\vect{k})V(\vect{k})=V(\vect{k})U^\dagger(\vect{k}'),\label{compatibilityconditions}
\end{equation}
where $\vect{k}'=(k_x,-k_y)$. The two conditions in Eq.\ \eqref{compatibilityconditions} actually imply each other because $V^\dagger(\vect{k})=V^{-1}(\vect{k})$. We will also assume that $U$ is continuous. 

We can think of $L^2(\cB)\otimes\CC^4$ as the section space of $\cE_{\rm reg, black}\oplus \cE_{\rm reg, brown}$ and $\mathsf{X}$ as a reference isomorphism identifying $\cE_{\rm reg, brown}\cong\cE_{\rm reg, black}$. Then $M_U$ may be viewed either as an(other) isomorphism between $\cE_{\rm reg, brown}$ and $\cE_{\rm reg, black}$, or as an automorphism of $\cE_{\rm reg}=\cE_{\rm reg, brown}=\cE_{\rm reg, black}$. In the latter point of view, the condition $U(\vect{k})V(\vect{k})=V(\vect{k})U(\vect{k}')$ is the statement that
\begin{equation*}
M_U(L^{(1)})M_U^\dagger\equiv M_U(M_VI)M_U^\dagger=(M_VI)\equiv L^{(1)},
\end{equation*}
i.e.\ $M_U$ commutes with glide reflections about $l=1$. Equivalently,
\begin{equation*}
M_U L^{(0)} M_U^\dagger\equiv M_U(M_{\overline{u_y}V}I)M_U^\dagger=(M_{\overline{u_y}V}I)=L^{(0)},
\end{equation*}
i.e.\ $M_U$ commutes with glide reflections about $l=0$.

The general form of $U(\vect{k})$ satisfying Eq.\ \eqref{compatibilityconditions} is
\begin{equation*}
U(\vect{k})=\begin{pmatrix}
 a(\vect{k}) & {u_x}(\vect{k})b(\vect{k}') \\ b(\vect{k}) & a(\vect{k}')\label{generalform}
\end{pmatrix}
\end{equation*}
for some complex-valued functions. If the Hamiltonian is gapped at 0 (thus invertible), $U(\vect{k})$ must be invertible for all $\vect{k}$. In particular, Eq.\ \eqref{compatibilityconditions} is satisfied by the matrix functions
\begin{equation}
U_r=\begin{pmatrix} 0 & u_x \\  1 & 0\end{pmatrix},\quad U_g=U_r^{-1}=\begin{pmatrix} 0 & 1 \\  \overline{u_x} & 0\end{pmatrix},\quad 
U_p=\begin{pmatrix} u_y & 0 \\ 0 & \overline{u_y}\end{pmatrix},\quad U_b=\begin{pmatrix} 1 & 0 \\  0 & 1\end{pmatrix}\label{compatibleunitaries}
\end{equation}
which give rise to respective compatible gapped Hamiltonians $H_r, H_g, H_p, H_b$. Note that
\begin{equation}
H_b=\begin{pmatrix}
0 & U_b \\ U_b^\dagger & 0
\end{pmatrix}=\mathsf{X},\nonumber
\end{equation}
in accordance with $U_b$ being the identity automorphism of $\cE_{\rm reg}$.

The unitary part of $U$ (in its polar decomposition) implements an isomorphism between the Hilbert subspaces for the two sublattices. Thus it comprises ``hopping terms'' between the black sublattice and the brown sublattice which are compatible with the $\pg$ action. A hopping term (plus its adjoint) between two sites may be represented by a line (with some complex coefficient). It is easy to see that in position space, $H_r$ (resp.\  $H_g$) corresponds to a ``dimerised'' Hamiltonian indicated by the red (resp.\ green) links in Fig.\ \ref{fig:Hrg}, while $H_p$ (resp.\ $H_b$) corresponds to the purple (resp.\ blue) links in Fig.\ \ref{fig:Hpb}.

\begin{figure}
\begin{tikzpicture}[every node/.style={scale=1.5}]
\draw[help lines,dotted,thick,step=2] (-2,0) grid (6,8);

\foreach \p in {(-1.8,0.5),(0.2,0.5),(2.2,0.5),(4.2,0.5)}{\draw \p node[scale=0.6] {$\vI$};}
\foreach \p in {(-1.8,2.5),(0.2,2.5),(2.2,2.5),(4.2,2.5)}{\draw \p node[scale=0.6] {$\vI$};}
\foreach \p in {(-1.8,4.5),(0.2,4.5),(2.2,4.5),(4.2,4.5)}{\draw \p node[scale=0.6] {$\vI$};}
\foreach \p in {(-1.8,6.5),(0.2,6.5),(2.2,6.5),(4.2,6.5)}{\draw \p node[scale=0.6] {$\vI$};}

\foreach \p in {(-0.8,1.5),(1.2,1.5),(3.2,1.5),(5.2,1.5)}{\draw \p node[scale=0.6] {$\vG$};}
\foreach \p in {(-0.8,3.5),(1.2,3.5),(3.2,3.5),(5.2,3.5)}{\draw \p node[scale=0.6] {$\vG$};}
\foreach \p in {(-0.8,5.5),(1.2,5.5),(3.2,5.5),(5.2,5.5)}{\draw \p node[scale=0.6] {$\vG$};}
\foreach \p in {(-0.8,7.5),(1.2,7.5),(3.2,7.5),(5.2,7.5)}{\draw \p node[scale=0.6] {$\vG$};}

\foreach \p in {(-1.7,1.2),(0.3,1.2),(2.3,1.2),(4.3,1.2)}{\draw \p node[scale=0.6] {$\vC$};}
\foreach \p in {(-1.7,3.2),(0.3,3.2),(2.3,3.2),(4.3,3.2)}{\draw \p node[scale=0.6] {$\vC$};}
\foreach \p in {(-1.7,5.2),(0.3,5.2),(2.3,5.2),(4.3,5.2)}{\draw \p node[scale=0.6] {$\vC$};}
\foreach \p in {(-1.7,7.2),(0.3,7.2),(2.3,7.2),(4.3,7.2)}{\draw \p node[scale=0.6] {$\vC$};}

\foreach \p in {(-0.7,0.8),(1.3,0.8),(3.3,0.8),(5.3,0.8)}{\draw \p node[scale=0.6] {$\vCG$};}
\foreach \p in {(-0.7,2.8),(1.3,2.8),(3.3,2.8),(5.3,2.8)}{\draw \p node[scale=0.6] {$\vCG$};}
\foreach \p in {(-0.7,4.8),(1.3,4.8),(3.3,4.8),(5.3,4.8)}{\draw \p node[scale=0.6] {$\vCG$};}
\foreach \p in {(-0.7,6.8),(1.3,6.8),(3.3,6.8),(5.3,6.8)}{\draw \p node[scale=0.6] {$\vCG$};}

\draw[thick,green] (-1.8,0.5) -- (-0.7,0.8);
\draw[thick,green] (0.2,0.5) -- (1.3,0.8);
\draw[thick,green] (2.2,0.5) -- (3.3,0.8);
\draw[thick,green] (4.2,0.5) -- (5.3,0.8);
\draw[thick,green] (-1.8,2.5) -- (-0.7,2.8);
\draw[thick,green] (0.2,2.5) -- (1.3,2.8);
\draw[thick,green] (2.2,2.5) -- (3.3,2.8);
\draw[thick,green] (4.2,2.5) -- (5.3,2.8);
\draw[thick,green] (-1.8,4.5) -- (-0.7,4.8);
\draw[thick,green] (0.2,4.5) -- (1.3,4.8);
\draw[thick,green] (2.2,4.5) -- (3.3,4.8);
\draw[thick,green] (4.2,4.5) -- (5.3,4.8);
\draw[thick,green] (-1.8,6.5) -- (-0.7,6.8);
\draw[thick,green] (0.2,6.5) -- (1.3,6.8);
\draw[thick,green] (2.2,6.5) -- (3.3,6.8);
\draw[thick,green] (4.2,6.5) -- (5.3,6.8);

\draw[thick,green] (-2.8,1.5) -- (-1.7,1.2);
\draw[thick,green] (-0.8,1.5) -- (0.3,1.2);
\draw[thick,green] (1.2,1.5) -- (2.3,1.2);
\draw[thick,green] (3.2,1.5) -- (4.3,1.2);
\draw[thick,green] (5.2,1.5) -- (6.3,1.2);
\draw[thick,green] (-2.8,3.5) -- (-1.7,3.2);
\draw[thick,green] (-0.8,3.5) -- (0.3,3.2);
\draw[thick,green] (1.2,3.5) -- (2.3,3.2);
\draw[thick,green] (3.2,3.5) -- (4.3,3.2);
\draw[thick,green] (5.2,3.5) -- (6.3,3.2);
\draw[thick,green] (-2.8,5.5) -- (-1.7,5.2);
\draw[thick,green] (-0.8,5.5) -- (0.3,5.2);
\draw[thick,green] (1.2,5.5) -- (2.3,5.2);
\draw[thick,green] (3.2,5.5) -- (4.3,5.2);
\draw[thick,green] (5.2,5.5) -- (6.3,5.2);
\draw[thick,green] (-2.8,7.5) -- (-1.7,7.2);
\draw[thick,green] (-0.8,7.5) -- (0.3,7.2);
\draw[thick,green] (1.2,7.5) -- (2.3,7.2);
\draw[thick,green] (3.2,7.5) -- (4.3,7.2);
\draw[thick,green] (5.2,7.5) -- (6.3,7.2);

\draw[thick,red] (-1.8,0.5) -- (-2.7,0.8);
\draw[thick,red] (0.2,0.5) -- (-0.7,0.8);
\draw[thick,red] (2.2,0.5) -- (1.3,0.8);
\draw[thick,red] (4.2,0.5) -- (3.3,0.8);
\draw[thick,red] (6.2,0.5) -- (5.3,0.8);
\draw[thick,red] (-1.8,2.5) -- (-2.7,2.8);
\draw[thick,red] (0.2,2.5) -- (-0.7,2.8);
\draw[thick,red] (2.2,2.5) -- (1.3,2.8);
\draw[thick,red] (4.2,2.5) -- (3.3,2.8);
\draw[thick,red] (6.2,2.5) -- (5.3,2.8);
\draw[thick,red] (-1.8,4.5) -- (-2.7,4.8);
\draw[thick,red] (0.2,4.5) -- (-0.7,4.8);
\draw[thick,red] (2.2,4.5) -- (1.3,4.8);
\draw[thick,red] (4.2,4.5) -- (3.3,4.8);
\draw[thick,red] (6.2,4.5) -- (5.3,4.8);
\draw[thick,red] (-1.8,6.5) -- (-2.7,6.8);
\draw[thick,red] (0.2,6.5) -- (-0.7,6.8);
\draw[thick,red] (2.2,6.5) -- (1.3,6.8);
\draw[thick,red] (4.2,6.5) -- (3.3,6.8);
\draw[thick,red] (6.2,6.5) -- (5.3,6.8);

\draw[thick,red] (-0.8,1.5) -- (-1.7,1.2);
\draw[thick,red] (1.2,1.5) -- (0.3,1.2);
\draw[thick,red] (3.2,1.5) -- (2.3,1.2);
\draw[thick,red] (5.2,1.5) -- (4.3,1.2);
\draw[thick,red] (-0.8,3.5) -- (-1.7,3.2);
\draw[thick,red] (1.2,3.5) -- (0.3,3.2);
\draw[thick,red] (3.2,3.5) -- (2.3,3.2);
\draw[thick,red] (5.2,3.5) -- (4.3,3.2);
\draw[thick,red] (-0.8,5.5) -- (-1.7,5.2);
\draw[thick,red] (1.2,5.5) -- (0.3,5.2);
\draw[thick,red] (3.2,5.5) -- (2.3,5.2);
\draw[thick,red] (5.2,5.5) -- (4.3,5.2);
\draw[thick,red] (-0.8,7.5) -- (-1.7,7.2);
\draw[thick,red] (1.2,7.5) -- (0.3,7.2);
\draw[thick,red] (3.2,7.5) -- (2.3,7.2);
\draw[thick,red] (5.2,7.5) -- (4.3,7.2);

\node[scale=0.7] at (-1,-1) {$n_x=-1$};
\node[scale=0.7]  at (1,-1) {$n_x=0$};
\node[scale=0.7]  at (3,-1) {$n_x=1$};
\node[scale=0.7]  at (5,-1) {$n_x=2$};

\node[scale=0.7] at (-3,1) {$n_y=-1$};
\node[scale=0.7]  at (-3,3) {$n_y=0$};
\node[scale=0.7]  at (-3,5) {$n_y=1$};
\node[scale=0.7]  at (-3,7) {$n_y=2$};

\end{tikzpicture}
\caption{``Dimerised'' Hamiltonian $H_r$ (resp.\  $H_g$), indicated by the red (resp.\ green) links.}\label{fig:Hrg}
\end{figure}
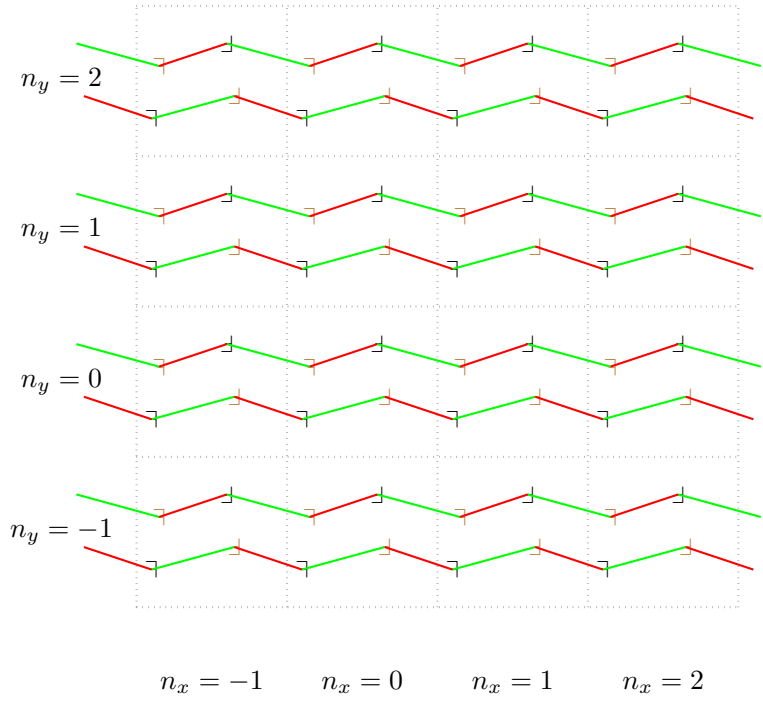

\begin{figure}
\begin{tikzpicture}[every node/.style={scale=1.5}]
\draw[help lines,dotted,thick,step=2] (-2,1) grid (6,7);

\foreach \p in {(-1.8,2.5),(0.2,2.5),(2.2,2.5),(4.2,2.5)}{\draw \p node[scale=0.6] {$\vI$};}
\foreach \p in {(-1.8,4.5),(0.2,4.5),(2.2,4.5),(4.2,4.5)}{\draw \p node[scale=0.6] {$\vI$};}
\foreach \p in {(-1.8,6.5),(0.2,6.5),(2.2,6.5),(4.2,6.5)}{\draw \p node[scale=0.6] {$\vI$};}

\foreach \p in {(-0.8,1.5),(1.2,1.5),(3.2,1.5),(5.2,1.5)}{\draw \p node[scale=0.6] {$\vG$};}
\foreach \p in {(-0.8,3.5),(1.2,3.5),(3.2,3.5),(5.2,3.5)}{\draw \p node[scale=0.6] {$\vG$};}
\foreach \p in {(-0.8,5.5),(1.2,5.5),(3.2,5.5),(5.2,5.5)}{\draw \p node[scale=0.6] {$\vG$};}

\foreach \p in {(-1.7,1.2),(0.3,1.2),(2.3,1.2),(4.3,1.2)}{\draw \p node[scale=0.6] {$\vC$};}
\foreach \p in {(-1.7,3.2),(0.3,3.2),(2.3,3.2),(4.3,3.2)}{\draw \p node[scale=0.6] {$\vC$};}
\foreach \p in {(-1.7,5.2),(0.3,5.2),(2.3,5.2),(4.3,5.2)}{\draw \p node[scale=0.6] {$\vC$};}

\foreach \p in {(-0.7,2.8),(1.3,2.8),(3.3,2.8),(5.3,2.8)}{\draw \p node[scale=0.6] {$\vCG$};}
\foreach \p in {(-0.7,4.8),(1.3,4.8),(3.3,4.8),(5.3,4.8)}{\draw \p node[scale=0.6] {$\vCG$};}
\foreach \p in {(-0.7,6.8),(1.3,6.8),(3.3,6.8),(5.3,6.8)}{\draw \p node[scale=0.6] {$\vCG$};}

\draw[thick,purple] (-1.7,1.2) -- (-1.8,2.5);
\draw[thick,purple] (0.3,1.2) -- (0.2,2.5);
\draw[thick,purple] (2.3,1.2) -- (2.2,2.5);
\draw[thick,purple] (4.3,1.2) -- (4.2,2.5);
\draw[thick,purple] (-1.7,3.2) -- (-1.8,4.5);
\draw[thick,purple] (0.3,3.2) -- (0.2,4.5);
\draw[thick,purple] (2.3,3.2) -- (2.2,4.5);
\draw[thick,purple] (4.3,3.2) -- (4.2,4.5);
\draw[thick,purple] (-1.7,5.2) -- (-1.8,6.5);
\draw[thick,purple] (0.3,5.2) -- (0.2,6.5);
\draw[thick,purple] (2.3,5.2) -- (2.2,6.5);
\draw[thick,purple] (4.3,5.2) -- (4.2,6.5);

\draw[thick,purple] (-0.8,1.5) -- (-0.7,2.8);
\draw[thick,purple] (1.2,1.5) -- (1.3,2.8);
\draw[thick,purple] (3.2,1.5) -- (3.3,2.8);
\draw[thick,purple] (5.2,1.5) -- (5.3,2.8);
\draw[thick,purple] (-0.8,3.5) -- (-0.7,4.8);
\draw[thick,purple] (1.2,3.5) -- (1.3,4.8);
\draw[thick,purple] (3.2,3.5) -- (3.3,4.8);
\draw[thick,purple] (5.2,3.5) -- (5.3,4.8);
\draw[thick,purple] (-0.8,5.5) -- (-0.7,6.8);
\draw[thick,purple] (1.2,5.5) -- (1.3,6.8);
\draw[thick,purple] (3.2,5.5) -- (3.3,6.8);
\draw[thick,purple] (5.2,5.5) -- (5.3,6.8);

\draw[thick,blue] (-1.7,1.2) -- (-1.8,0.5);
\draw[thick,blue] (0.3,1.2) -- (0.2,0.5);
\draw[thick,blue] (2.3,1.2) -- (2.2,0.5);
\draw[thick,blue] (4.3,1.2) -- (4.2,0.5);
\draw[thick,blue] (-1.7,3.2) -- (-1.8,2.5);
\draw[thick,blue] (0.3,3.2) -- (0.2,2.5);
\draw[thick,blue] (2.3,3.2) -- (2.2,2.5);
\draw[thick,blue] (4.3,3.2) -- (4.2,2.5);
\draw[thick,blue] (-1.7,5.2) -- (-1.8,4.5);
\draw[thick,blue] (0.3,5.2) -- (0.2,4.5);
\draw[thick,blue] (2.3,5.2) -- (2.2,4.5);
\draw[thick,blue] (4.3,5.2) -- (4.2,4.5);
\draw[thick,blue] (-1.7,7.2) -- (-1.8,6.5);
\draw[thick,blue] (0.3,7.2) -- (0.2,6.5);
\draw[thick,blue] (2.3,7.2) -- (2.2,6.5);
\draw[thick,blue] (4.3,7.2) -- (4.2,6.5);

\draw[thick,blue] (-0.8,1.5) -- (-0.7,0.8);
\draw[thick,blue] (1.2,1.5) -- (1.3,0.8);
\draw[thick,blue] (3.2,1.5) -- (3.3,0.8);
\draw[thick,blue] (5.2,1.5) -- (5.3,0.8);
\draw[thick,blue] (-0.8,3.5) -- (-0.7,2.8);
\draw[thick,blue] (1.2,3.5) -- (1.3,2.8);
\draw[thick,blue] (3.2,3.5) -- (3.3,2.8);
\draw[thick,blue] (5.2,3.5) -- (5.3,2.8);
\draw[thick,blue] (-0.8,5.5) -- (-0.7,4.8);
\draw[thick,blue] (1.2,5.5) -- (1.3,4.8);
\draw[thick,blue] (3.2,5.5) -- (3.3,4.8);
\draw[thick,blue] (5.2,5.5) -- (5.3,4.8);

\node[scale=0.7] at (-1,0) {$n_x=-1$};
\node[scale=0.7]  at (1,0) {$n_x=0$};
\node[scale=0.7]  at (3,0) {$n_x=1$};
\node[scale=0.7]  at (5,0) {$n_x=2$};

\node[scale=0.7] at (-3,1) {$n_y=-1$};
\node[scale=0.7]  at (-3,3) {$n_y=0$};
\node[scale=0.7]  at (-3,5) {$n_y=1$};
\node[scale=0.7]  at (-3,7) {$n_y=2$};

\end{tikzpicture}
\caption{``Dimerised'' $H_p$ (resp.\  $H_b$), indicated by the purple (resp.\ blue) links.}\label{fig:Hpb}
\end{figure}
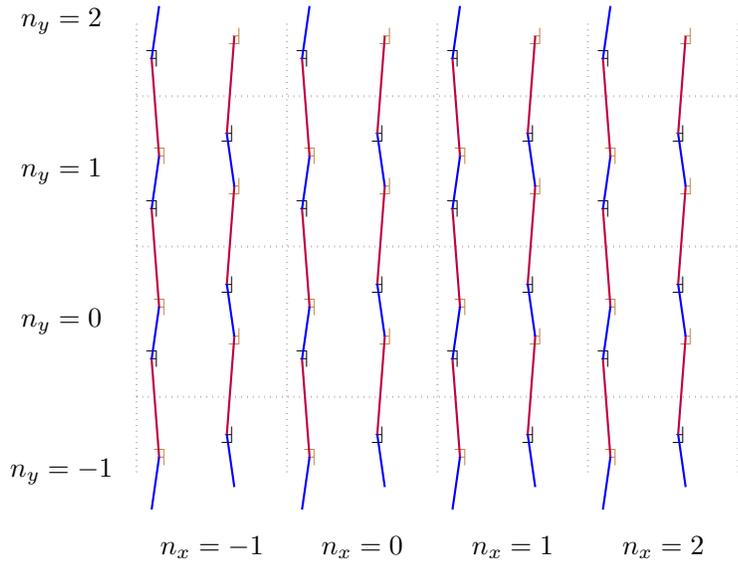

\section{Topological zero modes on 1D edge: heuristics}\label{sec:1Dedge}
\subsection{Integer index for zero modes transverse to glide axis}
Suppose we truncate the model to the half-plane with $n_x\geq 0$ so that $H\mapsto\breve{H}$, then $\breve{H}_p$, $\breve{H}_r$ remain fully dimerised, corresponding to  Wind${}_x$($U_p$)$=$ Wind${}_x$($U_b$)$=0$ where Wind${}_x$ refers to the winding number around $k_x$. However, $\breve{H}_r$ (resp.\ $\breve{H}_g$) acquires a black (resp.\ brown) unpaired zero mode, for each $n_y\in \ZZ$, which corresponds to Wind${}_x$($U_r$)$=+1$ (resp.\ Wind${}_x$($U_g$)$=-1$). This bulk-edge correspondence of integer indices is discussed in more detail in Appendix \ref{appendix:integerbec}.

Clearly, Wind${}_y$ of $U_r, U_g, U_b$ are all zero, and also Wind${}_y$($U_p$)$=0$ due to cancelling contributions from $u_y$ and $\overline{u_y}$, so we could not hope to detect non-triviality of $H_p$ by the ordinary winding number of $U$ along $k_y$. As we will see, truncating $H_p$ along a horizontal glide axis causes it to acquire an interesting collection of edge zero modes which ``detects'' a certain non-trivial bulk invariant of $H_p$. 

\subsection{Mod 2 index for zero modes along glide axis}\label{sec:edgezeroesheuristic}
Consider the horizontal 1D glide axis $l=0$ as an edge in our model, and truncate along this edge in the sense of killing any hopping terms that go across the edge. Thus the upper and lower half-planes are separated. Na\"{i}vely, it appears that only the subgroup generated by $T_x$ remain symmetries of the truncated model, and so the edge zero modes should be characterised by a topological invariant associated to 1D systems with only $\ZZ\cong\ZZ_x$ translation symmetry. It turns out that the appropriate edge topological invariant is more subtle and associated to a ``graded edge symmetry group'', and that the invariant is a mod 2 number rather than an integer. Before explaining this claim, we analyse the edge zero modes for the basic dimerised Hamiltonians $H_p, H_r, H_g, H_b$ in order to gain a heuristic understanding of the general case.

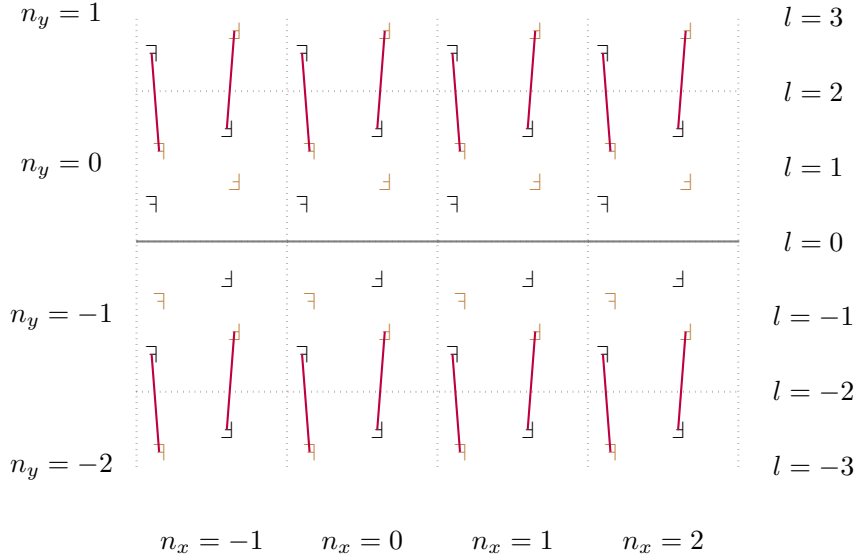
\begin{figure}
\begin{tikzpicture}[every node/.style={scale=1.5}]
\draw[help lines,dotted,thick,step=2] (-2,-3) grid (6,3);
\draw[help lines,thick,step=2] (-2,0) grid (6,0);

\foreach \p in {(-1.8,0.5),(0.2,0.5),(2.2,0.5),(4.2,0.5)}{\draw \p node[scale=0.6] {$\vI$};}
\foreach \p in {(-1.8,2.5),(0.2,2.5),(2.2,2.5),(4.2,2.5)}{\draw \p node[scale=0.6] {$\vI$};}
\foreach \p in {(-1.8,-1.5),(0.2,-1.5),(2.2,-1.5),(4.2,-1.5)}{\draw \p node[scale=0.6] {$\vI$};}

\foreach \p in {(-0.8,1.5),(1.2,1.5),(3.2,1.5),(5.2,1.5)}{\draw \p node[scale=0.6] {$\vG$};}
\foreach \p in {(-0.8,-0.5),(1.2,-0.5),(3.2,-0.5),(5.2,-0.5)}{\draw \p node[scale=0.6] {$\vG$};}
\foreach \p in {(-0.8,-2.5),(1.2,-2.5),(3.2,-2.5),(5.2,-2.5)}{\draw \p node[scale=0.6] {$\vG$};}

\foreach \p in {(-1.7,1.2),(0.3,1.2),(2.3,1.2),(4.3,1.2)}{\draw \p node[scale=0.6] {$\vC$};}
\foreach \p in {(-1.7,-0.8),(0.3,-0.8),(2.3,-0.8),(4.3,-0.8)}{\draw \p node[scale=0.6] {$\vC$};}
\foreach \p in {(-1.7,-2.8),(0.3,-2.8),(2.3,-2.8),(4.3,-2.8)}{\draw \p node[scale=0.6] {$\vC$};}

\foreach \p in {(-0.7,0.8),(1.3,0.8),(3.3,0.8),(5.3,0.8)}{\draw \p node[scale=0.6] {$\vCG$};}
\foreach \p in {(-0.7,2.8),(1.3,2.8),(3.3,2.8),(5.3,2.8)}{\draw \p node[scale=0.6] {$\vCG$};}
\foreach \p in {(-0.7,-1.2),(1.3,-1.2),(3.3,-1.2),(5.3,-1.2)}{\draw \p node[scale=0.6] {$\vCG$};}

\draw[thick,purple] (-1.7,1.2) -- (-1.8,2.5);
\draw[thick,purple] (0.3,1.2) -- (0.2,2.5);
\draw[thick,purple] (2.3,1.2) -- (2.2,2.5);
\draw[thick,purple] (4.3,1.2) -- (4.2,2.5);
\draw[thick,purple] (-1.7,-2.8) -- (-1.8,-1.5);
\draw[thick,purple] (0.3,-2.8) -- (0.2,-1.5);
\draw[thick,purple] (2.3,-2.8) -- (2.2,-1.5);
\draw[thick,purple] (4.3,-2.8) -- (4.2,-1.5);

\draw[thick,purple] (-0.8,1.5) -- (-0.7,2.8);
\draw[thick,purple] (1.2,1.5) -- (1.3,2.8);
\draw[thick,purple] (3.2,1.5) -- (3.3,2.8);
\draw[thick,purple] (5.2,1.5) -- (5.3,2.8);
\draw[thick,purple] (-0.8,-2.5) -- (-0.7,-1.2);
\draw[thick,purple] (1.2,-2.5) -- (1.3,-1.2);
\draw[thick,purple] (3.2,-2.5) -- (3.3,-1.2);
\draw[thick,purple] (5.2,-2.5) -- (5.3,-1.2);

\node[scale=0.7] at (-1,-4) {$n_x=-1$};
\node[scale=0.7]  at (1,-4) {$n_x=0$};
\node[scale=0.7]  at (3,-4) {$n_x=1$};
\node[scale=0.7]  at (5,-4) {$n_x=2$};

\node[scale=0.7]  at (-3,1) {$n_y=0$};
\node[scale=0.7]  at (-3,3) {$n_y=1$};
\node[scale=0.7]  at (-3,-1) {$n_y=-1$};
\node[scale=0.7]  at (-3,-3) {$n_y=-2$};

\node[scale=0.7]  at (7,-3) {$l=-3$};
\node[scale=0.7]  at (7,-2) {$l=-2$};
\node[scale=0.7]  at (7,-1) {$l=-1$};
\node[scale=0.7]  at (7,0) {$l=0$};
\node[scale=0.7]  at (7,1) {$l=1$};
\node[scale=0.7]  at (7,2) {$l=2$};
\node[scale=0.7]  at (7,3) {$l=3$};

\end{tikzpicture}
\caption{Unpaired chain of zero edge modes left behind when the purple bonds of Fig.\ \ref{fig:Hpb} passing through $n_y=0$ are intercepted by the edge glide axis.}\label{fig:horizontaledge}
\end{figure}

Clearly, for Hamiltonians $H_r, H_g, H_b$ there is no dangling edge mode as no bond is cut by the edge. For $H_p$, however, an entire Hilbert subspace along and just above $l=0$, namely $\mathscr{H}_e^{\rm upper}=l^2(\ZZ_{\rm black}^{\rm upper})\oplus l^2(\ZZ_{\rm brown}^{\rm upper})$, is left dangling, together with its counterpart just below $l=0$, namely $\mathscr{H}_e^{\rm lower}=l^2(\ZZ_{\rm black}^{\rm lower})\oplus l^2(\ZZ_{\rm brown}^{\rm lower})$ (the `$e$' subscript stands for `edge'), see Fig.\ \ref{fig:horizontaledge}.

The truncated Hamiltonian $\breve{H}=\breve{H}^{\rm upper}\oplus\breve{H}^{\rm lower}$ (acting on the direct sum of upper and lower half-space Hilbert subspaces) is consistent not just with $T_x$ symmetry, but also with $L^{(0)}$ glide symmetry, although $T_y$ symmetry is broken. The edge zero mode Hilbert space $\mathscr{H}_e^{\rm upper}\oplus\mathscr{H}_e^{\rm lower}$ is naturally $\ZZ_2$-graded by the upper/lower label\footnote{Not to be confused with the black/brown sublattice grading by the operator $\fS$.}, and is an invariant subspace for the action of $L^{(0)}$ which is furthermore \emph{odd} with respect to the upper/lower grading. The black edge subspace $\mathscr{H}_{e, \rm black}=\mathscr{H}_{e, \rm black}^{\rm upper}\oplus\mathscr{H}_{e, \rm black}^{\rm lower}$ is also invariant under $L^{(0)}$, as is the brown edge subspace $\mathscr{H}_{e, \rm brown}$, and these are intertwined by $\mathsf{X}$. We see that $L^{(0)}$ generates an unusual ``graded symmetry group'', or ``nonsymmorphic chiral symmetry'', explained below.

\subsubsection{Graded edge symmetry group}
Quite generally, under a homotopy of $U_p$ respecting $\pg$ symmetry, the truncated operator $\breve{H}_p^{\rm upper}\oplus\breve{H}_p^{\rm lower}$ is no longer symmetric with respect to vertical translation $T_y$, but nevertheless remains symmetric with respect to glide reflection $L^{(0)}$ (and also $T_x=(L^{(0)})^2$). Thus we say that the truncated bulk-with-edge system is compatible with a $\ZZ$ symmetry generated by $L^{(0)}$. However, the upper/lower grading is part of the data, and only the ``even'' subgroup generated by $T_x=(L^{(0)})^2$ preserves this grading whereas $L^{(0)}$ reverses it. Note that the other glide reflections $L^{(l)}, l\neq 0$ are \emph{not} symmetries of the truncated operators (irrespective of the grading).

\begin{figure}
\centering
\begin{tikzpicture}[every node/.style={scale=1.5}]
\draw[help lines,thick,step=2] (-2,0) grid (6,0);

\foreach \p in {(-1.8,0.2),(0.2,0.2),(2.2,0.2),(4.2,0.2)}{\draw \p node[scale=0.6] {$\bullet$};}

\foreach \p in {(-0.8,-0.2),(1.2,-0.2),(3.2,-0.2),(5.2,-0.2)}{\draw \p node[scale=0.6] {$\bullet$};}

\end{tikzpicture}
\caption{The generating translation along a glide axis also effects an exchange of the internal label ``above/below the axis''.}\label{fig:glideaxis}
\end{figure}
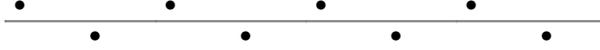

Therefore, when there is an edge at $l=0$, we should be considering the \emph{graded} symmetry group $\ZZ[L]$ with \emph{odd} generator $L$, which is to be represented in a graded sense (e.g.\ as the glide reflection operator $L^{(0)}$ above). Specifically, we have the following non-split exact sequence
\begin{equation} 
1\longrightarrow \ZZ[T_x]\overset{\times 2}{\longrightarrow}\ZZ[L]\overset{(-1)^{(\cdot)}}{\longrightarrow} \ZZ_2\longrightarrow 1.\label{1Dnonsymmorphicgroup}
\end{equation}
A graded representation of $\ZZ[L]$ is a \emph{graded} Hilbert space $\mathscr{H}_+\oplus\mathscr{H}_-$ with grading operator $\epsilon=\epsilon^\dagger$, $\epsilon^2=1$, such that $L\epsilon=-\epsilon L$ and $L^2=T_x$ (retaining the same symbols for the abstract group elements and their operator representatives).

In \cite{SSG1}, the sequence Eq.\ \eqref{1Dnonsymmorphicgroup} was studied, and the grading operator $\epsilon$ has the interpretation of a spectrally flattened gapped 1D Hamiltonian which anticommutes with the so-called ``nonsymmorphic chiral symmetry'' $L$. We do not interpret 
$\epsilon$ as a 1D Hamiltonian in this paper but as a grading operator on the edge mode Hilbert space.

It was found in \cite{SSG1} that there is a $\ZZ/2$ topological classification of Hamiltonians with nonsymmorphic chiral symmetry, using the language of twisted $K$-theory, which we will employ later. In our context, this suggests a $\ZZ/2$ classification of the edge modes for truncated $\pg$-symmetric Hamiltonians, and we shall make this precise via an index theorem.

\section{Twisted $K$-theory indices in the bulk}\label{sec:twistedindices}
In the subsequent sections, we define the bulk and edge topological indices in the twisted $K$-theory of bulk and edge Brillouin tori, and link them via a twisted version of the Toeplitz index theorem. Thus an analytic index of $\breve{H}_{\rm bulk}$ (a certain counting of the zero modes on an edge) will be identified with a $K$-theoretic index for the edge Brillouin torus (a circle), obtained by ``integrating'' over the transverse momenta. Unlike the situation in the SSH model, the relevant indices are much more subtle mod 2 numbers.

\subsection{Bloch vector bundle and $K$-theory}
Let us recall the construction of $K$-theoretic indices for topological phases with space group symmetries \cite{FM,Thiang}, which generalises the ideas put forth in \cite{Kitaev} for the case where the point group is trivial. The basic idea is that under the Bloch--Floquet transform with respect to a lattice $\ZZ^d$ of translation symmetries, the valence states of an insulator form a Hermitian vector bundle $\cE$ over the Brillouin $d$-torus $\cB$, and $\cE$ may be non-trivial as a complex vector bundle. The Chern insulator \cite{Haldane} is an example where $\cE$ has non-trivial Chern number, and leads to gapless edge states on a 1D edge. More abstractly, $\cB$ is the Pontryagin dual of irreducible characters for $\ZZ^d$. Every character $\chi\in\cB$ has the form $\chi_{\vect{k}}:\vect{n}\mapsto e^{\im \vect{n}\cdot\vect{k}}$ for some $\vect{k}\in [0,2\pi]^d/_{0\equiv 2\pi}$.

For a Hamiltonian invariant under the action of a space group $\sG$, the valence states must also host an action of $\sG$. In the decomposition of $\sG$ in Eq.\ \eqref{spacegroupsequence}, the translation part $N=\ZZ^d$ acts as multiplication by $\chi(\vect{n})$, as usual under the Fourier transform. Each $a\in F$ acts on the normal abelian subgroup $\ZZ^d$ by conjugation by any lift $\breve{a}\in\sG$, and there is a canonical dual action of $F$ on $\cB$, defined by 
\begin{equation*}
(a\cdot\chi)(\vect{n})=\chi(a^{-1}\cdot\vect{n}).\label{dualaction}
\end{equation*}
We may also regard Eq.\ \eqref{dualaction} as giving $C(\cB)$ the structure of a left $F$-module, with action $\alpha$ by precomposition,
\begin{equation*}
(\alpha(a)f)(\chi)=f(a^{-1}\cdot\chi),\qquad f\in C(\cB).
\end{equation*}

If $\sG$ is symmorphic, $F$ may be regarded as a subgroup of $\sG$, and there should be a linear action of $F$ on the sections $\Gamma(\cE)$ of the valence bundle $\cE$. Note that $\Gamma(\cE)$ has an action of $\vect{n}\in\ZZ^d$ via pointwise multiplication by the Fourier transformed function $\wh{\vect{n}}:\chi\mapsto \chi(\vect{n})$.
Because of Eq.\ \eqref{dualaction}, the action of $F$ on $\Gamma(\cE)$ is not fibrewise, but is instead a lift of the dual $F$-action on the base $\cB$ --- in other words, $\cE$ is an $F$-equivariant vector bundle over the $F$-space $\cB$. This ensures that the $F$ and $\ZZ^d$ actions on $\Gamma(\cE)$ together furnish an action of $\sG=\ZZ^d\rtimes F$.

\subsection{Twisted vector bundles from nonsymmorphic space groups}\label{sec:twistedvbnonsymmorphic}
In the general nonsymmorphic case, there is a group 2-cocycle $\nu:F\times F\rightarrow \ZZ^d$ defined by
\begin{equation*}
\nu(a,b)(\breve{ab})=\breve{a}\,\breve{b},\qquad a,b,\in F
\end{equation*}
where $\breve{(\cdot)}$ is a lift of $(\cdot)$ in $\sG$. Thus $\nu$ measures the failure of the lifting map $a\mapsto \breve{a}$ to be a homomorphism. By associativity in $\sG$, the map $\nu$ satisfies a cocycle condition
\begin{equation*}
\nu(a,b)+\nu(ab,c)=a\cdot\nu(b,c)+\nu(a,bc).
\end{equation*}
A different choice of lifting map modifies $\nu$, but only in such a way as to maintain its group cohomology class $[\nu]\in H^2(F,\ZZ^d)$. The Fourier transform of $\nu$ defines a $C(\cB,{\rm U}(1))$-valued 2-cocycle by the formula
\begin{equation}
\wh{\nu}(a,b)(\chi)=(ab\cdot\chi)(\nu(a,b))\in{\rm U}(1).\label{dualcocycle}
\end{equation}

The interpretation of $\wh{\nu}$ is that it measures the failure of the linear action of $a,b\in F$ on $\Gamma(\cE)$ to compose homomorphically. In more detail, if 
$$
\begin{tikzcd}
\chi\ar[r,"b"]\ar[rr,out=-30,in=210,swap,"ab"] & b\cdot\chi\ar[r,"a"] & (ab)\cdot\chi
\end{tikzcd}
$$
on the base, then for composable linear maps
$$
\begin{tikzcd}
\cE_\chi\ar[r,"L_b"]\ar[rr,out=-30,in=210,swap,"L_{ab}"] & \cE_{b\cdot \chi}\ar[r,"L_a"] & \cE_{(ab)\cdot\chi}
\end{tikzcd}
$$
a phase factor is acquired,
\begin{equation}
L_a L_b=\wh{\nu}(a,b)((ab)\cdot\chi) L_{ab}.\label{phaseassignments}
\end{equation}
In particular, when $\chi$ is a fixed point for the $F$-action on $\cB$, there is a \emph{projective} representation of $F$ on the fibre $\cE_\chi$. The assignment of phases in Eq.\ \eqref{phaseassignments} globally is needed for there to be a genuine action of $\sG$ on $\Gamma(\cE)$, and has important physical consequences. For instance, if there is a whole cycle of fixed points in $\cB$, the variation of the phase factors $\wh{\nu}(a,b)((ab)\cdot\chi)$ as $\chi$ varies along the cycle can lead to ``monodromy of representations''. This phenomenon was pointed out in the context of band theory of solid state physics as a kind of degeneracy and connectivity of energy bands enforced by non-symmorphic space group symmetry \cite{MZ}.

\begin{example}
For $\pg$, the circle $k_y=0$ (and also $k_y=\pi$) is a set of fixed points for the dual action of the point group $F=\ZZ_2$ on $\cB$. The fibre over such a point hosts a projective representation of $F$ with $(-1)\in F$ represented by an operator $L_{k_x}$ such that $L_{k_x}^2=e^{\im k_x}$, so on one-dimensional invariant subspaces we can have $L_{k_x}=\pm e^{\im k_x/2}$. If we start from the trivial (ordinary) representation at $(k_x,k_y)=(0,0)$ and move along fixed points by increasing $k_x$, only the option $L_{k_x}=+e^{\im k_x/2}$ satisfies continuity with respect to $k_x$, but this means that we end up at the \emph{sign} representation $L_{2\pi}=e^{\im\pi}=-1$ when we finally reach $k_x=2\pi\sim 0$. Thus a ``doubling'' of dimension is enforced by the family of projective phase factors coming from nonsymmorphicity of $\pg$.
\end{example}

Another way of thinking about $\cE$ is that there is a transformation groupoid $\cB//F$ associated to the $F$-action on $\cB$, centrally extended by $\wh{\nu}$, and $\cE$ is a ``linearisation'' or a $\wh{\nu}$-twisted equivariant vector bundle, cf.\ \cite{FM,Kubota}. In the symmorphic case, $\cE$ is an ordinary linearisation of the transformation groupoid, i.e.\ an equivariant vector bundle over $\cB$, or a bundle over the groupoid $\cB//F$.

Given an action of $F$ on $\cB$ and an equivariant twist $\widehat{\nu}$ as above, a $\widehat{\nu}$-twisted $F$-equivariant vector bundle on $\mathcal{B}$ (a twisted bundle for short) can be defined as a complex vector bundle $\pi : \mathcal{E} \to \mathcal{B}$ equipped with a $\widehat{\nu}$-twisted $F$-action, namely, a set of vector bundle maps $\{ L_a \}_{a \in F}$,
$$
\begin{CD}
\mathcal{E} @>{L_a}>> \mathcal{E} \\
@V{\pi}VV @VV{\pi}V \\
\mathcal{B} @>{a}>> \mathcal{B},
\end{CD}
$$
which satisfy 
$$
L_aL_b = \widehat{\nu}(a, b) L_{ab}.
$$
When we make the $\chi$-dependence explicit, the above formula is 
$$L_aL_b\psi = \widehat{\nu}(a, b)(ab \cdot \chi) L_{ab} \psi$$
for $a, b \in F$, $\chi \in \mathcal{B}$ and $\psi \in \mathcal{E}|_{\chi} = \pi^{-1}(\chi)$, in accordance with Eq.\ \eqref{phaseassignments}.

As with ordinary vector bundles, there is a category of finite-rank $\wh{\nu}$-twisted $F$-equivariant vector bundles, and the Grothendieck group of classes of such bundles gives the finite-rank twisted equivariant $K$-theory group $K^{0+\wh{\nu}}_F(\cB)_{\rm fin}$, see Appendix E of \cite{FM} and \cite{Karoubi2, Gomi2}.

\begin{remark}
For $K$-theory twisted by a general element of $H^3(\cB,\ZZ)$, an infinite-dimensional model, e.g.\ using Fredholm operators, is needed, and the twisted $K$-theory classes are not necessarily realised by a finite-rank model. In our equivariant setting, the 2-cocycle $\wh{\nu}$ above defines a \emph{torsion} class in $H^2(F,C(\cB,{\rm U}(1)))$, which can be understood as an element of $H^3_F(\cB,\ZZ)$, i.e.\ an \emph{equivariant twist}, cf.\ \cite{Gomi1}. Non-equivariantly, this twist is trivial. This nature of $\widehat{\nu}$ ensures that the finite-rank model $K^{0 + \widehat{\nu}}_F(\mathcal{B})_{\mathrm{fin}}$ agrees with the ``true'' $K$-theory $K^{0 + \widehat{\nu}}_F(\mathcal{B})$, see (\cite{FM,Gomi2}).
\end{remark}

\subsection{Classification of twisted vector bundles for $\pg$ symmetry}
For $\sG=\pg$, the action of $(-1)\in \ZZ_2=F$ on $N$ takes $(n_y,n_x)\mapsto (-n_y,n_x)$, so the dual action on $\cB$ takes $(k_x,k_y)\mapsto(k_x,-k_y)$ (writing $\vect{k}$ for the character $\chi_{\vect{k}}$). Taking $(0;1)\in\pg$ as a lift of $(-1)\in \ZZ_2$, the group 2-cocycle $\nu$ is simply $\nu(-1,-1)=(0,1)\in N$, and trivial otherwise, so $\wh{\nu}(-1,-1)(\vect{k})=e^{\im k_x}=u_x(\vect{k})$. Thus $\wh{\nu}$-twisted vector bundles over the $\ZZ_2$-space $\cB$ are complex bundles $\cE\rightarrow\cB$ equipped with a twisted $\ZZ_2$-action, the latter being specified by a bundle map $L_{-1}$ lifting the action of $(-1)\in\ZZ_2$ on $\cB$, which squares to the multiplication operator $M_{u_x}$. Such bundles are classified $K$-theoretically by the twisted $K$-theory group $K^{0 + \widehat{\nu}}_{\Z_2}(\mathcal{B})$.

The computation of $K^{\bullet + \widehat{\nu}}_{\Z_2}(\mathcal{B}), \bullet=0,1$, is detailed in \cite{SSG2} (page 25--30): In that paper, we first compute $K^\bullet_{\Z_2}(S^1)$, where the circle $S^1 = \R/2\pi\Z = \mathcal{B}_y$ is given the flip $\Z_2$-action $k_y \mapsto -k_y$. This computation is done by applying the Mayer--Vietoris sequence to the decomposition of $S^1$ into two intervals. (This $K$-theory $K^\bullet_{\Z_2}(S^1)$ is relatively well known. Its $R(\Z_2)$-module structure is for example determined in \cite{MD-R}.) We then compute $K^{\bullet + \widehat{\nu}}_{\Z_2}(\mathcal{B})$ of our interests. This is again done by applying the Mayer--Vietoris sequence to a decomposition of the $2$-dimensional torus $\mathcal{B}$, where knowledge about the $R(\Z_2$)-module structure on $K^0_{\Z_2}(S^1)$ is needed; see Appendix \ref{appendix:computations} for details.

The result of this computation is that $K^{0 + \widehat{\nu}}_{\Z_2}(\mathcal{B}) \cong \Z$. In fact, $K^{0 + \widehat{\nu}}_{\Z_2}(\mathcal{B}) = K(\mathrm{Vect}^{\widehat{\nu}}_{\Z_2}(\mathcal{B}))$ can be realised by applying the Grothendieck construction to the monoid $\mathrm{Vect}^{\widehat{\nu}}_{\Z_2}(\mathcal{B})$ of isomorphism classes of $\widehat{\nu}$-twisted $\Z_2$-equivariant vector bundles of finite rank (see Section \ref{sec:finiterankK}). Given such a vector bundle $\mathcal{E} \to \mathcal{B}$, we can show that:

\begin{lem}\label{lem:trivialregularbundle}
$\cE$ is trivialisable as a complex bundle, and $\mathrm{dim}\, \mathcal{E} \in 2\Z$.
\end{lem}

\begin{proof}
Because the $\Z_2$-action on $\mathcal{B}$ is orientation reversing, the (first) Chern class of the underlying bundle of $\mathcal{E}$ is trivial. This allows us to assume that $\mathcal{E} = \mathcal{B} \times \C^{r}$ and that its twisted $\Z_2$-action is of the form
$$
L_{-1}:((k_x, k_y), \vec{\psi}) \mapsto
((k_x, -k_y), W(k_x, k_y) \vec{\psi}),
$$
where $W : \mathcal{B} \to {\mathrm U}(r)$ is a continuous map satisfying
\begin{equation}
W(k_x, -k_y)W(k_x, k_y) = e^{\im k_x} 1_{\C^r}.\label{Wcondition}
\end{equation}
This relation leads to
$$
\deg \det W + \deg \det W = r,
$$
where $\deg \det W$ is the the winding number (degree) of $\det W(\cdot, 0) : \mathcal{B}_x \to U(1)$ around the $k_x$-circle $\mathcal{B}_x$ at $k_y = 0$.
\end{proof}
As a consequence of Lemma \ref{lem:trivialregularbundle}, we get a monoid homomorphism $\frac{1}{2}\mathrm{dim} : \mathrm{Vect}^{\widehat{\nu}}_{\Z_2}(\mathcal{B}) \to \Z$. By the universality of the Grothendieck construction, this monoid homomorphism extends to a group homomorphism $\frac{1}{2}\mathrm{dim} : K^{0 + \widehat{\nu}}_{\Z_2}(\mathcal{B}) \to \Z$. Now, the ``regular vector bundle'' $\mathcal{E}_{\mathrm{reg}}$ is the product bundle $\mathcal{E}_{\mathrm{reg}} = \mathcal{B} \times \C^2$ of rank $2$ with the twisted $\Z_2$-action
\begin{equation}
L^{(1)}:((k_x, k_y), \vec{\psi}) \mapsto
((k_x, -k_y), 
\left(
\begin{array}{cc}
0 & e^{\im k_x} \\
1 & 0
\end{array}\right)
\vec{\psi})\equiv 
((k_x, -k_y),V(k_x)
\vec{\psi}),\label{regularbundle}
\end{equation}
where $V:\cB_x\rightarrow {\mathrm U}(2)$ is the unitary matrix-valued function $\begin{pmatrix}0 & u_x \\ 1 & 0\end{pmatrix}$ specifying $W$ via $W(k_x,k_y)=V(k_x)$. As a result, the homomorphism $\frac{1}{2}\mathrm{dim} : \ZZ\cong K^{0 + \widehat{\nu}}_{\Z_2}(\mathcal{B}) \to \Z$ is \emph{surjective}, and hence is bijective. 

\begin{corollary}\label{cor:regularbundle}
$\mathcal{E}_{\mathrm{reg}}$ represents a generator of $K^{0 + \widehat{\nu}}_{\Z_2}(\mathcal{B}) \cong \Z$. 
\end{corollary}

\begin{remark}\label{rem:choiceofglide}
We could also endow $\cE_{\rm reg}$ with the alternative twisted $\ZZ_2$-action
$$
L_{-1}=L^{(0)}:((k_x, k_y), \vec{\psi}) \mapsto
((k_x, -k_y), 
\left(
\begin{array}{cc}
0 & e^{\im(k_x-k_y)} \\
e^{-\im k_y} & 0
\end{array}
\right)
\vec{\psi}),
$$
in which $W$ is instead the function $V'=\begin{pmatrix}0 & u_x\overline{u_y} \\ \overline{u_y} & 0\end{pmatrix}=\overline{u_y}V$, which also satisfies Eq.\ \eqref{Wcondition}. In relation to the construction of Section \ref{sec:pgtight}, $L^{(1)}$ and $L^{(0)}$ are the glide reflection operators with glides axes at $l=1$ and at $l=0$ respectively. We may verify that the twisted actions $L^{(1)}$ and $L^{(0)}$ are intertwined by the multiplication operator $\begin{pmatrix}0 & u_x \\ \overline{u_y} & 0\end{pmatrix}$, which corresponds to changing the unit cell such that the glide reflection axis remains in the middle of the unit cell, see Fig.\ \ref{fig:pgunitcell}. Thus $(\cE_{\rm reg}, L^{(1)})$ and $(\cE_{\rm reg}, L^{(0)})$ define isomorphic twisted vector bundles. The choice $L^{(1)}$ is convenient because its $V$ depends only on $k_x\in\cB_x$, but the choice $L^{(0)}$ with its $V'$ is also needed for the analysis of the bulk-edge correspondence with edge at $l=0$. \end{remark}

\subsection{Twisted bundles from $\pg$ and chiral symmetry: Klein bottle phase}
The notion of twisted vector bundles and twisted $K$-theory can be generalised greatly, e.g. \cite{FM, Gomi2}, and we mention one generalisation --- degree shift--- which is relevant when chiral symmetry $\fS$ is present. By definition, $\fS$ is a grading operator which commutes with $\sG$, so $\cE$ has a ``sublattice decomposition'' into $\cE_A\oplus \cE_B$ where each factor individually hosts a $\sG$ action on its sections. If a gapped Hamiltonian $H$ not only commutes with $\sG$ but also anticommutes with $\fS$, it must have the form
$$H=\begin{pmatrix} 0 & U \\ U^\dagger & 0\end{pmatrix} $$
where $U\in{\rm End}(\cE_B,\cE_A)$ is an invertible bundle map which intertwines the $\sG$ actions. If $\cE_A, \cE_B$ are trivialisable, we get a continuous assignment $U:k\mapsto U(k)$ of invertible matrices with respect to some trivialisation. Usually, there is a reference isomorphism identifying $\cE_A\cong \cE_B\cong\cE$, then we may regard $U\in{\rm Aut}(\cE)$.

For the case $\sG=\ZZ^d$, there is no (twisted) point group action to intertwine, so $U(k)$ is arbitrary and the homotopy classes of $U$ (in the stabilised sense) give the group $K^{-1}(\cB)$. For $d=1$, this led to $K^{-1}(S^1)$ providing the bulk topological index for $\fS$-symmetric (and translation-symmetric) Hamiltonians, e.g.\ in the SSH model.

For a general space group $\mathscr{G}$, $U$ is not arbitrary, but must satisfy a compatibility condition with $\sG$. We had an example for $\mathscr{G}=\pg$ with ${\rm dim}\,\cE=2$, where the condition was Eq.\ \eqref{compatibilityconditions}. More generally, we can assume that $\mathcal{E}$ is the direct product of some copies of the regular $\widehat{\nu}$-twisted $\Z_2$-equivariant vector bundle $\mathcal{E}_{\mathrm{reg}}$ as in Lemma \ref{lem:trivialregularbundle} and Corollary \ref{cor:regularbundle}, so that it is the product bundle of rank $2n$. Accordingly, its twisted $\Z_2$-action $L$ and the commuting unitary automorphism are described as maps 
\begin{align}
V &: \mathcal{B}_x \to {\rm U}(2n), &
U &: \mathcal{B}=\mathcal{B}_x \times \mathcal{B}_y \to {\rm U}(2n)\label{inputVU}
\end{align}
subject to the relations
\begin{align}
V(k_x) V(k_x) &= e^{\im k_x} 1_{\C^{2n}}, &
V(k_x) U(k_x, k_y) &= U(k_x, -k_y) V(k_x).\label{inputVUrelations}
\end{align}
We consider such unitary automorphisms $U$ in the direct limit as $n\rightarrow\infty$, then the homotopy classes of $U$ give the \emph{twisted} degree-shifted $K$-theory group $K^{1+\wh{\nu}}_{\ZZ_2}(\cB)$. The computation of this group was given in \cite{SSG2} with the result $K^{1+\wh{\nu}}_{\ZZ_2}(\cB)\cong\ZZ\oplus\ZZ/2$.

\begin{remark}
As noted in Remark \ref{rem:choiceofglide}, other choices of twisted $\ZZ_2$ actions are possible, being specified by a gauge-transformed $V'$ that could generally depend on $(k_x,k_y)\in\cB$ instead of just $k_x\in\cB$. This applies to both the black and the brown sublattices. Under such a transformation, the matrix function $U$ needs to satisfy a modified compatibility condition from Eq.\ \eqref{inputVUrelations} in order to commute with the twisted $\ZZ_2$ action. Just as in the SSH model, the bulk-edge correspondence must be formulated with respect to a choice of edge. The edge selects for us a choice of unit cell with respect to which we take the twisted $\ZZ_2$ action $L_{-1}$ to be the glide reflection $L^{(1)}$ with glide axis passing through the middle of the unit cells with $n_y=0$; then $V$ has the above $\cB_y$-independent form.
\end{remark}

\begin{proposition}\label{prop:twistedK1}
Let $\sG$ be the nonsymmorphic wallpaper group $\pg$ with point group $F=\ZZ_2$ and 2-cocycle $\nu$. The topological phases of chiral and $\pg$ symmetric Hamiltonians are labelled by $K^{1+\wh{\nu}}_{\ZZ_2}(\cB)$ with $\ZZ_2$ acting on $\cB=\TT^2$ by $(k_x,k_y)\mapsto(k_x,-k_y)$. This group is isomorphic to $\ZZ\oplus\ZZ/2$, with the torsion phase represented by $U_p$ of Eq.\ \eqref{compatibleunitaries} and a free generator represented by $U_r$.
\end{proposition}
\begin{proof}
$K^{1+\wh{\nu}}_{\ZZ_2}(\cB)\cong\ZZ\oplus\ZZ/2$ is computed in Section \ref{sec:Gysin} via a Gysin sequence, and Corollary \ref{cor:freegenerator} there shows that $[U_r]$ can be taken to be a free generator. In Section \ref{sec:twistedToeplitzindex}, the torsion part is shown to be detected by an analytic mod 2 index which is non-zero on the 2-torsion class $[U_p]$.
\end{proof}

\begin{definition}
The bulk topological phase represented by $U_p$ as in Proposition \ref{prop:twistedK1} is called the Klein bottle phase.
\end{definition}
This terminology is justified by the relationship between $[U_p]$ and the 2-torsion homology cycle of the Klein bottle, explained in Section \ref{sec:BaumConnes}.

\section{Edge symmetries and topological invariants}\label{sec:edgeinvariants}
\subsection{Graded frieze group symmetry, graded twists, and the edge index}
Chiral symmetry can be generalised, as in \cite{FM}, by considering \emph{graded} groups, i.e.\ a group $G_c$ together with a surjective grading homomorphism $c$ into the two-element group,
$$ 1\longrightarrow G\longrightarrow G_c\overset{c}{\longrightarrow}\ZZ_2\longrightarrow 1.$$
Then $G_c$ is required to act on a graded Hilbert space, or sections of a graded Hermitian vector bundle, in such a way that the even subgroup $G=$ker$(c)$ acts as even operators while the odd coset acts as odd operators. For example, the $\sG$ and chiral symmetry $\fS$ can be combined into a single graded group $\sG_c$,
$$ 0\longrightarrow \sG\longrightarrow \sG_c=\sG\times \ZZ_2[\fS]\overset{c=(0,{\rm id})}{\longrightarrow}\ZZ_2\longrightarrow 1.$$ 

In general, the grading homomorphism need not be split. A simple example is Eq.\ \eqref{1Dnonsymmorphicgroup} which we rewrite as
\begin{equation*} 
0\longrightarrow \ZZ\overset{\times 2}{\longrightarrow}\ZZ_c\overset{c=(-1)^{(\cdot)}}{\longrightarrow} \ZZ_2\longrightarrow 1.
\end{equation*}
The graded group $\ZZ_c$ is generated by a glide reflection $L$ which is given the odd grading, and can be thought of as a generalisation of a 1D space group. Indeed, a 1D glide axis for $\pg$ precisely enjoys such a generalised type of ``graded'' symmetry! A glide axis is not a purely 1D concept, but comes with a notion of having two sides which are exchanged by the glide reflection generating $\ZZ_c$ (see Fig.\ \ref{fig:glideaxis}). A glide reflection precisely generates the \emph{frieze} group $\sf{p11g}$ in the crystallography literature \cite{Kopsky}. Thus we gain the valuable insight that we need not be limited to ordinary space groups when considering symmetries of lower-dimensional \emph{boundaries} in Euclidean space. This is in contrast to the viewpoint in \cite{Wieder}, where symmetries of surfaces in 3D are required to have a wallpaper group of symmetries.

We wish to classify the possible zero modes, or representation spaces, which can be consistent with $\ZZ_c$. We can Fourier transform with respect to the even subgroup $\ZZ$ generated by $T_x=L_{-1}^2$, to get the 1D Brillouin zone $\cB_x$, on which the dual action of the quotient ``point group'' $F=\ZZ_2$ is \emph{trivial}. The group 2-cocycle $\nu$ is $\nu(-1,-1)=1\in\ZZ$ so that $\wh{\nu}(-1,-1)(k_x)=e^{\im k_x}$. Note that this cocycle is essentially the same as our earlier cocycle defining $\pg$ (which was valued in $\ZZ_x\subset\ZZ^2$), which is why we re-use the notations $\nu$ and $\wh{\nu}$ here.

We obtain a graded vector bundle on which $L_{-1}$ is an \emph{odd} bundle map and furnishes a $\wh{\nu}$-twisted $\ZZ_2$ action. The requirement of $L$ to be odd is encoded by the identity grading homomorphism $F=\ZZ_2\rightarrow\ZZ_2$, which we also denote by $c$. This is an example of a $(\wh{\nu},c)$-twisted $F$ equivariant vector bundle over $\cB_x$, which are classified by an equivariant $K$-theory group denoted $K^{0 + c + \widehat{\nu}}_{\ZZ_2}(\mathcal{B}_x)$. Strictly speaking, such finite-rank twisted vector bundles represent elements of a ``finite-rank twisted equivariant $K$-theory'' group $K^{0 + c + \widehat{\nu}}_{\ZZ_2}(\mathcal{B}_x)_{\rm fin}$, but this turns out to be isomorphic to the ``true'' $K^{0 + c + \widehat{\nu}}_{\ZZ_2}(\mathcal{B}_x)$ as explained in the next Subsection.

We will show that the zero modes of a $\pg$ and chiral symmetric Hamiltonian truncated along a glide axis, are ``topologically protected'' and have a mod 2 \emph{edge index} in $K^{0 + c + \widehat{\nu}}_{\ZZ_2}(\mathcal{B}_x)$. This requires some preparation in the notion of $K$-theory with \emph{graded} equivariant twists.

We mention in passing that groups with gradings induced by reflections in a hyperplane, their $C^*$-algebras, and equivariant $K$-theory, had been studied in \cite{Stolz} in the context of positive scalar curvature metrics, and that related notions appear in twisted equivariant $KR$-theory studied in \cite{Moutuou}.

\subsection{Finite rank twisted equivariant $K$-theory $K^{0 + c + \widehat{\nu}}_{\Z_2}(\mathcal{B}_x)_{\mathrm{fin}}$}\label{sec:finiterankK}
Let us review the finite rank formulation of twisted equivariant $K$-theory \cite{FM,Gomi2}. As in the above example, let $\mathcal{B}_x = \R/2\pi\Z$ be the circle with the trivial $\Z_2$-action, $c : \Z_2 \to \Z_2$ the identity homomorphism, and $\widehat{\nu}$ the cocycle 
\begin{align}
\widehat{\nu}(1, 1)(k_x)
&= \widehat{\nu}(1, -1)(k_x)
= \widehat{\nu}(-1, 1)(k_x) = 1, &
\widehat{\nu}(-1, -1)(k_x)
&= e^{\im k_x}.\label{Bxcocycle}
\end{align}
A \textit{$(\widehat{\nu}, c)$-twisted $\Z_2$-equivariant vector bundle} on $\mathcal{B}_x$ is a complex vector bundle $\mathcal{E} \to \mathcal{B}_x$ equipped with a $\Z_2$-grading $\epsilon : \mathcal{E} \to \mathcal{E}$ and a $\widehat{\nu}$-twisted $\Z_2$-action $L_a : \mathcal{E} \to \mathcal{E}$ such that $\epsilon L_a = c(a) L_a \epsilon$ for $a \in \Z_2$. We say that a $(\widehat{\nu}, c)$-twisted $\Z_2$-equivariant vector bundle $(\mathcal{E}, \epsilon, L)$ admits a \textit{compatible $Cl_1$-action} (or \emph{Clifford} action) if there is a vector bundle map
$$
\begin{CD}
\mathcal{E} @>{\gamma}>> \mathcal{E} \\
@V{\pi}VV @VV{\pi}V \\
\mathcal{B} @= \mathcal{B}
\end{CD}
$$
such that
\begin{align*}
\gamma^2 &= 1, &
\epsilon \gamma &= - \gamma \epsilon, &
\gamma L_a &= c(a) L_a \gamma.
\end{align*}
For such twisted bundles, there is an obvious notion of isomorphisms. 

Let us write $\mathrm{Vect}^{(\widehat{\nu}, c)}_{\Z_2}(\mathcal{B}_x)$ for the monoid of isomorphism classes of $(\widehat{\nu}, c)$-twisted $\Z_2$-equivariant vector bundles on $\mathcal{B}_x$ of finite rank, and $\mathrm{Triv}^{(\widehat{\nu}, c)}_{\Z_2}(\mathcal{B}_x)$ for the submonoid of those admitting compatible $Cl_1$-actions. It can be shown that the quotient monoid gives rise to an abelian group (reverse the grading to get the inverse), which we denote by
$$
K^{0 + c + \widehat{\nu}}_{\Z_2}(\mathcal{B}_x)_{\mathrm{fin}}
= \mathrm{Vect}^{(\widehat{\nu}, c)}_{\Z_2}(\mathcal{B}_x)/
\mathrm{Triv}^{(\widehat{\nu}, c)}_{\Z_2}(\mathcal{B}_x),
$$
and compute explicitly in Corollary \ref{cor:finiteedgeKcomputation}. As it turns out, the ``true'' twisted equivariant $K$-theory group $K^{0+c+\wh{\nu}}_{\ZZ_2}(\cB_x)$ achieves the same classification \cite{SSG1,SSG2}. To be more precise, there is a homomorphism $K^{0+c+\wh{\nu}}_{\ZZ_2}(\cB_x)_{\rm fin}\rightarrow K^{0+c+\wh{\nu}}_{\ZZ_2}(\cB_x)$ which is an isomorphism in this case.

In \cite{SSG2} (VIII, E, 2,pp.\ 29--30), the computation $K^{0+c+\wh{\nu}}_{\ZZ_2}(\cB_x)\cong\ZZ/2$ was carried out. As explained later, there is a natural \emph{topological index map} $K^{1+\wh{\nu}}_{\ZZ_2}(\cB)\rightarrow K^{0+c+\wh{\nu}}_{\ZZ_2}(\cB_x)$ which is surjective, and therefore detects the $\pg$ and $\fS$ symmetric bulk topological phase corresponding to the 2-torsion element of $K^{1+\wh{\nu}}_{\ZZ_2}(\cB)$.


\subsection{Fredholm formulation of $K^{0 + c + \widehat{\nu}}_{\Z_2}(\mathcal{B}_x)$}

Next, we review an infinite-dimensional formulation by using Fredholm operators. Let $\mathcal{E} \to \mathcal{B}_x$ be a $(\widehat{\nu}, c)$-twisted vector bundle. We assume that $\mathcal{E}$ is a Hermitian vector bundle, and the $\Z_2$-grading and the twisted group actions are unitary. We can allow infinite-dimensional $\mathcal{E}$. In this case, we assume that the fibres are separable Hilbert spaces, and the structure group is the group of unitary operators topologised by the compact-open topology in the sense of Atiyah and Segal. Now, a (self-adjoint) \textit{Fredholm family} on $\mathcal{E}$ is a vector bundle map $A : \mathcal{E} \to \mathcal{E}$ such that:
\begin{itemize}
\item
$A_{k_x} = A_{k_x}^\dagger$ is bounded for each $k_x \in \mathcal{B}_x$,

\item
$A_{k_x}^2 - 1_{\mathcal{E}|_{k_x}}$ is compact for each $k_x \in \mathcal{B}_x$,

\item
$\mathrm{Spec}A_{k_x} \subset [-1, 1]$ for each $k_x \in \mathcal{B}_x$, and

\item
$A$ is odd with respect to the $\Z_2$-grading:
\begin{align*}
A \epsilon &= - \epsilon A, &
A L_a &= c(a) L_a A.
\end{align*}
\end{itemize}
Regarding the continuity of $A$, we assume that the family of bounded operators $A_{k_x}$ is continuous in the compact-open topology in the sense of Atiyah and Segal, and that of compact operators $A^2_{k_x} - 1_{\mathcal{E}|_{k_x}}$ is continuous in the operator norm topology. From a general argument, there exists a \textit{locally universal} $(\widehat{\nu}, c)$-twisted vector bundle $\mathcal{E}_{\mathrm{univ}}$ on $\mathcal{B}_x$, and any $(\widehat{\nu}, c)$-twisted vector bundle $\mathcal{E}$ (possibly infinite rank) can be embedded into $\mathcal{E}_{\mathrm{univ}}$. The space of invertible Fredholm families on $\mathcal{E}_{\mathrm{univ}}$ is non-empty and contractible. Considering the fibrewise homotopy classes of Fredholm families on $\mathcal{E}_{\mathrm{univ}}$, we get the Fredholm formulation of the twisted $K$-theory
$$
K^{0 + c + \widehat{\nu}}_{\Z_2}(\mathcal{B}_x)
= \{ A |\ \mbox{Fredholm family on $\mathcal{E}_{\mathrm{univ}}$} \}/
\mathrm{homotopy},
$$
in which the zero element is represented by an invertible Fredholm family.

\medskip

The (local) universality of $\mathcal{E}_{\mathrm{univ}}$ leads to a homomorphism
$$
\imath : K^{0 + c + \widehat{\nu}}_{\Z_2}(\mathcal{B}_x)_{\mathrm{fin}} \to
K^{0 + c + \widehat{\nu}}_{\Z_2}(\mathcal{B}_x).
$$
Concretely, given a finite rank twisted bundle $\mathcal{E}$, we have the trivial Fredholm family $A \equiv 0$. An embedding $\mathcal{E} \to \mathcal{E}_{\mathrm{univ}}$, which is essentially unique, leads to the orthogonal decomposition $\mathcal{E}_{\mathrm{univ}} = \mathcal{E} \oplus \mathcal{E}^\perp$. The orthogonal complement $\mathcal{E}^\perp$ also has the local universality, so that there is an invertible Fredholm family $\gamma^\perp$ on $\mathcal{E}^\perp$ . Then, we get a Fredholm family $0 \oplus \gamma^\perp$ on $\mathcal{E}_{\mathrm{univ}}$. The assignment $\mathcal{E} \mapsto 0 \oplus \gamma^\perp$ induces the homomorphism $\imath$. 

\medskip

As a matter of fact \cite{SSG2} (VIII, E, 2, page 29--30), we have
$$
K^{0 + c + \widehat{\nu}}_{\Z_2}(\mathcal{B}_x) \cong \ZZ/2.
$$
This computation is also given in the Appendix. Further, as is pointed out in \cite{Gomi2} (\S\S4.4), we also have $K^{0 + c + \widehat{\nu}}_{\Z_2}(\mathcal{B}_x)_{\mathrm{fin}} \cong \ZZ/2$, and so the homomorphism $\imath$ in this case is an isomorphism. We here provide the proof of this fact, which is not detailed in \cite{Gomi2}.

\begin{lem} \label{lem:classification_twistd_bundle}
Any $(\widehat{\nu}, c)$-twisted vector bundle on $\mathcal{B}_x$ of finite rank is isomorphic to the product bundle $\mathcal{B}_x \times \C^{2n}$ with its $\Z_2$-grading $\epsilon$ and twisted action $L_{-1}$ given by
\begin{align*}
\epsilon
&=
\left(
\begin{array}{cc}
1_{\C^n} & 0 \\
0 & -1_{\C^n} 
\end{array}
\right), &
L_{-1} &: (k_x, \vec{\psi})
\mapsto
(k_x, 
\left(
\begin{array}{cc}
0 & e^{\im k_x}1_{\C^n} \\
1_{\C^n} & 0 
\end{array}
\right)
\vec{\psi}
).
\end{align*}
\end{lem}

\begin{proof}
Given such a twisted vector bundle $\mathcal{E}$, the same argument as in the previous section shows that its rank is an even number $2n$. Any complex vector bundle on a circle is topologically trivial, so that we can assume $\mathcal{E} = \mathcal{B}_x \times \C^{2n}$. Because of the anti-commutation relation $\epsilon L_{-1} = - L_{-1} \epsilon$, the eigenspaces of $\epsilon$ have the same multiplicity. Hence we can assume that the $\Z_2$-grading $\epsilon$ is as stated. We can also assume that $L_1 : \mathcal{E} \to \mathcal{E}$ is the identity. Now, let us express $L_{-1} : \mathcal{E} \to \mathcal{E}$ as
$$
L_{-1} : (k_x, \vec{\psi})
\mapsto
(k_x, U(k_x) \vec{\psi})
$$
by using a map $U : \mathcal{B}_x \to {\rm U}(2n)$. The anti-commutation relation $\epsilon L_{-1} = - L_{-1} \epsilon$ allows us to express $U(k_x)$ as
$$
U(k_x)
=
\left(
\begin{array}{cc}
0 & B(k_x) \\
C(k_x) & 0 
\end{array}
\right).
$$
Because $L$ is a twisted action, it must hold that $U(k_x)U(k_x) = e^{\im k_x}$, which is equivalent to the constraint $B(k_x)C(k_x) = C(k_x)B(k_x) = e^{\im k_x}$. Hence we have $B(k_x) = e^{\im k_x}C(k_x)^{-1}$. Since
$$
\left(
\begin{array}{cc}
1 & 0 \\
0 & C(k_x)^{-1}
\end{array}
\right)
\left(
\begin{array}{cc}
0 & e^{\im k_x}C(k_x)^{-1} \\
C(k_x) & 0
\end{array}
\right)
\left(
\begin{array}{cc}
1 & 0 \\
0 & C(k_x)
\end{array}
\right)
=
\left(
\begin{array}{cc}
0 & e^{\im k_x} \\
1 & 0 
\end{array}
\right),
$$
this twisted vector bundle is isomorphic to the one in this lemma.
\end{proof}

\begin{prop} \label{prop:classification_twisted_bundle_with_Clifford}
A $(\widehat{\nu}, c)$-twisted vector bundle $\mathcal{E} \to \mathcal{B}_x$ of finite rank admits a compatible $Cl_1$-action if and only if $\mathrm{dim}\,\mathcal{E} \in 4\mathbb{N}$.
\end{prop}

\begin{proof}
Without loss of generality, we can assume that $\mathcal{E}$ is as realised in Lemma \ref{lem:classification_twistd_bundle}. If such $\mathcal{E}$ admits a compatible $Cl_1$-action $\gamma$, then it is of the form
$$
\gamma : (k_x, \vec{\psi})
\mapsto
(k_x, 
\left(
\begin{array}{cc}
0 & A(k_x)^{-1} \\
A(k_x) & 0
\end{array}
\right)
\vec{\psi}),
$$
where $A : \mathcal{B}_x \to {\rm U}(n)$ is a map. The anti-commutation relation $L_{-1}\gamma = - \gamma L_{-1}$ is equivalent to $A(k_x)^2 = e^{\im k_x}$. Taking the determinant and its degree around the circle, we get
$$
2 \deg \det A = n.
$$
Hence $\mathrm{dim}\,\mathcal{E} = 2n \in 4\Z$. Conversely, if $n=2m$ is even, then we set
$$
A(k_x)
= 
\left(
\begin{array}{cc}
0 & e^{\im k_x} \\
1 & 0
\end{array}
\right).
$$
This constructs a compatible $Cl_1$-action $\gamma$.
\end{proof}

From the definition of $K^{0 + c + \widehat{\nu}}_{\Z_2}(\mathcal{B}_x)_{\mathrm{fin}}$, it follows that:

\begin{cor}\label{cor:finiteedgeKcomputation}
The assignment 
$$
[\mathcal{E}] \mapsto \frac{1}{2}\mathrm{dim}\, \mathcal{E} \mod 2
$$
realises an isomorphism $K^{0 + c + \widehat{\nu}}_{\Z_2}(\mathcal{B}_x)_{\mathrm{fin}} \cong \Z/2$.
\end{cor}

In view of the construction of $\imath$, it also follows that:

\begin{cor} \label{cor:classification_twisted_bundle_with_Clifford}
Any Fredholm family on a $(\widehat{\nu}, c)$-twisted vector bundle $\mathcal{E} \to \mathcal{B}_x$ is homotopic to a family $A$ such that the family of vector spaces $\mathrm{Ker}\,A = \bigcup_{k_x} \mathrm{Ker}\,A_{k_x}$ gives rise to a $(\widehat{\nu}, c)$-twisted vector bundle. Furthermore, an isomorphism $K^{0 + c + \widehat{\nu}}_{\Z_2}(\mathcal{B}_x) \cong \Z/2$ is realised by the assignment
$$
[A] \mapsto \frac{1}{2} \mathrm{dim} \,\mathrm{Ker}\, A_{k_x} \mod 2,
$$
where $k_x \in \mathcal{B}_x$ is any point.
\end{cor}


\subsection{Karoubi formulation of $K^{0 + c + \widehat{\nu}}_{\Z_2}(\mathcal{B}_x)$}

Karoubi's formulations are reviewed here, which play a key role in the construction of the ``index'' of our interests.

\medskip

In the present context, a \textit{(Karoubi) triple} $(\mathcal{E}, \eta_0, \eta_1)$ on $\mathcal{B}_x$ consists of a finite rank $\widehat{\nu}$-twisted vector bundle $(\mathcal{E}, L)$ on $\mathcal{B}_x$ such that $(\mathcal{E}, \eta_0, L)$ and $(\mathcal{E}, \eta_1, L)$ are $(\widehat{\nu}, c)$-twisted vector bundles. We write $\mathscr{M}^{(\widehat{\nu}, c)}_{\Z_2}(\mathcal{B}_x)$ for the monoid of isomorphism classes of triples $(\mathcal{E}, \eta_0, \eta_1)$, and $\mathscr{Z}^{(\widehat{\nu}, c)}_{\Z_2}(\mathcal{B}_x)$ for its submonoid consisting of isomorphism classes of triples $(\mathcal{E}, \eta_0, \eta_1)$ such that there is a homotopy $\eta_t$ of $\Z_2$-gradings between $\eta_0$ and $\eta_1$ and $(\mathcal{E}, \eta_t, L)$ is a $(\widehat{\nu}, c)$-twisted vector bundle for each $t \in [0, 1]$. As before, the quotient monoid gives rise to an abelian group, which we denote by
$$
\mathscr{K}^{0 + c + \widehat{\nu}}_{\Z_2}(\mathcal{B}_x)_{\mathrm{fin}}
= \mathscr{M}^{(\widehat{\nu}, c)}_{\Z_2}(\mathcal{B}_x)/
\mathscr{Z}^{(\widehat{\nu}, c)}_{\Z_2}(\mathcal{B}_x).
$$

Let $\mathcal{E}_{\mathrm{univ}} = (\mathcal{E}_{\mathrm{univ}}, \epsilon_{\mathrm{univ}}, L_{\mathrm{univ}})$ be the locally universal $(\widehat{\nu}, c)$-twisted vector bundle on $\mathcal{B}_x$. We write $\mathscr{G}^{0 + c + \widehat{\nu}}_{\Z_2}(\mathcal{E}_{\mathrm{univ}})$ for the space of $\Z_2$-gradings (self-adjoint involutions) $\eta$ on $\mathcal{E}_{\mathrm{univ}}$ such that $(\mathcal{E}_{\mathrm{univ}}, \eta, L_{\mathrm{univ}})$ are $(\widehat{\nu}, c)$-twisted vector bundles and $\eta - \epsilon_{\mathrm{univ}}$ are compact operators on fibres. We topologise $\mathscr{G}^{0 + c + \widehat{\nu}}_{\Z_2}(\mathcal{E}_{\mathrm{univ}})$ by using the operator norm topology. Considering the homotopy $\sim$ within such $\Z_2$-gradings as above, we can define an abelian group
$$
\mathscr{K}^{0 + c + \widehat{\nu}}_{\Z_2}(\mathcal{B}_x)
=
\mathscr{G}^{0 + c + \widehat{\nu}}_{\Z_2}(\mathcal{E}_{\mathrm{univ}})/\sim.
$$

The two Karoubi's formulations (finite-dimensional and infinite-dimensional) are related as follows: Given a triple $(\mathcal{E}, \eta_0, \eta_1)$, we take $\eta_1$ as a $\Z_2$-grading to embedded the $(\widehat{\nu}, c)$-twisted vector bundle $(E, \eta_1)$ into the universal bundle $(\mathcal{E}_{\mathrm{univ}}, \epsilon_{\mathrm{univ}})$. The embedding induces the orthogonal decomposition $\mathcal{E}_{\mathrm{univ}} = \mathcal{E} \oplus \mathcal{E}^\perp$, and also the decomposition of the $\Z_2$-grading $\epsilon_{\mathrm{univ}} = \eta_1 \oplus \epsilon^\perp$. Now, $\eta_0 \oplus \epsilon^\perp$ is a $\Z_2$-grading on $\mathcal{E}_{\mathrm{univ}}$, and the assignment $(E, \eta_0, \eta_1) \mapsto \eta_0 \oplus \epsilon^\perp$ induces a well-defined homomorphism
$$
\jmath : \ 
\mathscr{K}^{0 + c + \widehat{\nu}}_{\Z_2}(\mathcal{B}_x)_{\mathrm{fin}}
\longrightarrow
\mathscr{K}^{0 + c + \widehat{\nu}}_{\Z_2}(\mathcal{B}_x).
$$
It can be shown \cite{Gomi2} that $\jmath$ above is bijective.

The infinite-dimensional Karoubi formulation and the Fredholm formulation are related by a homomorphism
$$
\vartheta : \ 
K^{0 + c + \widehat{\nu}}_{\Z_2}(\mathcal{B}_x)
\longrightarrow
\mathscr{K}^{0 + c + \widehat{\nu}}_{\Z_2}(\mathcal{B}_x).
$$
This is induced from
$$
\vartheta(A) = - e^{\pi A\epsilon_{\mathrm{univ}}}\epsilon_{\mathrm{univ}}
$$
and is also bijective, see Section 4 of \cite{Gomi2}.

In terms of these formulations, the $\ZZ/2$-invariants are described as follows:

\begin{lem} \label{lem:isomorphism_in_Karoubi_formulations}
The following holds true.
\begin{itemize}
\item[(a)]
The isomorphism $\mathscr{K}^{0 + c + \widehat{\nu}}_{\Z_2}(\mathcal{B}_x) \cong \Z/2$ is induced from the assignment
$$
\eta = \{ \eta_{k_x} \}_{k_x \in \mathcal{B}_x} 
\mapsto 
\dim\big(
\mathrm{Ker}\,\frac{1 - \eta_{k_x}}{2} \cap 
\mathrm{Ker}\,\frac{1 + \epsilon_{\mathrm{univ}}}{2}
\big) \mod 2,
$$
where $k_x \in \mathcal{B}_x$ is any point.

\item[(b)]
The isomorphism $\mathscr{K}^{0 + c + \widehat{\nu}}_{\Z_2}(\mathcal{B}_x)_{\mathrm{fin}} \cong \Z/2$ is induced from the assignment
$$
(\mathcal{E}, \eta_0, \eta_1) \mapsto
\dim\big(
\mathrm{Ker}\,\frac{1 - (\eta_0)_{k_x}}{2} \cap 
\mathrm{Ker}\,\frac{1 + (\eta_1)_{k_x}}{2}
\big)
\mod 2,
$$
where $k_x \in \mathcal{B}_x$ is any point.
\end{itemize}
\end{lem}

\begin{proof}
For a Fredholm family $A = \{ A_{k_x} \}_{k_x \in \mathcal{B}_x}$ on $(\mathcal{E}_{\mathrm{univ}}, \epsilon_{\mathrm{univ}})$ representing an element of $K^{0 + c + \widehat{\nu}}_{\Z_2}(\mathcal{B}_x)$, we have
\begin{align*}
\mathrm{Ker}\,A_{k_x} \cap \mathrm{Ker}\,\frac{1 + \epsilon_{\mathrm{univ}}}{2}
&=
\mathrm{Ker}\,\frac{1 - \vartheta(A_{k_x})}{2} \cap 
\mathrm{Ker}\,\frac{1 + \epsilon_{\mathrm{univ}}}{2}, \\
\mathrm{Ker}\,A_{k_x} \cap \mathrm{Ker}\frac{1 - \epsilon_{\mathrm{univ}}}{2}
&=
\mathrm{Ker}\,\frac{1 + \vartheta(A_{k_x})}{2} \cap 
\mathrm{Ker}\,\frac{1 - \epsilon_{\mathrm{univ}}}{2}.
\end{align*}
In showing these identifications, we use the fact that $e^{iA_{k_x}}\psi = \psi$ if and only if $A_{k_x}\psi = 0$. The ``if'' part is clear, while the converse follows from Proposition \ref{prop:resolutionidentity} below. Therefore Corollary \ref{cor:classification_twisted_bundle_with_Clifford} proves (a). In view of the construction of the isomorphism $\jmath$, (b) follows from (a).
\end{proof}

\begin{prop}\label{prop:resolutionidentity}
Let $A$ be a bounded self-adjoint operator on a separable infinite-dimensional Hilbert space $\mathscr{H}$, such that
$\mathrm{Spec}(A) \subset [-1, 1]$ and $A^2 - 1$ is compact. If $v\in\mathscr{H}$ satisfies $e^{\pi \im A}v = v$, then $v \in \mathrm{Ker}\,A$.
\end{prop}
\begin{proof}
Let $\{ E(\lambda) \}_{\lambda \in \R}$ be the resolution of the identity associated to the bounded self-adjoint operator $A$. From the hypothesis, we have
$$
0 = \lVert (1 - e^{\pi \im A})v \rVert
= \int_{[-1, 1]} 
\lvert 1 - e^{\pi \im \lambda} \rvert d \lVert E(\lambda) v \rVert.
$$
On the region of $\lambda$ such that $\lvert 1 - e^{\pi \im \lambda} \rvert = 0$, the contribution of the integrand is trivial (zero). Because $\mathrm{Spec}(A) \subset [-1, 1]$, we have $\lvert 1 - e^{\pi \im \lambda} \rvert = 0$ if and only if $\lambda = 0$. Thus, we can write 
$$
0 = \int_{[-1, 1] \backslash \{ 0 \}} 
\lvert 1 - e^{\pi \im \lambda} \rvert d \lVert E(\lambda) v \rVert.
$$
Because the integrand is non-negative, it follows that $\lVert E( [-1, 1] \backslash \{ 0 \}) v \rVert = 0$. In view of the definition of the resolution of the identity, we can identify $E( [-1, 1] \backslash \{ 0 \})$ with the orthogonal projection onto the orthogonal complement of $\mathrm{Ker}\,A$. In other words, we have $E(0)v = v$, in which $E(0) = 1 -  E( [-1, 1] \backslash \{ 0 \})$ is the orthogonal projection onto $\mathrm{Ker}\,A$. Therefore we can conclude that $v = E(0)v \in \mathrm{Ker}\,A$.
\end{proof}

To summarise, we have the following homomorphisms
$$
\begin{CD}
K^{0 + c + \widehat{\nu}}_{\Z_2}(\mathcal{B}_x)_{\mathrm{fin}}
@>{\imath}>{\cong}>
K^{0 + c + \widehat{\nu}}_{\Z_2}(\mathcal{B}_x) \\
@. @V{\cong}V{\vartheta}V \\
\mathscr{K}^{0 + c + \widehat{\nu}}_{\Z_2}(\mathcal{B}_x)_{\mathrm{fin}}
@>{\jmath}>{\cong}>
\mathscr{K}^{0 + c + \widehat{\nu}}_{\Z_2}(\mathcal{B}_x).
\end{CD}
$$
Since $\imath$ is isomorphic in the present case, the composition
$$
\jmath^{-1} \circ \vartheta \circ \imath : \
K^{0 + c + \widehat{\nu}}_{\Z_2}(\mathcal{B}_x)_{\mathrm{fin}}
\to
\mathscr{K}^{0 + c + \widehat{\nu}}_{\Z_2}(\mathcal{B}_x)_{\mathrm{fin}},
$$
is also an isomorphism. We can readily see that, given a finite rank $(\widehat{\nu}, c)$-twisted vector bundle $\mathcal{E}$ with its $\Z_2$-grading $\epsilon$, the triple representing its image under $\jmath^{-1} \circ \vartheta \circ \imath$ is $(\mathcal{E}, -\epsilon, \epsilon)$. In particular, applying this construction to the $(\widehat{\nu}, c)$-twisted vector bundle $(\mathcal{E}, \epsilon, L)$ of rank $2$ in Lemma \ref{lem:classification_twistd_bundle}, we conclude that the non-trivial element in $\mathscr{K}^{0 + c + \widehat{\nu}}_{\Z_2}(\mathcal{B}_x)_{\mathrm{fin}} \cong \ZZ/2$ is represented by the triple $(\mathcal{E}, -\epsilon, \epsilon)$ on the rank $2$ bundle $\mathcal{E}$.

\subsection{Local formula for $K^{0 + c + \widehat{\nu}}_{\Z_2}(\mathcal{B}_x)\cong\ZZ/2$}
The implication of $K^{0 + c + \widehat{\nu}}_{\Z_2}(\mathcal{B}_x)\cong K^{0 + c + \widehat{\nu}}_{\Z_2}(\mathcal{B}_x)_{\mathrm{fin}} \cong \mathscr{K}^{0 + c + \widehat{\nu}}_{\Z_2}(\mathcal{B}_x)_{\mathrm{fin}}\cong\ZZ/2$ is that there are two distinct topological phases in the $1$-dimensional gapped quantum systems described by $2$ by $2$ Hamiltonians $H(k_x)$ subject to 
\begin{equation}
H(k_x) V(k_x) = - V(k_x) H(k_x)\label{basicNSCrelation}
\end{equation}
with respect to the symmetry $V(k_x)$ such that $V(k_x)V(k_x) = e^{\im k_x}$. In particular, $H_1(k_x) = \epsilon$ and $H_2(k_x) = -\epsilon$ represent the distinct phases \cite{SSG1}, although this should be understood in the \emph{relative sense}, cf.\ the dimerised SSH model phases and \cite{Thiang2}. More precisely, $H_1$ and $H_2$ both represent the non-trivial class in $K^{0 + c + \widehat{\nu}}_{\Z_2}(\mathcal{B}_x)_{\mathrm{fin}}$, but they are not \emph{homotopic} --- the obstruction to the latter is encoded by the nontrivial Karoubi triple 
$(\mathcal{E}, -\epsilon, \epsilon)$. However, $\epsilon\oplus\epsilon$ and $-\epsilon\oplus-\epsilon$ \emph{are} homotopic because of the relation $2[\cE,\epsilon,-\epsilon]=0$ on $K$-theory classes. Thus the obstruction between $H_1$ and $H_2$ is 2-torsion; alternatively, $H_1\oplus H_1(k_x)=\epsilon\oplus\epsilon$ trivialises in $K^{0 + c + \widehat{\nu}}_{\Z_2}(\mathcal{B}_x)_{\mathrm{fin}}$ by Proposition \ref{prop:classification_twisted_bundle_with_Clifford}.

For a general $H(k_x)$, let us now provide a local integral formula for computing its $\Z/2$-invariant, based on the idea in \cite{SSG1}. Without loss of generality, we can assume that $V(k_x)$ is of the form
$$
V(k_x)
=
\left(
\begin{array}{cc}
0 & e^{\im k_x} \\
1 & 0
\end{array}
\right).
$$
By Eq.\ \eqref{basicNSCrelation}, the self-adjoint matrix $H(k_x)$ is trace free. Therefore we can express $H(k_x)$ as
\begin{equation}
H(k_x)
=
\left(
\begin{array}{cc}
a(k_x) & \overline{b(k_x)} \\
b(k_x) & - a(k_x)
\end{array}
\right)\label{generalH}
\end{equation}
by using $a : \mathcal{B}_x \to \R$ and $b : \mathcal{B}_x \to \C$. Eq.\ \eqref{basicNSCrelation} also leads to $\overline{b(k_x)} = - e^{\im k_x} b(k_x)$. Because $\det H(k_x) \neq 0$, we have $a(k_x)^2 + \lvert b(k_x) \rvert^2 \neq 0$. As a result, we can define a map
\begin{align*}
\zeta &: \R/2\pi\Z \to \mathrm{U}(1), &
\zeta(\ell) &=
\frac{a(2\ell) + e^{\im \ell}b(2\ell)}
{\lvert a(2\ell) + e^{\im \ell}b(2\ell) \rvert}.
\end{align*}
With $R/2\pi\Z$ given the involution $\ell\mapsto\ell+\pi$ and $\mathrm{U}(1)$ the complex conjugation involution, $\zeta$ is $\Z_2$-equivariant in the sense that $\zeta(\ell + \pi) = \overline{\zeta(\ell)}$. Up to homotopy, such maps are classified by the equivariant cohomology $H^1_{\Z_2}(\R/2\pi\Z; \Z(1))$, where $\Z(1)$ denotes local coefficients $\ZZ$ with involution given by negation. It is known \cite{D-G} that $H^1_{\Z_2}(\R/2\pi\Z; \Z(1)) \cong \Z/2$. To see this fact, we can appeal to a direct analysis: A continuous equivariant map $\zeta : \R/2\pi \Z \to \mathrm{U}(1)$ has the trivial winding number around the circle. Hence $\zeta$ admits an expression $\zeta(\ell) = \exp2\pi \im f(\ell)$ in terms of a continuous map $f : \R \to \R$ with the periodicity $f(\ell + 2\pi) = f(\ell)$. Note that $f$ is not unique, but $\zeta(\ell) = \exp2\pi \im f(\ell) = \exp2\pi \im f'(\ell)$ if and only if $f'(\ell) - f(\ell) = n$ for an integer $n \in \Z$. The condition $\zeta(\ell + \pi) = \overline{\zeta(\ell)}$ is equivalent to $f(\ell) + f(\ell + \pi) = m$ for an integer $m \in \Z$. Considering the Fourier expansion of $f$, we can see that $\zeta$ is equivariantly homotopic to the constant map at $e^{\pi \im m} \in \Z_2$. As a result, an invariant in $\Z/2$ detecting the equivariant homotopy class of $\zeta$ (and hence $H$) is given by
$$
\ZZ/2\ni\mu(\zeta) = \mu(H) :=
m = \frac{1}{\pi} \int_0^{2\pi}f(\ell)d\ell 
= \frac{1}{\pi} \int_0^{2\pi} \frac{1}{2\pi \im}\log\zeta(\ell) d\ell
\mod 2\Z.
$$
It should be noticed that the nature of this $\Z/2$-invariant is \textit{relative} rather than \textit{absolute}: In constructing the invariant $\mu \in \Z/2$, a number of choices, such as a presentation of $H$ in terms of $a$ and $b$, are made (we could, e.g.\ replace $a$ by $-a$ in Eq.\ \eqref{generalH}). These choices can change the value of $\mu \in \Z/2$, so that there is no absolute meaning of $\mu(H)$ as an invariant of $H$. However, once the choices are fixed, $\mu(H) - \mu(H') \in \Z/2$ is capable of detecting the difference of the equivariant homotopy classes of $H$ and $H'$.


\section{Construction of $\ZZ/2$ bulk-edge index map}\label{sec:genToeplitz}


\subsection{Topological bulk-edge Gysin map}\label{sec:Gysin}

Over $\mathcal{B}_x$ is the trivial ``Real'' line bundle $\underline{\tilde{\C}}= \mathcal{B}_x \times \C$ in the sense of Atiyah, whose involution is $(k_x, z) \mapsto (k_x, \bar{z})$. Then we have the following form of the Thom isomorphism theorem
$$
K^{n + c + \widehat{\nu}}_{\Z_2}(\mathcal{B}_x)
\cong K^{n + \widehat{\nu}}_{\Z_2}
(D(\underline{\tilde{\C}}), S(\underline{\tilde{\C}})),
$$
where $D(\underline{\tilde{\C}})$ is the unit disk bundle of $\underline{\tilde{\C}}$, and $S(\underline{\tilde{\C}})$ the unit circle bundle. As a space with $\Z_2$-action, the disk bundle $D(\underline{\tilde{\C}})$ is equivariantly homotopy equivalent to $\mathcal{B}_x$, and $S(\underline{\tilde{\C}})$ is nothing but the torus $\mathcal{B}_x \times \mathcal{B}_y$. Thus, the long exact sequence for the pair $(D(\underline{\tilde{\C}}), S(\underline{\tilde{\C}}))$ leads to the Gysin exact sequence
$$
\begin{CD}
\overbrace{K^{1 + \widehat{\nu}}_{\Z_2}
(\mathcal{B}_x \times \mathcal{B}_y)}^{\Z\oplus\Z/2} @<{\pi^*}<<
\overbrace{K^{1 + \widehat{\nu}}_{\Z_2}(\mathcal{B}_x)}^{\Z} @<<<
\overbrace{K^{1 + c + \widehat{\nu}}_{\Z_2}(\mathcal{B}_x)}^0 \\
@V{\pi_*}VV @. @AA{\pi_*}A \\
\underbrace{K^{0 + c + \widehat{\nu}}_{\Z_2}(\mathcal{B}_x)}_{\Z/2} @>>>
\underbrace{K^{0 + \widehat{\nu}}_{\Z_2}(\mathcal{B}_x)}_{\Z} @>>{\pi^*}>
\underbrace{K^{0 + \widehat{\nu}}_{\Z_2}
(\mathcal{B}_x \times \mathcal{B}_y)}_{\Z},
\end{CD}
$$
which is split due to the existence of a $\Z_2$-equivariant section of the circle bundle $\pi : \mathcal{B}_x \times \mathcal{B}_y \to \mathcal{B}_x$ (here $\cB_y$ is equivariantly collapsed under $\pi$ to one of its fixed points). Thus $K^{1 + \widehat{\nu}}_{\Z_2}
(\mathcal{B}_x \times \mathcal{B}_y)\cong\ZZ\oplus\ZZ/2$ with a free generator the pullback of a generator for $K^{1 + \widehat{\nu}}_{\Z_2}
(\mathcal{B}_x)\cong \ZZ$.

We can exhibit a generator for $K^{1 + \widehat{\nu}}_{\Z_2}
(\mathcal{B}_x)$ as follows (cf.\ VIII.D.3-4 of \cite{SSG2}). 
The proof of Corollary \ref{cor:regularbundle} also shows that $K^{0 + \widehat{\nu}}_{\Z_2}
(\mathcal{B}_x)\cong \ZZ$ is generated by the trivial bundle $\check{\cE}_{\rm reg}=\cB_x\times \CC^2$ with twisted $\ZZ_2$-action given by
$$
L_{-1} : (k_x, \vec{\psi})
\mapsto
(k_x, V(k_x) \vec{\psi})=(k_x, \begin{pmatrix} 0 & e^{\im k_x} \\ 1 & 0 \end{pmatrix} \vec{\psi})
$$
(any twisted bundle has trivial underlying complex bundle by the low dimension of $\cB_y$). Elements of $K^{1 + \widehat{\nu}}_{\Z_2}
(\mathcal{B}_x)\cong \ZZ$ are represented by automorphisms $\check{U}$ of direct sums of copies of $\check{\cE}_{\rm reg}$. The degree
\begin{align*}
K^{1 + \widehat{\nu}}_{\Z_2}(\mathcal{B}_x) &\to \Z, &
\check{U} &\mapsto \deg \det \check{U}
\end{align*}
is an isomorphism, since the automorphism $\check{U}_r$ of $\check{\cE}_{\rm reg}$ given by $\check{U}_r(k_x)=V(k_x)=\begin{pmatrix} 0 & e^{\im k_x} \\ 1 & 0 \end{pmatrix}$ already has degree $-1$. Thus $[\check{U}_r]$ generates $K^{1 + \widehat{\nu}}_{\Z_2}
(\mathcal{B}_x)$.

The pullback $\pi^*\check{\cE}_{\rm reg}$ is just $\cE_{\rm reg}\rightarrow\cB$ with the twisted action Eq.\ \eqref{regularbundle}, on which the pullback automorphism $\pi^*\check{U}_r$ is just $U_r=\begin{pmatrix} 0 & u_x \\ 1 & 0 \end{pmatrix}$ of Eq.\ \eqref{compatibleunitaries}. Thus
\begin{corollary}\label{cor:freegenerator}
$K^{1 + \widehat{\nu}}_{\Z_2}(\mathcal{B}_x \times \mathcal{B}_y)\cong\ZZ[U_r]\oplus\ZZ/2$. 
\end{corollary}

The topological push-forward $\pi_* : K^{1 + \widehat{\nu}}_{\Z_2}(\mathcal{B}_x \times \mathcal{B}_y) \to K^{0 + c + \widehat{\nu}}_{\Z_2}(\mathcal{B}_x)$ kills $[U_r]$, but detects the torsion part. We would like to have an explicit description of $\pi_*$ as an \emph{analytic} index map. We will further show that $U_p=\begin{pmatrix} u_y & 0 \\ 0 & \overline{u_y}\end{pmatrix}$ indeed represents the non-trivial 2-torsion element of $K^{1 + \widehat{\nu}}_{\Z_2}(\mathcal{B}_x \times \mathcal{B}_y)$.

\bigskip

\subsection{Analytic bulk-edge map}
\subsubsection*{Input and output for analytic push-forward}
The input for the push-forward map is an element of $K^{1 + \widehat{\nu}}_{\Z_2}(\mathcal{B}_x \times \mathcal{B}_y)$, specified by $V,U$ as in Eq.\ \eqref{inputVU}-\eqref{inputVUrelations}.

As the output, we consider a representative of the infinite-dimensional Karoubi formulation $\mathscr{K}^{0 + c + \widehat{\nu}}_{\Z_2}(\mathcal{B}_x)$. That is, 
\begin{itemize}
\item
a separable infinite-dimensional Hilbert space $\mathscr{H}_y$, 

\item
a $\Z_2$-grading $\epsilon$ on $\mathscr{H}_y$ such that $\mathrm{Ker}\,\frac{1 \pm \epsilon}{2}$ are both infinite-dimensional,

\item
a map $\wt{V} : \mathcal{B}_x \to {\mathrm U}(\mathscr{H}_y)$ such that 
\begin{align*}
\wt{V}(k_x) \epsilon &= - \epsilon \wt{V}(k_x), &
\wt{V}(k_x)\wt{V}(k_x) &= e^{\im k_x} 1_{\mathscr{H}_y},
\end{align*}
We require that $\wt{V}$ is continuous with respect to the compact-open topology in the sense of Atiyah and Segal.

\item
a family of self-adjoint involutions $\eta = \{ \eta(k_x) \}_{k_x \in \mathcal{B}_x}$ such that $\eta(k_x) - \epsilon$ is compact for each $k_x \in \mathcal{B}_x$ and
$$
\wt{V}(k_x) \eta(k_x)
= - \eta(k_x) \wt{V}(k_x).
$$
We assume that the family is continuous in $k_x$ with respect to the operator norm topology.

\end{itemize}

We remark that $\mathscr{H}_y\times\cB_x$ gives rise to a locally universal $(\widehat{\nu}, c)$-twisted bundle over $\cB_x$ by the $\Z_2$-grading $\epsilon$ and the twisted $\Z_2$-action $\wt{V}$.


\subsubsection{Main construction}

Suppose that the input ``bulk'' data 
\begin{align*}
V &: \mathcal{B}_x \to {\rm U}(2n), &
U &: \mathcal{B}_x \times \mathcal{B}_y \to {\rm U}(2n)
\end{align*}
such that
\begin{align*}
V(k_x) V(k_x) &= e^{\im k_x} 1_{\C^{2n}}, &
V(k_x) U(k_x, k_y) &= U(k_x, -k_y) V(k_x)
\end{align*}
are given. We construct the output ``edge'' data as follows:
\begin{enumerate}
\item
We define the separable infinite-dimensional Hilbert space $\mathscr{H}_y=L^2(\mathcal{B}_y) \otimes \C^{2n}$

\item
We define a $\Z_2$-grading $\epsilon$ on $\mathscr{H}_y$ by
\begin{align*}
\mathrm{Ker}\,\frac{1 + \epsilon}{2}
&= \widehat{\bigoplus}_{l \le -1}
\C u_y^l \otimes \C^{2n}, \\
\mathrm{Ker}\,\frac{1 - \epsilon}{2}
&= \widehat{\bigoplus}_{l \ge 0}
\C u_y^l \otimes \C^{2n}.
\end{align*}

\item
We define a twisted $\Z_2$-action as follows: Let $M_{V(k_x)}$ and $I$ be the following unitary operators (the multiplication with $V(k_x)$ and the inversion):
\begin{align*}
M_{V(k_x)} &: \mathscr{H}_y \to  \mathscr{H}_y, &
(M_{V(k_x)}\psi)(k_y) &= V(k_x) \psi(k_y), \\
I &: \mathscr{H}_y \to  \mathscr{H}_y, &
(I\psi)(k_y) &= \psi(-k_y).
\end{align*}
We then define $\wt{V} : \mathcal{B}_x \to {\rm U}(\mathscr{H}_y)$ by
$$
\wt{V}(k_x) = M_{\overline{u_y}V(k_x)}I = IM_{u_y V(k_x)}.
$$
This definition is designed for $\wt{V}$ to effect glide reflections along a specified edge, see Remark \ref{rem:choiceofglide}, and we have
\begin{align*}
\wt{V}(k_x) \epsilon &= - \epsilon \wt{V}(k_x), &
\wt{V}(k_x)\wt{V}(k_x) &= e^{\im k_x} 1_{\mathscr{H}_y}.
\end{align*}

\item
We define a family of self-adjoint involutions $\eta(k_x)$ as follows.
A direct computation shows that
$$
\wt{V}(k_x) M_{U(k_x, \cdot)} \wt{V}(k_x)^\dagger
= M_{U(k_x, \cdot)}.
$$
The continuous family of self-adjoint operators
\begin{align*}
k_x\mapsto\eta(k_x) &: \mathscr{H}_y \to \mathscr{H}_y, &
\eta(k_x) 
&= M^\dagger_{U(k_x, \cdot)} \epsilon M_{U(k_x, \cdot)}
\end{align*}
is such that $\eta(k_x) - \epsilon$ is compact (Lemma \ref{lem:compactdifference}) and the anti-commutation relation
$$
\wt{V}(k_x) \eta(k_x)
= - \eta(k_x) \wt{V}(k_x)
$$
holds.
\end{enumerate}

The construction above gives a representative $\eta$ of $\mathscr{K}^{0 + c + \widehat{\nu}}_{\Z_2}(\mathcal{B}_x)$ in the infinite-dimensional Karoubi formulation. As shown in Lemma \ref{lem:isomorphism_in_Karoubi_formulations}, the $\Z/2$-invariant for $\eta$ is given by
$$
\dim\big(
\mathrm{Ker}\,\frac{1 - \eta(k_x)}{2} \cap
\mathrm{Ker}\,\frac{1 + \epsilon}{2}
\big)
\mod 2.
$$
We call this the \emph{analytic $\ZZ/2$ invariant} for $U$ (relative to $V$).

The dimension above has an alternative description in terms of Toeplitz operators. To give it, we note that $$
\mathscr{H}_y^+ = {\rm Ker}\, \frac{1-\epsilon}{2} = \widehat{\bigoplus}_{l \ge 0} \C u_y^l \otimes \C^{2n}.
$$
is the Hardy space $\mathcal{H}^2 \subset L^2(\mathcal{B}_y) \otimes \C^{2n}$, and we let 
\begin{align*}
P_+ &: L^2(\mathcal{B}_y) \otimes \C^{2n} \to \mathcal{H}^2, &
\sum_{l \in \Z} u_y^l \otimes v &\mapsto
\sum_{l \ge 0} u_y^l \otimes v
\end{align*}
be the projection to Hardy space, and $P_-$ the projection onto the complementary subspace $\mathscr{H}_y^-$ in $\mathscr{H}_y$. Then the adjoint $P_+^\dagger : \mathcal{H}^2 \to L^2(\mathcal{B}_y)$ is the inclusion. We define the $k_x$-dependent family of Toeplitz operators associated to $U$ by
\begin{align*}
T_U(k_x) &: \mathscr{H}_y^+ \to \mathscr{H}_y^+ , &
T_U(k_x) &= P_+ M_{U(k_x, \cdot)} P_+^\dagger.
\end{align*}
Similarly, there is a family of ``lower half-plane'' Toeplitz operators 
$$T'_U(k_x): \mathscr{H}_y^- \rightarrow \mathscr{H}_y^- , \qquad\qquad T'_U(k_x)=P_-M_{U(k_x,\cdot)}P_-^\dagger.$$

\begin{proposition}\label{prop:z2invariant}
The analytic $\Z/2$-invariant of an element in $K^{1 + \widehat{\nu}}_{\Z_2}(\mathcal{B}_x \times \mathcal{B}_y)$ described by $V(k_x)$ and $U(k_x, k_y)$ can be computed as
$$
\dim \mathrm{Ker}\, T_U(k_x) \mod 2,
$$
or equivalently as  
$$\dim\mathrm{Ker}\,T'_U(k_x) \mod 2,$$
where $k_x$ is any point in $\cB_x$.
\end{proposition}

\begin{proof}
Recall that $\eta(k_x) = M^\dagger_{U(k_x, \cdot)} \epsilon M_{U(k_x, \cdot)}$. According to the decomposition $\mathscr{H}_y = \mathrm{Ker}\,\frac{1 - \epsilon}{2} \oplus \mathrm{Ker}\,\frac{1 + \epsilon}{2}$, we can express $M_{U(k_x, \cdot)}$ in the following block matrix
$$
M_{U(k_x, \cdot)}
=
\left(
\begin{array}{cc}
P_+ M_{U(k_x, \cdot)} P_+^\dagger &
P_+ M_{U(k_x, \cdot)} P_-^\dagger \\
P_- M_{U(k_x, \cdot)} P_+^\dagger &
P_- M_{U(k_x, \cdot)} P_-^\dagger
\end{array}
\right).
$$
This expression helps us to verify that
$$
\mathrm{Ker}\,\frac{1 + \eta(k_x)}{2} \cap
\mathrm{Ker}\,\frac{1 - \epsilon}{2}
= 
\mathrm{Ker}\, P_+ M_{U(k_x, \cdot)} P_+^\dagger = {\rm Ker}\, T_U(k_x),
$$
so for any $k_x \in \mathcal{B}_x$, we have
$$
\dim\big(
\mathrm{Ker}\,\frac{1 + \eta(k_x)}{2} \cap
\mathrm{Ker}\,\frac{1 - \epsilon}{2}
\big)
= \dim\mathrm{Ker}\,T_U(k_x).
$$
Because of the twisted $\Z_2$-action, we also have that the $\Z/2$ invariant for $\eta$ is
$$
\dim\big(
\mathrm{Ker}\,\frac{1 - \eta(k_x)}{2} \cap
\mathrm{Ker}\,\frac{1 + \epsilon}{2}
\big)
=
\dim\big(
\mathrm{Ker}\,\frac{1 + \eta(k_x)}{2} \cap
\mathrm{Ker}\,\frac{1 - \epsilon}{2}
\big) = \dim\mathrm{Ker}\,T'_U(k_x).
$$
and the Proposition is established.
\end{proof}

Finally, we give a short computation to verify the following, used in part 4 of the main construction.
\begin{lemma}\label{lem:compactdifference}
For each $k_x\in\cB_x$, $\eta(k_x)-\epsilon$ is a compact operator on $\mathscr{H}_y$.
\end{lemma}
\begin{proof}
The continuous function $k_y\mapsto U(k_x,\cdot)$ can be approximated by linear combinations of the exponentials $u_y^l:k_y\mapsto e^{\im l k_y}$. The multiplication operator $M_{U(k_x,\cdot)}$ is thus norm-approximated by the operator of multiplication by such linear combinations. It is easy to see that $P_+M_{u_y^l}P_-^\dagger$ takes the basis functions $u_y^{l'}$ to $u_y^{l+l'}$ for $l'=-l,\ldots,-1$ and is zero otherwise, so it is a finite-rank operator. Then $P_+M_{U(k_x,\cdot)}P_-^\dagger$ is a limit of finite-rank operators and is thus compact. A similar argument shows that $P_-M_U(k_x,\cdot)P_+^\dagger$ is compact. Working modulo the ideal of compact operators, we can rewrite
\begin{align*}
\eta(k_x)-\epsilon&=
\left(
\begin{array}{cc}
T_{U^\dagger}(k_x) &
0 \\
0 &
T'_{U^\dagger}(k_x)
\end{array}
\right)
\left(
\begin{array}{cc}
1 &
0 \\
0 &
-1
\end{array}
\right)
\left(
\begin{array}{cc}
T_U(k_x)  &
0 \\
0 &
T'_U(k_x) 
\end{array}
\right) 
-\left(
\begin{array}{cc}
1 &
0 \\
0 &
-1
\end{array}
\right) \\
&=
\left(
\begin{array}{cc}
T_{U^\dagger}(k_x)T_U(k_x)&
0 \\
0 &
-T'_{U^\dagger}(k_x)T'_U(k_x)
\end{array}
\right) -\left(
\begin{array}{cc}
1 &
0 \\
0 &
-1
\end{array}
\right) 
\end{align*}
which is compact on $\mathscr{H}_y$, because $T_{W^\dagger}T_W-1=T_{W^\dagger}T_W-T_{W^\dagger W}$ is a compact operator for any continuous symbol $W$ (e.g.\  3.5.9-10 of \cite{Murphy} generalised to matrix-valued symbols).
\end{proof}

\subsubsection*{Examples}
Thanks to Proposition \ref{prop:z2invariant} above, we can readily compute the analytic $\ZZ/2$ indices for the basic examples in our model: Let us consider $V : \mathcal{B}_x \to {\mathrm U}(2)$ given by
$$
V(k_x) =
\left(
\begin{array}{cc}
0 & e^{\im k_x} \\
1 & 0
\end{array}
\right).
$$
\begin{itemize}

\item
For $U_r(k_x, k_y) = V(k_x)$, we have $\mathrm{Ker}\,T_{U_r}(k_x) = 0$ since it has no $\cB_y$ dependence, and its $\Z/2$-index is trivial. 
\item
For $U=U_p : \mathcal{B}_x \times \mathcal{B}_y \to U(2)$ given by
$$
U_p(k_x, k_y)
=
\left(
\begin{array}{cc}
e^{\im k_y} & 0 \\
0 & e^{-\im k_y}
\end{array}
\right),
$$
we have $\mathrm{dim\, Ker} \,T_{U_p}(k_x) = 1$, and its $\Z/2$-index is non-trivial. 
\end{itemize}
Let $f^+\in {\rm Ker}\,T_U(k_x)\subset \mathscr{H}_y^+$, which is equivalent to $M_{U(k_x)}f^+\in\mathscr{H}_y^-$. Since $\wt{V}$ is an odd operator, we have 
$$M_{U(k_x)}\wt{V}(k_x)f^+=\wt{V}(k_x)M_{U(k_x)}f^+\in\mathscr{H}_y^+ $$
so that $\wt{V}(k_x)f^+\in{\rm Ker}\,T_U'(k_x)\subset \mathscr{H}_y^-$. Similarly, $\wt{V}(k_x)f^- \in {\rm Ker}\,T_U(k_x)$ for $f^-\in {\rm Ker}\,T_U(k_x)^-$. Thus we see that $\wt{V}(k_x)$ restricts to the graded Hilbert subspace ${\rm Ker}\,T_U(k_x)\oplus{\rm Ker}\,T'_U(k_x)$.

In the case $U=U_p$, the Toeplitz families $T(k_x), T'(k_x)$ do not depend on $k_x$ at all, and their kernels correspond to the dangling $\vC$ and $\vCG$ modes above and below the edge, respectively, as illustrated in Fig.\ \ref{fig:horizontaledge}. When we take the direct integral over $\cB_x$, the operator $L_{-1}=\int^\oplus_{\cB_x}\wt{V}$ becomes an odd ``nonsymmorphic chiral'' symmetry for the graded Hilbert space of zero modes $\int^\oplus_{\cB_x} {\rm Ker}\,T_U(k_x)\oplus \int^\oplus_{\cB_x}{\rm Ker}\,T'_U(k_x)\cong L^2(\cB_x)\oplus L^2(\cB_x)$. This Hilbert space of (brown) edge zero modes coincides with $\mathscr{H}_{e,\rm brown}$ of Section \ref{sec:edgezeroesheuristic}, which was constructed on heuristic grounds. $\mathscr{H}_{e,\rm brown}$, together with the grading and $L_{-1}$, is the section space of the basic $(\wh{\nu},c)$-twisted bundle over $\cB_x$ of Lemma \ref{lem:classification_twistd_bundle} (the $n=1$ case there).

\begin{remark}
For a general $U(k_x,\cdot)$ which depends on $k_x$, the dimension of the kernel of $T_U(k_x)$ may vary with $k_x$, but its mod 2 dimension is invariant.
\end{remark}

\subsubsection{Index theorem for twisted family of Toeplitz operators}\label{sec:twistedToeplitzindex}
To summarise, the input data $U, V$ give rise to a ``twisted family'' of Toeplitz operators $T_U(k_x)$, and we have 
\begin{thm}\label{thm:twistedindextheorem}
The topological push-forward map in the Gysin exact sequence
$$
\pi_* : \
K^{1 + \widehat{\nu}}_{\Z_2}(\mathcal{B}_x \times \mathcal{B}_y) 
\to K^{0 + c + \widehat{\nu}}_{\Z_2}(\mathcal{B}_x) \cong \Z/2
$$
is induced from the assignment $\pi_*'$ to the input data $U$, $V$, of the mod $2$ dimension of the kernel of the Toeplitz operator $T_U(k_x)$ associated to $U$,
$$
\dim \mathrm{Ker}\, T_U(k_x) \mod 2
$$
where $k_x \in \mathcal{B}_x$ is any point.
\end{thm}

\begin{proof}
It is straightforward to see that $\pi'_*$ is a well-defined homomorphism on the classes $[U]\in K^{1 + \widehat{\nu}}_{\Z_2}(\mathcal{B}_x \times \mathcal{B}_y)$. 
We saw that the free generator $U_r$ of $K^{1 + \widehat{\nu}}_{\Z_2}(\mathcal{B}_x \times \mathcal{B}_y)$ has trivial $\ZZ/2$-index, so $\pi'_*$ kills $[U_r]$ just as $\pi_*$ did. Furthermore, the computation that 
$$
U_p(k_x, k_y)
=
\left(
\begin{array}{cc}
e^{\im k_y} & 0 \\
0 & e^{-\im k_y}
\end{array}
\right),
$$
has nontrivial $\ZZ/2$-index shows that $\pi'_*$ surjective, hence $\pi'_*$ agrees with $\pi_*$.
\end{proof}

This generalises the classical index theorem for Toeplitz operators, Theorem \ref{thm:Toeplitz}, to a twisted family of Toeplitz operators parametrised by a compact 1D circle.


Since we can add a free element $U_r$ to $U_p$ without changing its $\ZZ/2$-index, we should show explicitly that $2[U_p]=[U_p\oplus U_p]=[1]$ in $K$-theory. The required homotopy is
\begin{equation}
\begin{pmatrix}
U_p & 0 \\ 0 & 1
\end{pmatrix}
\begin{pmatrix}
\cos t & -V\sin t \\
V\sin t & \cos t
\end{pmatrix}
\begin{pmatrix}
1 & 0 \\ 0 & U_p
\end{pmatrix}
\begin{pmatrix}
\cos t & V^\dagger\sin t \\
-V^\dagger\sin t & \cos t
\end{pmatrix},\;\; t\in[0,\frac{\pi}{2}].\nonumber
\end{equation}

Note that the condition
$$(V\oplus V)(k_x)U(k_x,k_y)=U(k_x,-k_y)(V\oplus V)(k_x)$$
holds throughout the homotopy. At $t=0$, we start off with the direct sum of two copies of $U_p$, but at $t=\frac{\pi}{2}$, this becomes
\begin{equation}
\begin{pmatrix}
U_pVU_pV^\dagger & 0 \\ 0 & 1
\end{pmatrix}=
\begin{pmatrix}
U_pU_p^{-1} & 0 \\ 0 & 1
\end{pmatrix}=1.\nonumber
\end{equation}

Similarly, Proposition \ref{prop:classification_twisted_bundle_with_Clifford} shows that two copies of the Hilbert space of brown (or black) edge modes for $U_p$ trivialises in $K$-theory.

\subsection{Klein bottle phase via Baum--Connes isomorphism}\label{sec:BaumConnes}
The existence of a $\ZZ/2$ bulk phase for $\pg$-symmetric Hamiltonians can also be deduced by computing the $K$-theory of the (reduced) group $C^*$-algebra $C^*(\pg)$. A concrete description of $C^*(\pg)$ may be found in \cite{Taylor} and VII.4 of \cite{Davidson}. The twisted vector bundles of Section \ref{sec:twistedvbnonsymmorphic} arising from a nonsymmorphic space group $\sG$ were classified by $K^{0+\wh{\nu}}_F(\cB)_{\rm fin}$, which also turns out to be $K^{0+\wh{\nu}}_F(\cB)$ (\cite{FM}, Prop.\ 3.14 of \cite{Gomi2}), but we can instead compute the operator $K$-theory of $C^*(\sG)$ \cite{Thiang} --- this is justified by a twisted Green--Julg theorem (\cite{Kubota} Theorem 4.10).




The Baum--Connes conjecture is verified for space groups $\sG$, so that the assembly map $\mu:K_\bullet^\sG(\underline{E}\sG)\rightarrow K_\bullet(C^*(\sG))$ is an isomorphism. Here $\underline{E}\sG$ is the universal proper $\sG$-space which can be taken to be Euclidean space $\EE^d$ since it is contractible and $\sG$ acts on it with finite isotropy. Furthermore, $\ZZ^d$ acts freely on $\EE^d$, so $K_\bullet^\sG(\EE^d)\cong K_\bullet^F(\TT^d)$, and the latter can be computed using a spectral sequence to Bredon homology groups \cite{Valette-book}. The computation of $K_\bullet(C^*(\sG))$ was carried out for all the 2D wallpaper groups in \cite{LS}, and we give an independent computation of $K_\bullet(C^*(\pg))$ in Appendix \ref{appendix:integerbec}.

The Baum--Connes isomorphism above is particularly useful for the case $\sG=\pg$ which acts \emph{freely} on $\EE^2$, so that in fact $K_\bullet^\pg(\EE^2)\cong K_\bullet(B\pg)=K_\bullet(\cK)$ where the classifying space $B\pg=E\pg/\pg=\EE^d/\pg$ is the Klein bottle $\cK$. Thus the computation of $K_\bullet(C^*(\pg))$ reduces to that of $K_\bullet(\cK)$. By low-dimensionality of $\cK$, we simply have the ordinary homology $H_\bullet(\cK)\cong K_\bullet(\cK)$ (e.g.\ Lemma 4.1 of \cite{Mislin}). For $\bullet=1$, this easily gives $\ZZ\oplus\ZZ/2$ (e.g.\ we can abelianise $\pi_1(\cK)=\pg$). The $\ZZ/2$ factor coming from the unorientable cycle (indicated by a double arrow in Fig.\ \ref{fig:pgunitcell}) corresponds to the $\ZZ/2$ ``Klein bottle'' phase constructed in this paper.

\subsubsection{A crystallographic T-duality}
We can Poincar\'{e} dualise $K_\bullet(\cK)\cong K^{\bullet+\mathfrak{o}}(\cK)\cong K_{\ZZ_2}^{\bullet+\wt{\mathfrak{o}}}(\TT^2_{\rm free})$, where $\mathfrak{o}\in H^1(\cK,\ZZ_2)$ is twisting by the orientation bundle of $\cK$ and $\wt{\mathfrak{o}}\in H^1_{\ZZ_2}(\TT^2_{\rm free},\ZZ_2)$ is the corresponding equivariant twist for the 2-torus $\TT^2_{\rm free}$ with free involution (double-covering $\cK$). Then we get the ``crystallographic T-duality'' isomorphism
$$K_{\ZZ_2}^{\bullet+\wt{\mathfrak{o}}}(\TT^2_{\rm free})\cong K^{\bullet+\wh{\nu}}_{\ZZ_2}(\cB).$$

In $\TT^2_{\rm free}$, there is an $F$-invariant circle $\TT_{\rm free}$ on which $F$ acts freely (actually two such circles). Let $\iota$ be the inclusion of $\TT_{\rm free}$ in $\TT^2_{\rm free}$, then the restriction $\breve{\mathfrak{o}}\equiv\iota^*\wt{\mathfrak{o}}$ is an equivariant twist in $H^1_{\ZZ_2}(\TT_{\rm free},\ZZ_2)$. There is in fact also a T-duality
$$K^{\bullet+\breve{\mathfrak{o}}}_{\ZZ_2}(\TT_{\rm free}) \cong K_{\ZZ_2}^{\bullet+1+c+{\wh{\nu}}}(\cB_x),$$
based on the case without $H^1$ twists in \cite{Gomi3,Schneider}. The interpretation of the left-hand-side $(\TT_{\rm free},\breve{\mathfrak{o}})$ is that $\TT_{\rm free}$ is the classifying space for the even subgroup $2\ZZ\subset\ZZ_c$ while $\breve{\mathfrak{o}}$ is a ``grading $\ZZ_2$-bundle'' which has a reversal of grading after going around $\TT_{\rm free}$ half-way (translation by the generator of $\ZZ_c$). The right-hand-side is the momentum space picture --- $\cB_x$ is the Pontryagin dual of the even subgroup $2\ZZ$ on which $\ZZ_2$ acts dually and with a graded twist by $(\wh{\nu},c)$.

We may then verify through explicit computations that the following diagram commutes (for $\bullet=1$, only the torsion subgroups survive in the vertical maps, while for $\bullet=0$ they are trivial):
\begin{equation*}
\xymatrix{
K_{\ZZ_2}^{\bullet+\wt{\mathfrak{o}}}(\TT^2_{\rm free})  \ar[d]^{\iota^*} \ar[r]^{\sim}_{T} & K^{\bullet+\wh{\nu}}_{\ZZ_2}(\cB) \ar[d]^{\pi_*} \\
K^{\bullet+\breve{\mathfrak{o}}}_{\ZZ_2}(\TT_{\rm free}) \ar[r]^{\sim}_{T} & K_{\ZZ_2}^{\bullet+1+c+{\wh{\nu}}}(\cB_x)  }   
\end{equation*}
Thus we verify that the bulk-boundary homomorphism (the Gysin map $\pi_*$ of Sec. \ref{sec:Gysin}) in the momentum space picture is T-dual to the geometrical restriction $\iota^*$ in the position space picture, as is intuitively ``obvious'' and in the same vein as the results in \cite{HMT}. In \cite{GT} we show further that ``crystallographic T-dualities'' exist for \emph{any} crystallographic space group, and just as for $\pg$, graded twists are essential in their formulation.

\section*{Outlook}
\addcontentsline{toc}{section}{\protect\numberline{}Outlook}
The motivation for Theorem \ref{thm:twistedindextheorem} came from studying $\pg$ symmetric Hamiltonians, trying to find the appropriate topological bulk and edge invariants, and realising the bulk-edge correspondence through an index theorem. The wallpaper group $\pg$ is particularly nice because it acts freely on Euclidean space. Physically, this means that the Wyckoff positions (the positioning of $\vI, \vG$ in Euclidean space) are not too important since there are no special positions to consider.

For 3D space groups, there are nine other examples of space groups that act without fixed points (the Bieberbach groups), excluding the case of $\ZZ^3$ with trivial point group. The methods of this paper should yield similar index theorems for the resulting ``twisted families'' of Toeplitz operators parametrised by a 2-torus. One important simplification that occurred in our low-dimension (1D) case is that the edge twisted vector bundles are trivial as complex vector bundles. This need not be the case in 2D and higher, and already in the case of ordinary 2D families of Toeplitz operators, the index \emph{bundle} in $K^0(X)$ is not necessarily determined from data at a single point.

We have also given a glimpse of ``crystallographic T-duality'' in the case of $\sG=\pg$, which was guided by the physical intuition of the bulk-boundary correspondence. In a subsequent work \cite{GT}, we study such new dualities for general space groups, where further interesting phenomena arise.

\section*{Acknowledgements}
G.C.T.~is supported by ARC grant DE170100149, and would like to thank G.~De Nittis, V.~Mathai and K.~Hannabuss for helpful discussions. He also acknowledges H.-H.~Lee for his kind hospitality at the Seoul National University, where part of this work was completed. 
K.G.~is supported by JSPS KAKENHI Grant Number JP15K04871, and thanks I.~Sasaki, M.~Furuta and K.~Shiozaki for useful discussions.

\appendix
\section{Appendix: Computation of $K^{\bullet+\wh{\nu}}_{\ZZ_2}(\cB)$ and $K^{\bullet + c + \widehat{\nu}}_{\Z_2}(\mathcal{B}_x)$}\label{appendix:computations}

\subsection*{Computation of $K^{\bullet+\wh{\nu}}_{\ZZ_2}(\cB)$}

Let $\mathcal{B}_y = \R/2\pi \Z$ be the circle with the involution $k_y \mapsto -k_y$, and $R(\ZZ_2)=\ZZ[t]/(1-t^2)$ be the representation ring of $\ZZ_2$ with $t$ the sign representation. The following result is known, for example, in \cite{MD-R}:

\begin{lem}
We have the following identifications of $R(\Z_2)$-modules
\begin{align*}
K^0_{\Z_2}(\mathcal{B}_y) &\cong 
\overbrace{R(\Z_2)}^{\Z^2} \oplus \overbrace{(1 - t)}^{\Z}, &
K^1_{\Z_2}(\mathcal{B}_y) &= 0.
\end{align*}
\end{lem}

\begin{proof}
To apply the Mayer--Vietoris exact sequence, let us consider the $\Z_2$-invariant subspaces $U$ and $V$ given by
\begin{align*}
U &= \{ k_y \in \mathcal{B}_y |\ -\pi/2 \le k_y \le \pi/2 \}, &
V &= \{ k_y \in \mathcal{B}_y |\ \pi/2 \le k_y \le 3\pi/2 \}.
\end{align*}
These spaces are equivariantly contractible, so that
$$
K^n_{\Z_2}(U) \cong K^n_{\Z_2}(V)
\cong K^n_{\Z_2}(\mathrm{pt}) 
\cong
\left\{
\begin{array}{ll}
R(\Z_2) & (n = 0), \\
0 & (n = 1).
\end{array}
\right.
$$
The intersection $U \cap V$ is the space consisting of two points with free action of $\Z_2$, so that
$$
K^n_{\Z_2}(U \cap V)
\cong K^n_{\Z_2}(\mathrm{pt}\sqcup\mathrm{pt}) \cong K^n(\mathrm{pt}) 
\cong
\left\{
\begin{array}{ll}
\Z & (n = 0), \\
0 & (n = 1).
\end{array}
\right.
$$
Now, the Mayer--Vietoris exact sequence is 
$$
\begin{CD}
\overbrace{K^1_{\Z_2}(U \cap V)}^0 @<<<
\overbrace{K^1_{\Z_2}(U) \oplus K^1_{\Z_2}(V)}^0 @<<<
K^1_{\Z_2}(\mathcal{B}_y) \\
@VVV @. @AAA \\
K^0_{\Z_2}(\mathcal{B}_y) @>>>
\underbrace{K^0_{\Z_2}(U)}_{R(\Z_2)} \oplus 
\underbrace{K^0_{\Z_2}(V)}_{R(\Z_2)} 
@>{\Delta}>>
\underbrace{K^0_{\Z_2}(U \cap V)}_{\Z}.
\end{CD}
$$
The homomorphism $\Delta$ is realised as $\Delta(u, v) = i_U^*u - i_V^*v$, where $i_U^*$ and $i_V^*$ are induced from the inclusions $i_U : U \cap V \to U$ and $i_V : U \cap V \to V$. In the present case, we can identify $i^*_U$ as well as $i^*_V$ with the ``dimension'' $R(\Z_2) \to \Z$ given by $f(t) \mapsto f(1)$. This is surjective, and so is $\Delta$. As a result, we get $K^1_{\Z_2}(\mathcal{B}_y) \cong \mathrm{Coker}(\Delta) = 0$. We also get $K^0_{\Z_2}(\mathcal{B}_x) \cong \mathrm{Ker}(\Delta) \cong \Z^3$. As a basis of this abelian group, we can choose
$$
\{ (1, 1), (t, t), (0, 1 - t) \} \subset R(\Z_2) \oplus R(\Z_2).
$$
The former two base elements generate the $R(\Z_2)$-module $R(\Z_2)$, whereas the latter element $(0, 1 - t)$ generates the $R(\Z_2)$-module $(1 - t)$. 
\end{proof}

Let $\mathcal{B} = \mathcal{B}_x \times \mathcal{B}_y$ be the $2$-dimensional torus $(\R/2\pi \Z) \times (\R/2\pi \Z)$ with the $\Z_2$-action $(k_x, k_y) \mapsto (k_x, -k_y)$, and $\widehat{\nu} \in Z^2(\Z_2, C(\mathcal{B}, U(1)))$ the group $2$-cocycle \eqref{dualcocycle} induced from the wallpaper group \textsf{pg}.

\begin{prop} \label{prop:compute_K_of_torus}
We have the following identifications of $R(\Z_2)$-modules
\begin{align*}
K^{0 + \widehat{\nu}}_{\Z_2}(\mathcal{B}) &\cong 
\overbrace{(1 + t)}^{\Z}, &
K^{1 + \widehat{\nu}}_{\Z_2}(\mathcal{B}) &\cong
\overbrace{(1 + t)}^{\Z} \oplus \,\Z/2.
\end{align*}
where $\Z/2$ is the unique $R(\Z_2)$-module whose underlying abelian group is $\Z/2$.
\end{prop}

\begin{proof}
We cover $\mathcal{B}$ by the $\Z_2$-invariant subspaces
\begin{align*}
U &= 
\{ k_x \in \mathcal{B}_x |\ 
-\pi/2 \le k_x \le \pi/2 \} \times \mathcal{B}_y, \\
V &= 
\{ k_x \in \mathcal{B}_x |\ \pi/2 \le k_x \le 3\pi/2 \} \times \mathcal{B}_y,
\end{align*}
each of which is equivariantly homotopic to $\mathcal{B}_y$. Their intersection is identified with the disjoint union $U \cap V = \mathcal{B}_y ^+ \sqcup \mathcal{B}_y^-$ of two copies $\mathcal{B}_y^{\pm}$ of $\mathcal{B}_y$. The Mayer--Vietoris exact sequence associated to this cover is
$$
\begin{CD}
K^{1 + \widehat{\nu}|_{U \cap V}}_{\Z_2}(U \cap V) @<<<
K^{1 + \widehat{\nu}|_{U}}_{\Z_2}(U) \oplus
K^{1 + \widehat{\nu}|_{V}}_{\Z_2}(V) @<<<
K^{1 + \widehat{\nu}|_{U \cap V}}_{\Z_2}(\mathcal{B}) \\
@VVV @. @AAA \\
K^{0 + \widehat{\nu}}_{\Z_2}(\mathcal{B}) @>>>
K^{0 + \widehat{\nu}|_{U}}_{\Z_2}(U) \oplus 
K^{0 + \widehat{\nu}|_{V}}_{\Z_2}(V) @>{\Delta}>>
K^{0 + \widehat{\nu}|_{U \cap V}}_{\Z_2}(U \cap V),
\end{CD}
$$
where, for example, $\widehat{\nu}|_U \in Z^2(\Z_2, C(U, U(1)))$ is the restriction of $\widehat{\nu}$ to $U \subset \mathcal{B}$, and $\Delta$ is realised as $\Delta(u, v) = i_U^*u - i_V^*v$ by using the inclusions $i_U : U \cap V \to U$ and $i_V : U \cap V \to V$. The restricted cocycles can be trivialised, and a choice of these trivialisations induces the isomorphisms
\begin{align*}
K^{n + \widehat{\nu}|_{U}}_{\Z_2}(U)
&\cong K^{n}_{\Z_2}(U)
\cong K^n_{\Z_2}(\mathcal{B}_y), \\
K^{n + \widehat{\nu}|_{V}}_{\Z_2}(V)
&\cong K^{n}_{\Z_2}(V)
\cong K^n_{\Z_2}(\mathcal{B}_y), \\
K^{n + \widehat{\nu}|_{U \cap V}}_{\Z_2}(U \cap V)
&\cong K^{n}_{\Z_2}(U \cap V)
\cong K^n_{\Z_2}(\mathcal{B}_y^+ \sqcup \mathcal{B}_y^-).
\end{align*}
Notice, however, that there is no (global) trivialisation of $\widehat{\nu}$. At best, we can choose the (local) trivialisations on $U$, $V$ and $U \cap V$ so that the trivialisations on $U$ and $V$ agree with that on $\mathcal{B}_y^+$ and their discrepancy on $\mathcal{B}_y^-$ is the sign representation that is a group $1$-cocycle in $Z^1(\Z_2, U(1))$. This sign representation acts on $K^n_{\Z_2}(\mathcal{B}_y^-)$ as an automorphism, and is realised by multiplication by $t \in R(\Z_2)$. Taking the effects of trivialisations into account, we can identify $\Delta$ as the homomorphism
\begin{align*}
\Delta &: K^0_{\Z_2}(\mathcal{B}_x) \oplus K^0_{\Z_2}(\mathcal{B}_x)
\to K^0_{\Z_2}(\mathcal{B}_x^+) \oplus K^0_{\Z_2}(\mathcal{B}_x^-), &
\Delta(u, v) &= (u - v, u - tv).
\end{align*}
With this expression, the Mayer--Vietoris sequence is folded to
$$
0 \to K^{0 + \widehat{\nu}}_{\Z_2}(\mathcal{B}) \to
K^0_{\Z_2}(\mathcal{B}_y) \overset{\delta}{\to}
K^0_{\Z_2}(\mathcal{B}_y) \to
K^{1 + \widehat{\nu}}_{\Z_2}(\mathcal{B}) \to
0,
$$
in which $\delta(w) = (1 - t)w$. Now, we can readily see that
\begin{align*}
K^{0 + \widehat{\nu}}_{\Z_2}(\mathcal{B})
&\cong (1 + t), \\
K^{1 + \widehat{\nu}}_{\Z_2}(\mathcal{B})
&\cong R(\Z_2)/(1 - t) \oplus (1 - t)/(2 - 2t)
\cong (1 + t) \oplus \Z/2,
\end{align*}
as claimed.
\end{proof}

\subsection*{Computation of $K^{\bullet + c + \widehat{\nu}}_{\Z_2}(\mathcal{B}_x)$}
As in Section \ref{sec:finiterankK}, let $\mathcal{B}_x$ be the circle with trivial involution, $c:\ZZ_2\rightarrow\ZZ_2$ the identity map, and $\wh{\nu}$ the cocycle \eqref{Bxcocycle}.
\begin{prop} 
We have the following identifications of $R(\Z_2)$-modules
\begin{align*}
K^{0 + c + \widehat{\nu}}_{\Z_2}(\mathcal{B}_x) &\cong \Z/2, &
K^{1 + c + \widehat{\nu}}_{\Z_2}(\mathcal{B}_x) &=0
\end{align*}
where $\Z/2$ is the unique $R(\Z_2)$-module whose underlying abelian group is $\Z/2$.
\end{prop}

\begin{proof}
The computation is similar to that in the proof of Proposition \ref{prop:compute_K_of_torus}: we cover $\mathcal{B}_x$ by two closed intervals $U$ and $V$ so that $U \cap V \cong \mathrm{pt} \sqcup \mathrm{pt}$, which are invariant with respect to the trivial $\Z_2$-actions. Trivialising the twists $\widehat{\nu}|_U$, $\widehat{\nu}|_V$ and $\widehat{\nu}|_{U \cap V}$, we have
\begin{align*}
K^{n + c + \widehat{\nu}|_{U}}_{\Z_2}(U)
&\cong K^{n + c}_{\Z_2}(U)
\cong K^{n + c}_{\Z_2}(\mathrm{pt}), \\
K^{n + c + \widehat{\nu}|_{V}}_{\Z_2}(V)
&\cong K^{c + n}_{\Z_2}(V)
\cong K^{c + n}_{\Z_2}(\mathrm{pt}), \\
K^{n + c + \widehat{\nu}|_{U \cap V}}_{\Z_2}(U \cap V)
&\cong K^{c + n}_{\Z_2}(U \cap V)
\cong K^{c + n}_{\Z_2}(\mathrm{pt}) \oplus K^{c + n}_{\Z_2}(\mathrm{pt}).
\end{align*}
The $K$-theory $K^{c + n}_{\Z_2}(\mathrm{pt})$ twisted by the identity homomorphism $c : \Z_2 \to \Z_2$ is isomorphic to $K^n_{\pm}(\mathrm{pt})$ in \cite{W}, and is known to be \cite{Gomi3}
\begin{align*}
K^{c + 0}_{\Z_2}(\mathrm{pt}) 
&= 0, &
K^{c + 1}_{\Z_2}(\mathrm{pt})
\cong \overbrace{(1 - t)}^{\Z}.
\end{align*}
Taking care of the choice of the local trivialisations, we can reduce the Mayer--Vietoris exact sequence to
$$
0 \to 
K^{1 + c + \widehat{\nu}}_{\Z_2}(\mathcal{B}_x) \to
\overbrace{(1 - t)}^{\ZZ} \overset{1 - t}{\to}
\overbrace{(1 - t)}^{\ZZ} \to
K^{0 + c + \widehat{\nu}}_{\Z_2}(\mathcal{B}_x) \to
0,
$$
from which $K^{n + c + \widehat{\nu}}_{\Z_2}(\mathcal{B}_x)$ is determined immediately.
\end{proof}

\section{Appendix: Bulk-edge correspondence of integer indices for $\pg$-symmetric Hamiltonians}\label{appendix:integerbec}
There is another natural surjective index map $K^{1 + \widehat{\nu}}_{\Z_2}(\mathcal{B})=\ZZ\oplus\ZZ/2\rightarrow K^0(\mathcal{B}_y)\cong\ZZ$ which is most conveniently formulated using $C^*$-algebraic language, as briefly mentioned in Section \ref{sec:BaumConnes}. Namely, $K^{1 + \widehat{\nu}}_{\Z_2}(\mathcal{B})$ is also the operator algebraic  $K_1(C^*(\pg))$. Now $\pg\cong\ZZ_y\rtimes\ZZ_g$ where we have written $\ZZ_y\cong\ZZ$ and $\ZZ_g\cong \ZZ$ for the subgroups generated by vertical lattice translation and by glide reflection respectively, and the semidirect product is given by the nontrivial reflection action of $\ZZ_g$ on $\ZZ_y$. We can rewrite 
$$C^*(\pg)\cong C^*(\ZZ_y)\rtimes_\alpha\ZZ_g\cong C(\mathcal{B}_y)\rtimes_\alpha\ZZ_g$$
as a crossed product in which $\ZZ_g$ acts on $C(\mathcal{B}_y)$ by the reflection automorphism $\alpha$ taking $(\alpha\cdot f)(k_y)=f(-k_y)$. The Pimsner--Voiculescu (PV) exact sequence \cite{PV} for this crossed product is
\begin{equation*}
\begin{tikzcd}
  K_1(C(\mathcal{B}_y)) \arrow[r, "1-\alpha_*"] & K_1(C(\mathcal{B}_y)) \arrow[r] & K_1(C^*(\pg)) \arrow[d, "\partial^{(1)}_{\rm PV}"] \\
  K_0(C^*(\pg)) \arrow[u, "\partial^{(0)}_{\rm PV}"'] &  K_0(C(\mathcal{B}_y))  \arrow[l] &  K_0(C(\mathcal{B}_y))  \arrow[l, "1-\alpha_*"]
  \end{tikzcd}
\end{equation*}
If we write $K_1(C(\mathcal{B}_y))\cong\ZZ[u_y]$ with $u_y$ the generating unitary corresponding to vertical translation (i.e.\ winds around the Fourier transformed space $\mathcal{B}_y$ once), and $K_0(C(\mathcal{B}_y))=\ZZ[{\bf 1}]$ with ${\bf 1}$ the generating trivial projection given by the identity, then the PV sequence simplifies to
\begin{equation*}
\begin{tikzcd}
  \ZZ[u_y] \arrow[r, "1-(-1)"] & \ZZ[u_y] \arrow[r, "{(0, \text{mod}\, 2)}"] & \ZZ\oplus\ZZ/2 \arrow[d, "\partial^{(1)}_{\rm PV}"] \\
  \ZZ \arrow[u, "\partial^{(0)}_{\rm PV}"] &  \ZZ[{\bf 1}]  \arrow[l, "\cong"] &  \ZZ[{\bf 1}]  \arrow[l, "1-1=0"]
  \end{tikzcd}
\end{equation*}
in which $\partial^{(0)}_{\rm PV}$ is the zero map while the other $\ZZ$-valued connecting ``index'' map 
$$\partial^{(1)}_{\rm PV}:K^{1 + \widehat{\nu}}_{\Z_2}(\mathcal{B})=K_1(C^*(\pg))\rightarrow K_0(C(\mathcal{B}_y))=K^0(\mathcal{B}_y)$$ is surjective. The free generator of $K_1(C^*(\pg))\cong\ZZ\oplus\ZZ/2$ can be taken to be the unitary corresponding to the generating glide reflection, i.e.\ the $[U_r]$ of Section \ref{sec:tightbindingmodel} specifying the Hamiltonian $H_r$ (this has winding number $-1$ around $\mathcal{B}_x$). 

Thus the index map $\partial^{(1)}_{\rm PV}$ counts the winding around $\mathcal{B}_x$, and in our tight-binding model, we see from Fig.\ \ref{fig:Hrg} that it has the analytic interpretation of counting the number of black zero modes (per unit cell) left behind after truncating to the half-space on the right side of a vertical edge. The PV connecting homomorphism, where available, is an important ingredient in the formulation of bulk-edge correspondences of \emph{weak phases} via generalised Connes--Chern character formulae, as studied in \cite{PSB2}. From that point of view, $[U_r]$ is a ``weak'' topological phase, but as the Klein phase in our $\pg$ example shows, the notions of ``weak'' and ``strong'' topological insulators are somewhat blurred in the presence of crystallographic symmetries.

\end{document}